\documentclass[review, round, plainnat]{elsarticle}
\usepackage{amsmath}
\usepackage{graphicx}
\usepackage{amsthm,amsmath,amssymb}
 
\usepackage{mathrsfs}
\usepackage{caption}
\usepackage[table]{xcolor} 
\usepackage{array} 
\usepackage{subfigure}
\usepackage{lineno,hyperref}
\usepackage{cleveref}
\usepackage{multirow}
\usepackage{geometry}
\usepackage[utf8]{inputenc}
\usepackage{colortbl} 
\usepackage{hhline} 
\usepackage{capt-of}
\usepackage{ulem}
\usepackage{amssymb}
\usepackage{color}
\usepackage[figuresright]{rotating}
\usepackage{comment}
\usepackage{pdflscape}
\usepackage{amsthm}
\usepackage[section]{placeins}
\usepackage{graphicx} 
\usepackage{subcaption} 
\usepackage{dsfont}
\usepackage{capt-of}
\usepackage{natbib}\biboptions{authoryear}   
\modulolinenumbers[5]
\journal{Working Papers}

\newtheorem{assumption}{Assumption}  
\newtheorem{theorem}{Theorem}[section]
\newtheorem{corollary}{Corollary}[theorem]

\newtheorem{lemma}{Lemma}
\newtheorem{proposition}{Proposition}[section]
\newtheorem{corollaryP}{Corollary}[proposition]
\newtheorem{example}{Example}[section]
\usepackage{booktabs} 
\usepackage{siunitx}  
\renewcommand{\geq}{\geqslant}
\renewcommand{\leq}{\leqslant}
\renewcommand{\epsilon}{\varepsilon}
\usepackage{tikz}
\usetikzlibrary{calc}
\usetikzlibrary{positioning}
\usetikzlibrary{arrows}

\tikzset{
	block/.style={rectangle, draw, rounded corners, text centered, text width = 7em, minimum height = 2em},
	line/.style={draw, -latex'}
}

\begin{document}
\begin{frontmatter}

\title{
Equilibrium in Style: A Modeling Framework on the Cash Flow and the Life Cycle of a Consumer Store} 
	
\author[author1]{Shanyu Han}
\address[author1]{School of Mathematical Sciences, Peking University, Beijing, China}
\author[author2]{Jian Lei\corref{cor1}}
\address[author2]{PBC School of Finance, Tsinghua University, Beijing, China}
\cortext[cor1]{Corresponding author}\ead{leij.23@pbcsf.tsinghua.edu.cn}

\author[author3]{Yang Liu}
\address[author3]{School of Science and Engineering, The Chinese University of Hong Kong, Shenzhen, Guangdong, China}

\begin{abstract}
The consumer store is ubiquitous and plays an important role in our everyday lives. It is an open question why stores usually have such short life cycles (typically around 3 years in China). This paper proposes a theoretical framework based on an equilibrium in style supply of stores and style demand of consumers to characterize store cash flow (revenue), leading to a strong explanation of this puzzle. In our model, we derive that the preference shifting of consumers is the main reason for the cash fl ow decreasing to its break-even line over time, while the visibility broadening leads to initial growth, resulting in a rainbow-shaped cash flow and its life cycle. Moreover, the intensified spatial competition will lead to an unexpected decrease in the store's cash flow, or even closure. We calibrate our model with proprietary data of three Chinese stores from three representative industries and study the relationship between customers' preference shifting and cash flow. To our knowledge, there have been no prior attempts to quantitatively model the life cycle of the store.


\end{abstract}


\end{frontmatter}


\section{Introduction}
In the vast forest of the marketplace, conventional listed companies are the towering trees, everlasting and long-living, while consumer stores (micro-enterprises) embody the plentiful small plants which naturally grow and die with the changing seasons. These stores, such as convenience stores, restaurants, and hotels, exist everywhere in our daily life, but little is known about them academically. Especially their frequent opening and closing, leaving a question of why these stores always have short lifespans. This paper develops a framework based on equilibrium in style to model stores' cash flow over time, which can be calibrated by real data, and hence provides a capable explanation for this puzzle.

The focus of examining a store's life cycle is on the average level of cash flow and its changes over time\footnote{In this paper, we do not distinguish the two words: the cash flow and the revenue.}. Due to the necessary variable costs, such as costs of raw materials, rents, and salary, faced by a store, low levels of cash flow can lead to store closures. The generation of cash flow usually requires customers to pass through the front of the store, walk into the store and purchase goods or services (with a certain probability). Our model is built around this observation, taking into account several major factors that may affect a store's cash flow. From a general perspective, we model the \textit{cash flow} of one store via the following three steps, among which we propose novel concepts and assumptions to precisely characterize our framework. 

First, we decompose the cash flow of a store into the product of the \textit{average customer price} (ACP), the \textit{conversion rate}, and the \textit{flow of potential customers}. Conversion rate is an important concept as it characterizes the probability of consumers consuming in a store.
\cite{perdikaki2012effect} define the conversion rate as the ratio of the number of transactions to that of the traffic and study the impact of traffic on cash flow. Furthermore, \textit{population structure} is introduced to classify consumers into several types with different preference (such as based on their age, occupation, and income) with each type of consumers having its own conversion rate for a store. The overall conversion rate is the weighted average of conversion rates of each type of consumers.

Second, we define \textit{style} as the attributes of a store and its products, taking into account both the style demand of consumers and the style supply of store owners. On one hand, we capture the style preference of consumers by using the Cobb-Douglas (CD) utility function in \cite{cobb1928theory} and the Berry-Levinsohn-Pakes (BLP) model in \cite{berry1995automobile}. The latter model was the first to propose the utility of different types of consumers concerning the product style. The model of a time-varying form is later taken into account in \cite{berry2014identification}. An important assumption is that consumers' preferences shift over time. On the other hand, we explore the optimization goals of store owners to derive their optimal strategies. We assume that ACP and style of a given store is stable over time due to the price stickiness and the technical stickiness. But for some stores their styles can be updated over time. As a result, we find a Nash equilibrium, where we derive the mathematical form of the conversion rate for each type of consumers. 

Third, we turn to the flow of potential customers. The modeling of it includes three points: the flow of foot traffic, distance decrease and visibility broadening. The foot traffic is assumed uniformly distributed around the store location with a constant density. Meanwhile, similar to \cite{lucas2002internal}, we assume a store's attractiveness exponentially decreases with distance increasing. Furthermore, few consumers know about a store during its initial opening period and its visibility is broadening by a certain speed from the opening. Additionally, spatial competition is also considered. Intensive spatial competition leads to an unexpected drop in the flow of potential customers. To better illustrate, the debut of a new same-brand store near the initial store often prompts consumer defection. In cases where the local demand cannot sustain both establishments, the closure of one of the stores becomes inevitable.

Finally, we obtain an equation of a store's cash flow over time. In particular, we derive that when consumer preferences are not dispersed, the cash flow of a store over time is a rainbow-shaped curve. The increase in the cash flow curve can be interpreted as an increase in potential customers due to the visibility broadening, while the decrease can be interpreted as a decrease in the conversion rate due to consumers' preferences shifting. We also use simulation method to illustrate some important implications. First, the simulation results indicate that (i) the visibility broadening speed mainly affects the store's upfront cash flow; (ii) the decrease coefficient mainly affects the length of the store's life cycle; and (iii) the foot traffic density mainly affects the macro-scale of cash flow over time. The results show that our model provides a strong explanation for the puzzle. As a conclusion, we provide a road map for our modeling framework in Figure \ref{motivating-example}.

\begin{figure}[htp]		\small
	\begin{center}
		\begin{tikzpicture}[every node/.style={block}]
			\node (topic) {{\bf Cash flow }$\text{CF}_t$};
			\node[below left=1cm and 3cm of topic](case one) {\bf Average customer price $\theta_t$};
			\node[below = 1cm of topic](case three) {Preferences};
            \node[below right = 1cm and 3cm of topic](case two){\bf flow of potential customers $N_t$};
            \node[below = 8.15cm of case one](PPri){Product price $\boldsymbol{\theta}^*(\boldsymbol{z})$};
   			\node[below = 3cm of case two](PFD) {Foot traffic density $u$};
			\node[below right = 1cm and 0.01cm of case two](CF) {Attraction $\boldsymbol{q}_t$};
       		\node[right = 0.01cm of CF](SCom){Spatial competition};
            \node[right = 0.01cm of PFD](VS){Visibility broadening $k$};
            \node[right = 0.01cm of VS](AD){Attractiveness decrease $e^{-\delta\Vert\mathbf{x}\Vert}$};
			\node[below left = 1cm and 0.01cm of case three](CPS){Consumer population structure $\mathbf{P}$};
			\node[below = 0.9cm of CPS](SU){Score $F$, Utility $U$};
            \node[below right = 1cm and 0.01cm of case three](Store){Initial investment $I$};
			\node[below = 1.5cm of Store](IIS) { \textbf{Store style:} Product and storefront attributes $\mathbf{z} =(\mathbf{x}, \boldsymbol{\xi})$ };
			\node[below left = 1cm and 0.01cm of IIS](Eq){\bf Equilibrium};
			\node[below = 1cm of Eq](PP){Purchase probability density $\rho^*$};
			\node[below = 3.3cm of IIS](PSFA){Style set $\mathbb{Z}^*_t$};
			\node[below = 4cm of Eq](LC){\bf Consumer conversion rate $\tilde{\beta}_t$};
			\node[below = 6cm of Eq](CCR){\bf Decrease $\nu$};
			\node[below = 1cm of CCR](PS){\bf Life cycle $T$};
			
			\foreach \x/\y in {
				topic/case one,
				topic/case two,
				topic/case three,
				case one/PPri,
				case two/PFD,
				case two/CF,
				case two/SCom,
				case three/CPS,
				case three/Store,
                CF/VS,
                CF/AD,
				CPS/SU,
                Store/IIS,
                SU/Eq,
                IIS/Eq,
                Eq/PPri,
                Eq/PP,
                Eq/PSFA,
                PPri/LC,
                PP/LC,
                PSFA/LC,
                LC/CCR,
                CCR/PS}
			\draw [line] (\x) -- (\y);
		\end{tikzpicture}
	\end{center}
	\caption{A stylized roadmap for the modeling of a consumer store}
	\label{motivating-example}
\end{figure}
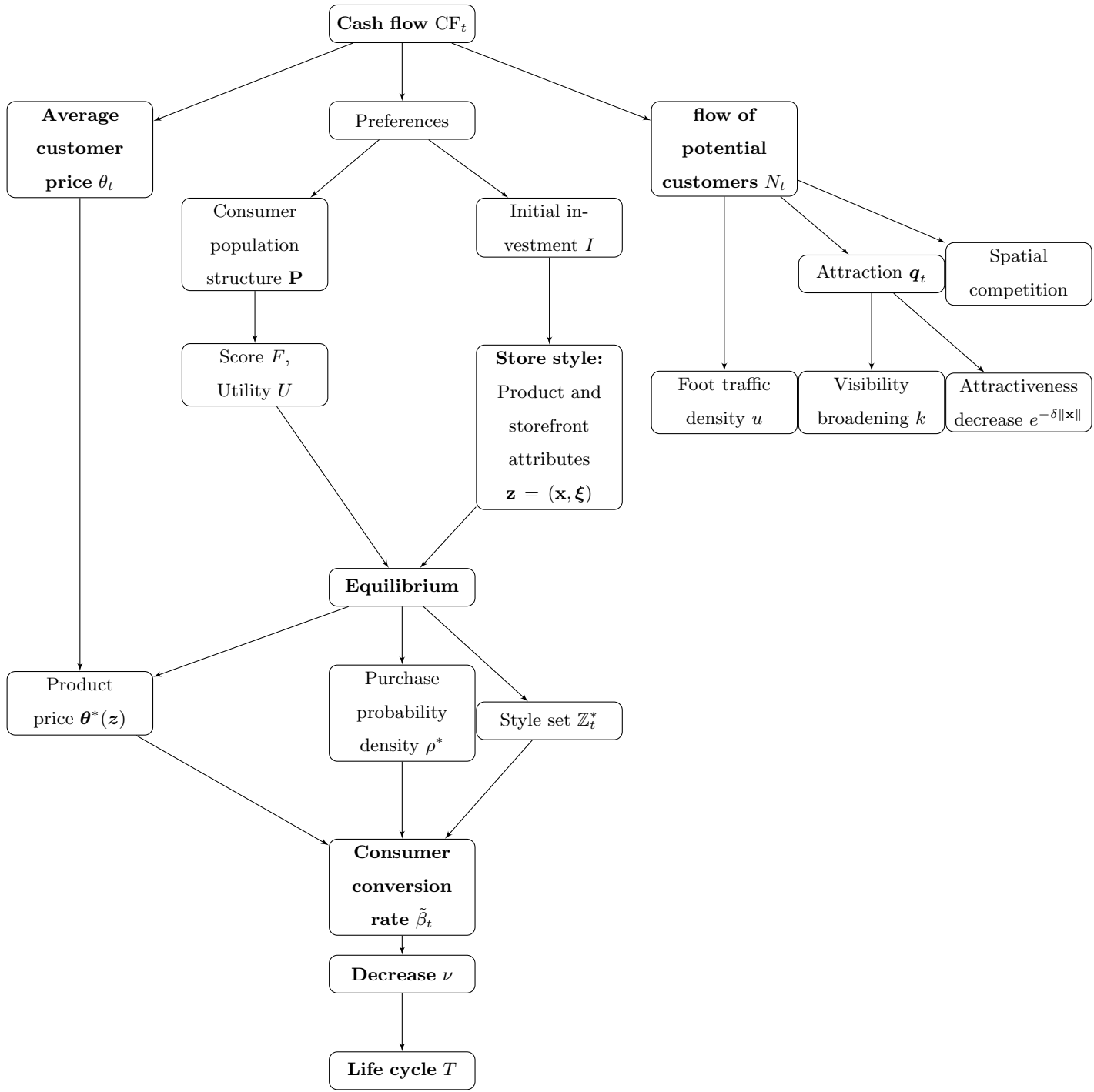

In our empirical study, we use the real data from three Chinese stores from three industries in the real economy: the retail industry, the restaurant industry, and the service industry. We conduct a Non-linear Least Squares regression (NLS) of the adjusted cash flow on the time scale to calibrate the visibility broadening speed, the initial conversion rate, and the decrease coefficient. The regression results uncover that the estimates are all statistically significant at a 1\% confidence level. Finally, we plot the real data and the fitted data to show that our model well captures the trend of the store's cash flow, which strongly justifies the explanation of our model.


Our model incorporates several key ingredients of interest in literature, including product specialization, customers' preferences, and purchase conversion rate. 
\cite{burdett1993equilibrium} establish the equilibrium for dispersed prices under certain cases. 
\cite{menzio2023optimal} discusses the relationship between product specialization to customers' preferences, search friction decrease, and price dispersion. \cite{menzio2023optimal} finds that the increase in product specialization helps sellers cater to customers' preference shifting, which leads to a higher store profit and also a higher customer surplus. The search friction can be seen as the visibility broadening in our paper. \cite{kashyap1995sticky} uses real data to illustrate the effects of the size, frequency, and synchronization of price changes for twelve selected retail goods. 
In the field of marketing science, \cite{lam2001evaluating} propose a framework to carefully classify the conversion rate into three effects of attraction (store-entry), conversion (purchase decision), and spending (purchase amount). They formulate a joint model of four simultaneous equations using front traffic, store traffic, number of store transactions, and store sales as the endogenous variables.
\cite{zimmermann2023developing} study a digital retailer and adopt a conversion rate optimization framework to increase sales.

The paper is organized as follows. Section \ref{Model Environments} proposes the model environment of our paper. The demand of consumers and the supply of stores are well defined in Section \ref{Part I: Style}. Section \ref{Part II: Equilibrium} defines the concept of preference shifting and derives a Nash equilibrium. The spatial competition is described and we close the model in Section \ref{Part III: Potential customers}. The full model is presented in Section \ref{Part IV: Full model of cash flow}. We provide empirical evidence in Section \ref{Empirical Evidence}. Section \ref{Additional Discussions} provides further discussions. Section \ref{Conclusion} concludes.

\section{Model Environments}
\label{Model Environments}

\textbf{Store.} We study the typical consumer store in this paper. A (consumer) store is a type of micro-enterprise, usually possessing the following features: (i) Initial monetary investment is small. For example, it costs the owner around 400,000 Chinese yuan to invest in a \textit{luckin coffee} store in Beijing, China. (ii) Routine expenditures include rent and salary, which form a surprisingly stable ratio to daily income. For instance, the ratio of daily rent to daily cash flow is around 6.17\% for a typical Chinese convenience store in 2022, according to the annual report of KPMG and China Chain-Store \& Franchise Association (CCFA). (iii) The store has a high-frequency cash flow. For example, a \textit{Family Mart} store typically has 1,000 orders (purchases) per day. (iv) The store caters to people's daily demands, which encompasses the restaurant, retail, and service industries. (v) Most consumers of stores purchase goods and services in-store, especially the stores in the service industry. Therefore, the location of the store impacts their financial performance. (vi) There exists perfect competition (over 100 million stores in China) in each industry. In short, our focus is not on the conventional company, but on the store, which is usually the following three types: direct-sale store, franchised outlet, and 'mom-and-pop store'.

\textbf{Cash flow decomposition.} In this paper, we fix one store on time set $\mathcal{T}.$ The cash flow $\text{CF}$ is a non-negative function on $\mathcal{T}.$ 
Then, we model the cash flow at time $t$ by a product of the average customer price (ACP) $\theta_t$, the flow of potential customers $N_t$ and the conversion rate $\beta_t$, which can be formulated as 
\begin{equation}
	\text{CF}_t = \theta_t\cdot N_t \cdot \beta_t.
\end{equation}
This decomposition of the store's cash flow is inspired by \cite{lam2001evaluating} and \cite{perdikaki2012effect}. 

In some cases, the customer flow of a store is bounded by actual conditions (usually in the form of $N_t \cdot \beta_t\leq n_{\text{max}}$). To eliminate this effect, the theoretical conversion rate, denoted as $\tilde{\beta}_t,$ is defined as the conversion rate without this actual limitation.
 We model for $\tilde{\beta}_t$ instead of ${\beta}_t,$ due to the fact that the theoretical conversion rate can be understood as the probability of consumers purchasing products. We do not consider the situation where the flow of customers reaches an upper bound, i.e. letting
\begin{equation}
	\tilde{\beta}_t = {\beta}_t.
\end{equation}

\textbf{Attributes and Styles.} Consider an economy corresponding to a consumer industry on time set $\mathcal{T}=\mathbb{R}.$ At any time $t\in\mathcal{T}$, there are new stores opened with different initial investments. Let the set of all initial investments be $\mathcal{I}\subseteq\mathbb{R}^+,$ which is a bounded set that does not change over time. Therefore, stores are infinitely numerous in this economy. The style of a store is the combination of its products and storefront. We assume that a store provides one representative product (commodity), characterized by a constant price $\theta$ and product attributes $\boldsymbol{x}\in\mathbb{R}^p$, which reflect the features of product size, taste, shape, etc. Similar assumption is mentioned or implied in \cite{chamberlin1938theory}, \cite{dixit1977monopolistic} and \cite{romer1986increasing}. Moreover, we denote the storefront attributes by $\boldsymbol{\xi}\in\mathbb{R}^q$, which reflect the implicit influencing factors of store service, decoration, reputation, etc. We define the style vector for a store $\boldsymbol{z} = (\boldsymbol{x}',\boldsymbol{\xi}')',$ and $\mathbb{R}^{p+q}$ is the entire style space. The bounded set of all styles provided in the industry at time $t$ is denoted as $\mathbb{Z}_t.$ It depends on the styles of stores opened before time $t$ and their lifespans. 

There are $m$ types of consumers who have different preferences for style. Their population structure is notated as $\mathbf{P}=\left(\mathbf{P}^{(1)},\mathbf{P}^{(2)},\ldots,\mathbf{P}^{(m)}\right)$ with $\mathbf{P}^{(j)}$ representing the proportion of the $j$th type of  consumers. Consumers will purchase products in different styles of stores with different probabilities. 
A consumer of type $j$th will assign a bounded continuous purchase probability density on $\mathbb{Z}_t$ at time $t$ and we denote this density function as $\rho_{jt}({\boldsymbol{z}}).$ We show in \ref{append1} the details about the definition of $\rho_{jt}$ and the integrals with respect to it. Furthermore, we assume that styles with small differences may be indistinguishable. $\epsilon$ represents the scale of indistinguishable difference. $\Gamma_{\epsilon}(\mathbb{Z}_t)$ is a $\epsilon$-equipartition of $\mathbb{Z}_t,$ which should satisfy that 
\begin{equation}
    \bigcup_{\Delta \in \Gamma_{\epsilon}(\mathbb{Z}_t)} \Delta = \mathbb{Z}_t,\ \ \bigcap_{\Delta \in \Gamma_{\epsilon}(\mathbb{Z}_t)} \Delta = \Phi, \end{equation}
\begin{equation}
    \sup\{\Vert\boldsymbol{z}_1-\boldsymbol{z}_2\Vert:\boldsymbol{z}_1,\boldsymbol{z}_2\in\Delta\}\leq \epsilon,\ \forall \Delta \in \Gamma_{\epsilon}(\mathbb{Z}_t),
\end{equation}
and $|\Delta|
$ is constant for all $\Delta\in \Gamma_{\epsilon}(\mathbb{Z}_t),$
where $\Phi$ is the empty set. 
Denote the subset in which style $\boldsymbol{z}$ is located as $\Delta_t(\boldsymbol{z},\epsilon) \in \Gamma_{\epsilon}(\mathbb{Z}_t).$ 
Given style ${\boldsymbol{z}}$ without analogs, a $j$th type consumer will purchase with a probability of $
\tilde{\beta}_{jt}=\int_{\Delta_t(\boldsymbol{z},\epsilon)} \rho_{jt}({\boldsymbol{z}_1})\mathrm{d}{\boldsymbol{z}_1}.
$
Thus, the theoretical conversion rate for a store with style $\boldsymbol{z}$ at time $t$ can be calculated as
\begin{equation}
\tilde{\beta}_{t} = 
\sum_{j=1}^m\mathbf{P}^{(m)}\int_{\Delta_t(\boldsymbol{z},\epsilon)} \rho_{jt}({\boldsymbol{z}_1})\mathrm{d}{\boldsymbol{z}_1} =
	\int_{\Delta_t(\boldsymbol{z},\epsilon)} \sum_{j=1}^m\mathbf{P}^{(m)}\rho_{jt}({\boldsymbol{z}_1})\mathrm{d}{\boldsymbol{z}_1}.
\end{equation}
  If the styles of two stores satisfy $\Delta_t(\boldsymbol{z}_1,\epsilon) = \Delta_t(\boldsymbol{z}_2,\epsilon)$, then $\boldsymbol{z}_1$ and $\boldsymbol{z}_2$ are similar styles and a given consumer can only be a potential customer of one of the stores. There is spatial competition between these two similarly styled stores, and this is explored in more depth in Section \ref{Part III: Potential customers} and modeled as a reduction in potential customers.
 \begin{figure}[h]
	\centering
	\includegraphics[width = 0.75\textwidth]{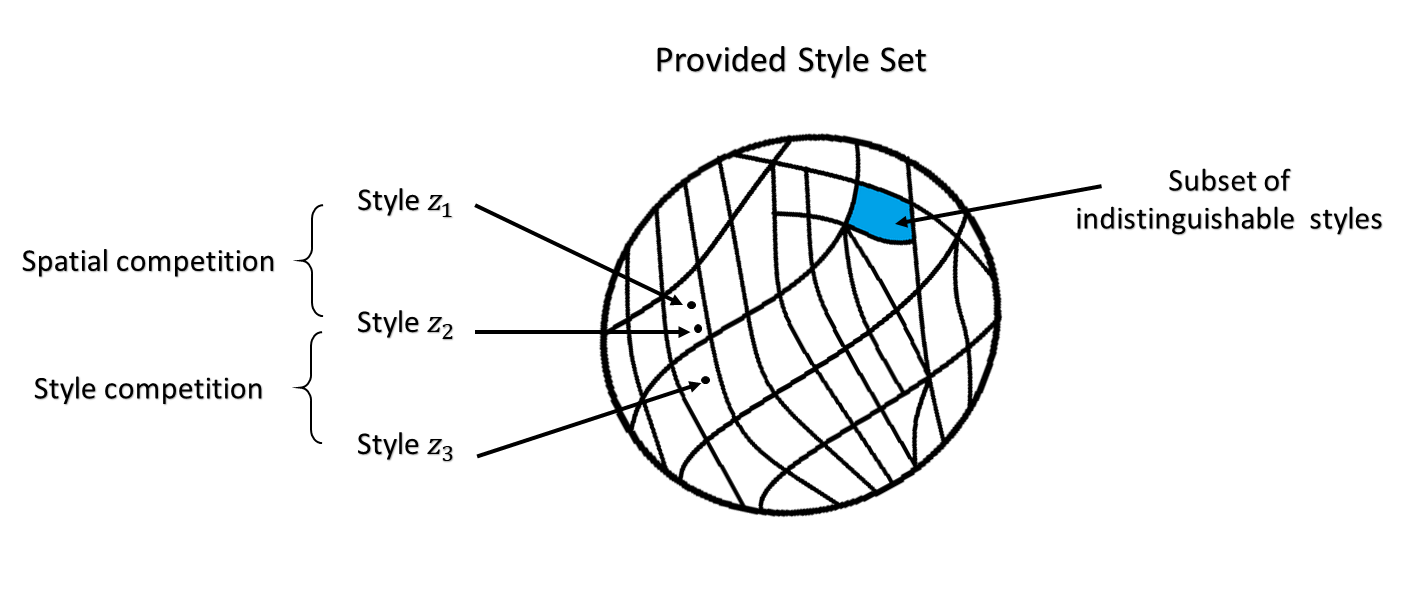}
	\caption{Equipartition on Style Space}   
	\label{stylespace}
\end{figure}
 
\section{Part I: Store style}
\label{Part I: Style}
\subsection{Style demand}
First, we are going to discuss the utility function, the strategy, and the demand of consumers. We start by defining the scoring function of type $j$ consumers for style. 
Similar to \cite{berry2014identification}, we use a function of consumer type, time, and attributes but remove the randomness, i.e. the $f$ function below is deterministic:
\begin{align}
	F_{jt}(\boldsymbol{z},\theta) &= f(\boldsymbol{\zeta}_{jt},\boldsymbol{x},\boldsymbol{\xi},\theta) \\
	\label{scoring}
	&=\phi\left(\boldsymbol{a}_j'\cdot(\boldsymbol{x}-\bar{\boldsymbol{x}}_t )+\boldsymbol{b}_j'\cdot(\boldsymbol{\xi}-\bar{\boldsymbol{\xi}}_t) - \frac{\lambda_j}{2}\left(\boldsymbol{x}'-\bar{\boldsymbol{x}}_t',\boldsymbol{\xi}'-\bar{\boldsymbol{\xi}}_t'\right)\cdot\left(\begin{array}{c}
		\boldsymbol{x}-\bar{\boldsymbol{x}}_t  \\
		\boldsymbol{\xi}-\bar{\boldsymbol{\xi}}_t 
	\end{array}\right)\right)\cdot e^{-\gamma_j\theta},
\end{align}
where $\theta$ is the average price required to purchase in the store and $\boldsymbol{\zeta}_{jt}=(\boldsymbol{a}_j,\boldsymbol{b}_j,\lambda_j,\gamma_j,\bar{\boldsymbol{x}}_t,\bar{\boldsymbol{\xi}}_t)$ is a fixed vector. $(\bar{\boldsymbol{x}}_t,\bar{\boldsymbol{\xi}}_t)$ is a certain traditional preference style and $\lambda_j,\gamma_j$ are the aversion coefficients. We call $\phi$ the transform function which is a smooth function on 
$\mathbb{U} = \left(-\infty,C\right]$ for some sufficiently large $C>0$ with
\begin{itemize}
	\item[(1)] $\phi(u)>0$ for $u \in \mathbb{U},$ 
	\item[(2)] $\phi'(u)>0$ for $u \in \mathbb{U},$ 
	\item[(3)]  $\int_{\mathbb{U}}\phi(u)\mathrm{d}u<\infty.$ 
\end{itemize}
The transform function maps elements in $\mathbb{U}$ to be positive and is strictly order-preserving. Some typical examples of transform functions are $e^{\kappa u}(\kappa>0)$ and $\ (C+C_{\epsilon} - u)^{-\alpha}(\alpha>1,C_{\epsilon}>0)$. (\ref{scoring}) implies that the scoring function is a linear combination of style minus a penalty for deviation from tradition, transformed and multiplied by the attenuation of price. The following proposition illustrates that each type of consumers has an optimal preference style and that scoring decreases with distance from the optimal preference style.

\begin{proposition}
	At time $t,$ there exists a unique optimal preference style for the $j$-th type of consumer\begin{equation}
		\label{z}
		\hat{\boldsymbol{z}}_{j,t} =\left(\begin{array}{c} \hat{\boldsymbol{x}}_{j,t}\\
			\hat{\boldsymbol{\xi}}_{j,t}
		\end{array}\right) =\left(\begin{array}{c}\frac{1}{\lambda_j} \boldsymbol{a}_j + \bar{\boldsymbol{x}}_t \\
			\frac{1}{\lambda_j} \boldsymbol{b}_j + \bar{\boldsymbol{\xi}}_t
		\end{array}\right),
	\end{equation} such that $F_{jt}(\hat{\boldsymbol{z}}_{j,t},\theta)\geq F_{jt}(\boldsymbol{z},\theta)$ for any ${\boldsymbol{z}}\in\mathbb{R}^{p+q}$ and $\theta >0.$ The scoring function can be rewritten as
	\begin{equation}
		\label{Score_simple}
		F_{jt}(\boldsymbol{z},\theta) = \phi\left(C_j - \frac{\lambda_j}{2}\left\Vert\boldsymbol{z}-\hat{\boldsymbol{z}}_{j,t} \right\Vert^2\right) \cdot e^{- \gamma_j \theta}.
	\end{equation}where $C_j$ is a constant of $t.$
\end{proposition}
\begin{proof}
    Notice that
    \begin{equation}
        F_{jt}(\boldsymbol{z},\theta) = \phi\left( \frac{1}{2\lambda_j}\boldsymbol{a}_j'\boldsymbol{a}_j + \frac{1}{2\lambda_j}\boldsymbol{b}_j'\boldsymbol{b}_j- \frac{\lambda_j}{2}\left\Vert\boldsymbol{z}-\hat{\boldsymbol{z}}_{j,t} \right\Vert^2\right) \cdot e^{- \gamma_j \theta}.
    \end{equation}
    Let $C_j=\frac{1}{2\lambda_j}\boldsymbol{a}_j'\boldsymbol{a}_j + \frac{1}{2\lambda_j}\boldsymbol{b}_j'\boldsymbol{b}_j$ and we have (\ref{Score_simple}). Thus, for any $\boldsymbol{z} \in \mathbb{R}^{p+q}$ we have
    \begin{equation}
        F_{jt}({\boldsymbol{z}},\theta) \leq\phi( C_j) e^{-\gamma_j \theta} = F_{jt}(\hat{{\boldsymbol{z}}}_{jt},\theta).
    \end{equation}
\end{proof}
Then we turn to the utility functions. The continuous form of the Cobb-Douglas utility function used in this paper is inspired by the pioneering analysis of factor combination efficiency by \cite{cobb1928theory}.
\begin{assumption}[Cobb-Douglas utility]
	\label{uf}
	The consumers have utility functions formulated as 
	\begin{align}
		U_{jt}(\rho)&= \int_{\mathbb{Z}_t} \ln\left(\rho({\boldsymbol{z}})^{F_{jt}({\boldsymbol{z}},\theta({\boldsymbol{z}}))} \right) \mathrm{d} {\boldsymbol{z}}\\
		&= \int_{\mathbb{Z}_t} F_{jt}({\boldsymbol{z}},\theta({\boldsymbol{z}}))\ln(\rho({\boldsymbol{z}})) \mathrm{d} {\boldsymbol{z}},
	\end{align}
	where $\rho$ is a purchase probability density.
\end{assumption}
Under Assumption \ref{uf}, The $j$-th type consumers' optimization goal is
\begin{align}
	\label{opt1}
	\max_{\rho} \quad &\int_{\mathbb{Z}_t} F_{jt}({\boldsymbol{z}},\theta({\boldsymbol{z}}))\ln(\rho({\boldsymbol{z}})) \mathrm{d} {\boldsymbol{z}}\\
	&\text{s.t.} \int_{\mathbb{Z}_t} \rho({\boldsymbol{z}}) \mathrm{d} {\boldsymbol{z}} = 1.\notag
\end{align}
We solve Problem \eqref{opt1} and obtain each type of consumers' demand, that is, the optimal purchase probability density given $\mathbb{Z}_t$. We denote it as $\rho^{(d)}_{jt}({\boldsymbol{z}})$ and we get 
\begin{equation}
	\rho^{(d)}_{jt}({\boldsymbol{z}}) = \frac{F_{jt}({\boldsymbol{z}},\theta({\boldsymbol{z}}))}{\int_{\mathbb{Z}_t} {F_{jt}({\boldsymbol{z}_1},\theta({\boldsymbol{z}}_1))} \mathrm{d} {\boldsymbol{z}}_1}\cdot\mathds{1}_{\{\boldsymbol{z}\in \mathbb{Z}_t\}},
\end{equation}
almost everywhere, where the convergence of the integral is guaranteed by the boundedness of $\mathbb{Z}_t.$ We can see that the purchase probability density is proportional to the scoring function, 
and so we can write it as 
\begin{equation}
    \label{interest}
	\rho^{(d)}_{jt}({\boldsymbol{z}}) = K_{jt} F_{jt}({\boldsymbol{z}},\theta({\boldsymbol{z}}))\cdot\mathds{1}_{\{\boldsymbol{z}\in \mathbb{Z}_t\}},
\end{equation}
where $K_{jt}$ is the level of purchase probability. In order to clarify some properties that will be mentioned later, we introduce several families of the generalized probability density of the form
\begin{equation}
    \mathscr{P}_{jt} = \left\{\rho(\boldsymbol{z}|\theta)=K\cdot \phi\left(C_j - \frac{\lambda_j}{2}\left\Vert\boldsymbol{z}-\hat{\boldsymbol{z}}_{j,t} \right\Vert^2\right)\cdot e^{ -\gamma_j \theta}:K>0\right\}.
\end{equation}
The elements in $\mathscr{P}_{jt}$ are functions on $\mathbb{R}^{p+q}$. Since they may not be normalized, they are not necessarily true probability density functions. An element in $\mathscr{P}_{jt}$ is a certain multiple of a purchase probability density on the entire style space. Later we will occasionally use functions of this form posing as purchase probability densities to illustrate some properties. However, it will be strictly guaranteed that the purchase probability densities are of the form (\ref{interest}) and are normalized in the equilibrium in style we finally derive.

\subsection{Style supply} Next, we will discuss the strategy and the supply of store owners. In our model, the store owner opens at time $t$ and needs to decide on the style ${\boldsymbol{z}}$ and the price $\theta,$ which cannot be changed afterward. \cite{kashyap1995sticky} shows that the change in retail prices is very slow. \cite{bils2004some} provides evidence that the prices of services like haircuts for males and beauty parlors for females change only every 20 months or more. So we first assume that price and style cannot be changed after opening for most of the stores. In Section \ref{chasing} we will consider situations where the style and price of particular stores can be changed over time. Moreover, storefront attributes are related to the initial investment and we consider the initial investment as exogenous variables. The store owner needs to satisfy constraint $g_I(\boldsymbol{\xi})\leq I,$ where $I$ is a fixed initial investment for a given store owner. $g_I(\cdot)$ is a smooth convex cost function. Denote that $G_I(\boldsymbol{z}) =G_I\left(\boldsymbol{x},\boldsymbol{\xi}\right) =\boldsymbol{x}'\cdot\boldsymbol{0}+g_I(\boldsymbol{\xi})$ and the constraint is also denoted as $G_I(\boldsymbol{z})\leq I$. Furthermore, we simply assume that this constraint is tight for all types of consumers, i.e. $g_I\left(\hat{\boldsymbol{\xi}}_{j,t}\right)>I,$ for $j=1,2,\ldots,m,$ $t>0$ and $I\in\mathcal{I}.$ To maximize the cash flow at time $t,$ the store owner's optimization goal is written as,
\begin{align}
	\max_{{\boldsymbol{z}},\theta} \quad &\theta\sum_{j=1}^m \mathbf{P}^{(j)}\int_{\Delta_t(\boldsymbol{z},\epsilon)} \rho_{jt}({\boldsymbol{z}_1}|\theta)\mathrm{d}{\boldsymbol{z}_1}\\
	&\text{s.t.\quad} G_I(\boldsymbol{z})\leq I.\notag
\end{align}
For the sake of simplicity, we assume $\epsilon$ is small and the partition is unobservable so that the store owner focuses on the consumers' purchase density instead of the purchase probability, but to exclude those isolated style points. The optimization goal is 
\begin{align}
	\label{opt0}
	\max_{{\boldsymbol{z}},\theta} \quad &\theta\sum_{j=1}^m \mathbf{P}^{(j)} \rho_{jt}({\boldsymbol{z}}|\theta)\\
	&\text{s.t.\quad} G_I(\boldsymbol{z})\leq I, \notag \\
   & \qquad \ (\boldsymbol{z},\theta) \in C_t \notag,
\end{align}
where 
\begin{equation}
    C_t = \left\{(\boldsymbol{z},\theta):\exists\delta_1>0,\sum_{j=1}^m \mathbf{P}^{(j)}\int_{\Vert\boldsymbol{z}_1-\boldsymbol{z}\Vert\leq \delta_1} \rho_{jt}({\boldsymbol{z}_1}|\theta)\mathrm{d}\boldsymbol{z}_1>0\right\}
\end{equation}
is the set of non-isolated points of the probability density domain. This implies that only styles that form a scale will lead to effective purchasing behavior among consumers, otherwise the purchase probability corresponding to one or several isolated styles is $0.$ When the density is non-zero on the full space, for example $\rho_{jt} \in \mathscr{P}_{jt}\ (j=1,2,\ldots,m),$ then $C_t = \mathbb{R}^{p+q}\times\mathbb{R}^+$ and the second constraint can be removed.
We call the objective in Problem \eqref{opt0} the cash flow density.
If the solution of Problem \eqref{opt0} exists and is unique, we denote it as $\left(\boldsymbol{z}_t^{(s)}(I),\theta\left(\boldsymbol{z}_t^{(s)}(I)\right)\right).$ When the cash flow density drops to a certain threshold $r(I),$ the store will shut down. Given $I,$ we define the longest lifespan of a store that exists at time $t$ as
\begin{equation}
	T(t,I) = \sup\left\{\tau>0\bigg| \theta\left({\boldsymbol{z}}_{t-\tau}^{(s)}(I)\right)\sum_{j=1}^m\mathbf{P}^{(j)}\rho_{jt}\left({\boldsymbol{z}}_{t-\tau}^{(s)}(I)\right)\geq r(I)\right\}.
\end{equation}
Thus, the provided style set at time $t$ is
\begin{equation}
	\mathbb{Z}^{(s)}_t = \left\{\boldsymbol{z}_{\tau}^{(s)}(I)| I\in \mathcal{I}, t-T(t,I)\leq\tau\leq t \right\}.
\end{equation}

\subsection{Optimal strategy uniqueness for supply}
\label{unqiue}
Problem \eqref{opt0} can be a non-convex optimization problem, so the solution to this optimization equation is not necessarily unique and even one style may correspond to multiple optimal prices. It is natural to assume that for a given style, only the infimum of the optimal prices will occur, otherwise, consumers would be perfectly capable of choosing a lower-priced store of the same style. Thus, in solving Problem \eqref{opt0}, we obtain the set of optimal styles 
\begin{equation}
	\mathscr{Z}_t(I) = \left\{{\boldsymbol{z}}_{t,i}^{(s)}(I):i\in \Lambda_t(I)\right\},
\end{equation}
and corresponding prices 
\begin{equation}
	\Theta_t(I) = \left\{\theta\left({\boldsymbol{z}}\right):{\boldsymbol{z}} \in \mathscr{Z}_t(I)\right\},
\end{equation}
of the store that opens at time $t$ with initial investment $I,$ where $\Lambda_t(I)$ is the set of indicators. These can be considered the supply of the store owner. 

\begin{example}
\label{ex1}
	Consider an illustrative example in the restaurant industry, where there are young and old people in an economy, $\frac13$ and $\frac23$ of the population, respectively. At time $0,$ the young like $3$g of salt in their food best, while the old like $1$g of salt best, and they both like it best when there are $5$ lights at the entrance of the restaurant, which costs \$$10$ each. Naturally, their generalized purchase probability densities are \begin{equation}
	    \rho_{10}(x,\xi|\theta) = K_{10}\exp\left(1-(x-3)^2-\left(\xi-5\right)^2-\theta\right),
	\end{equation} and \begin{equation}
	    \rho_{20}(x,\xi|\theta) = K_{20}\exp\left(1-(x-1)^2-\left(\xi-5\right)^2-\theta\right),
	\end{equation} respectively, where we use $\exp(\cdot)$ as the transform function.
 A new store would require an initial investment of \$$20.$ Obviously, the optimal strategy for the store includes $\xi_0^{(s)}(20) = 2$ and $\theta\left(\boldsymbol{z}_0^{(s)}(20)\right)=1.$ However, 
 \begin{equation}
     x_0^{(s)}(20) = \arg\max \left\{\frac13 K_{10}\exp\left(-(x-3)^2-9\right) + \frac23 K_{10}\exp\left(-(x-1)^2-9\right)\right\}
 \end{equation}
 is related to $K_{10}$ and $K_{20}.$ Figure \ref{eq2} shows the cash flow density of the store with respect to the amount of salt at different ratios for $K_{10}$ and $K_{20}$, respectively. 
\begin{figure}[h]

	\centering
	\begin{minipage}[b]{0.3\textwidth}
		\centering
		\includegraphics[width=\textwidth]{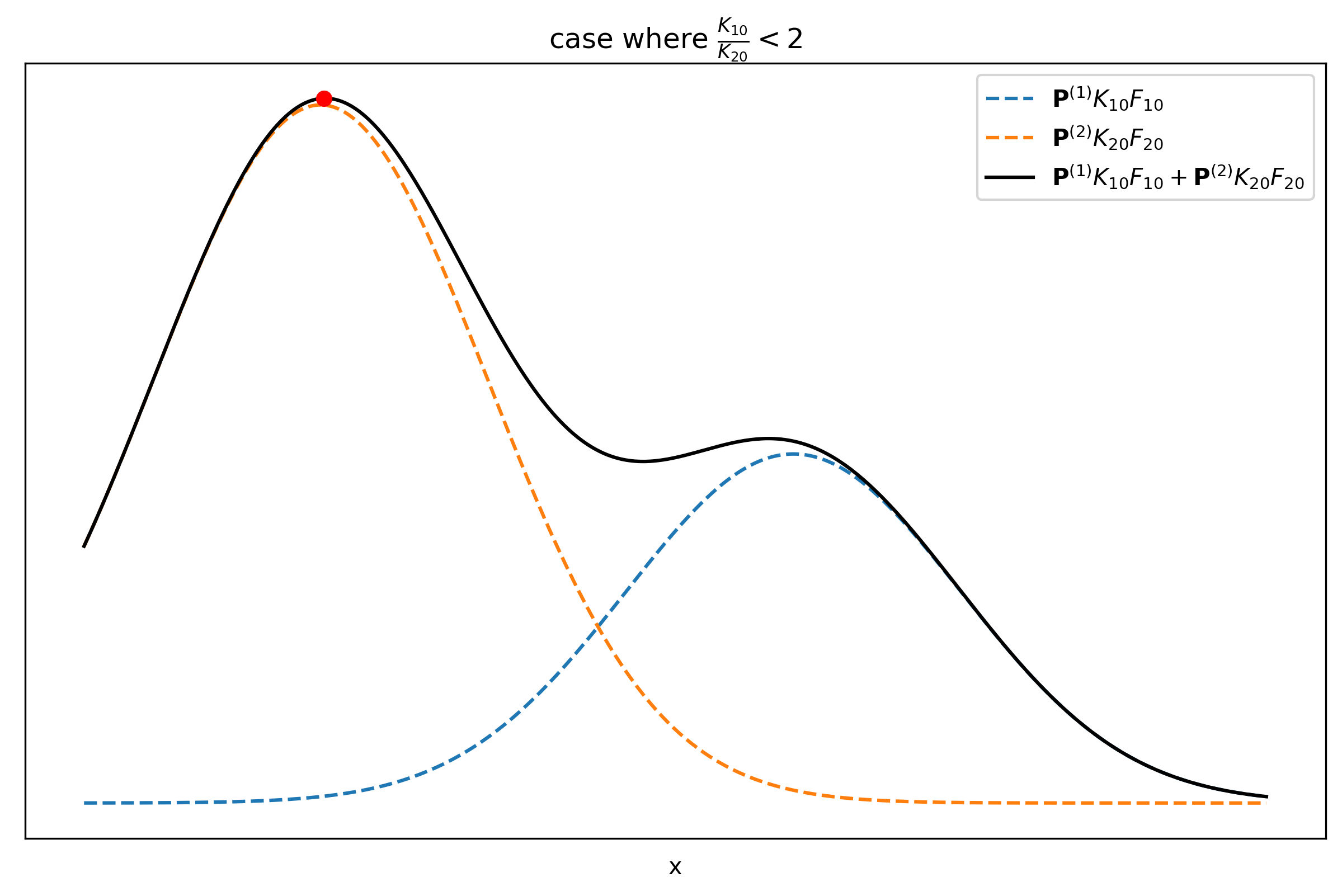}
		\raisebox{1.5ex}{(1)}
		\label{fig:image1}
	\end{minipage}
	\hfill 
	\begin{minipage}[b]{0.3\textwidth}
		\centering
		\includegraphics[width=\textwidth]{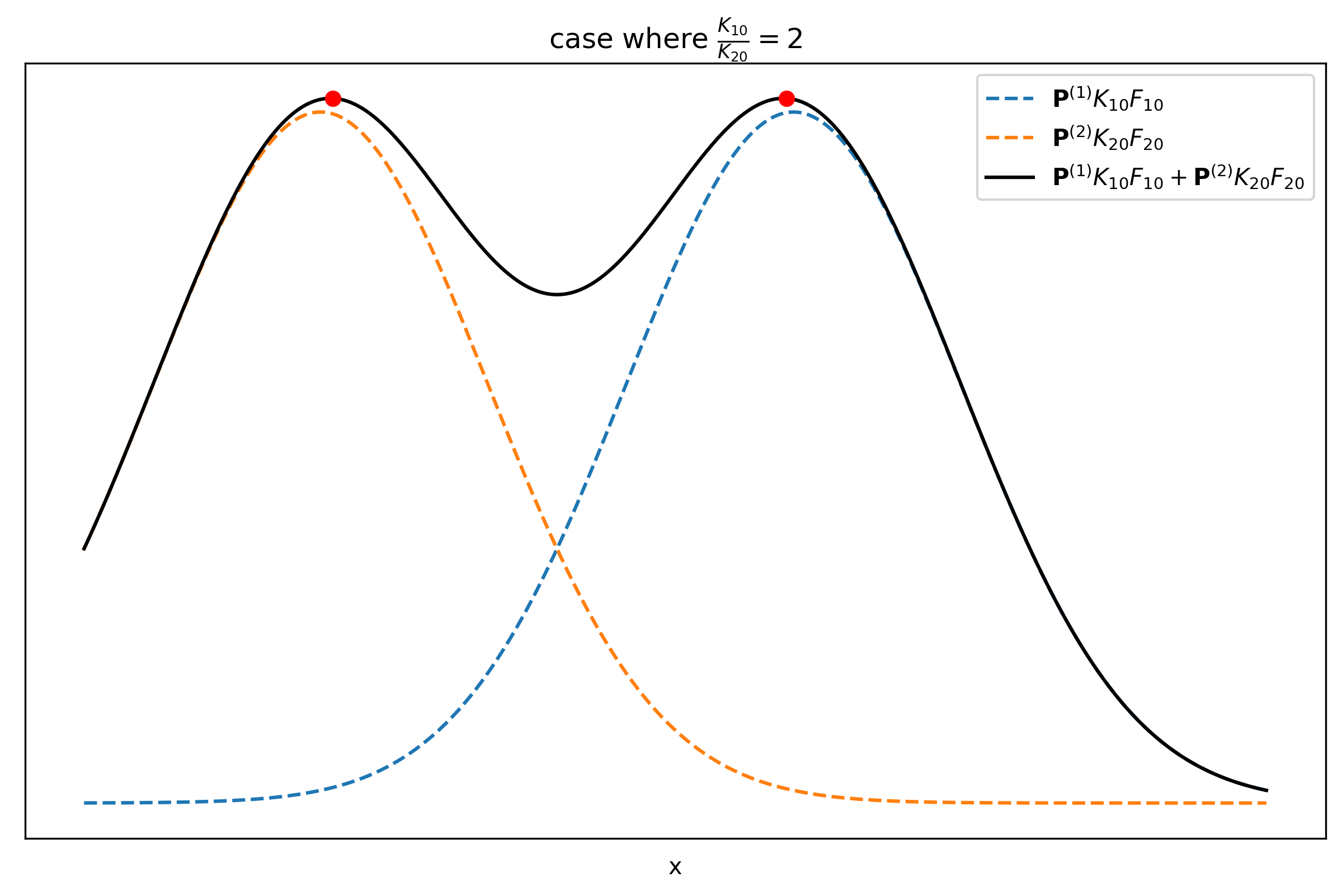}
		\raisebox{1.5ex}{(2)}
		\label{fig:image2}
	\end{minipage}
	\hfill 
	\begin{minipage}[b]{0.3\textwidth}
		\centering
		\includegraphics[width=\textwidth]{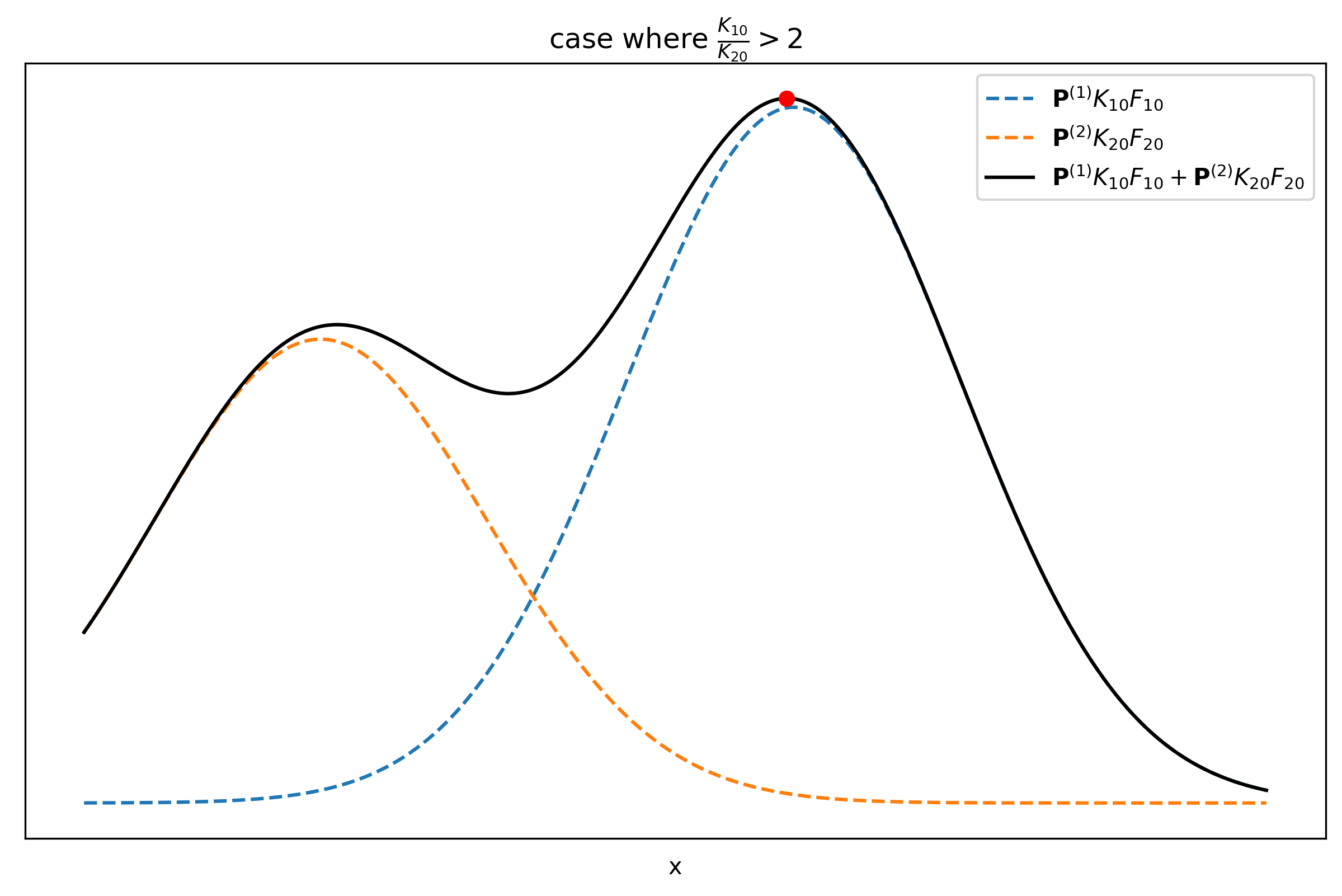}
		\raisebox{1.5ex}{(3)}
		\label{fig:image3}
	\end{minipage}
	\caption{Cash Flow Density with Different $\frac{K_{10}}{K_{20}}$}
	\label{eq2}
\end{figure}
When $\frac{K_{10}}{K_{20}}$ is less than $2$, the store favors a low-salt style, and greater than $2$, the opposite is true. However, when the ratio is $2$, the two consumers have the same power, and the store has two cash-flow-equivalent choices, which are $1.05$g and $2.97$g salt, respectively. Thus,
\begin{equation}
	\mathscr{Z}_0(I) = \left\{{\boldsymbol{z}}_{0,1}^{(s)}(I)=(1.05,2),{\boldsymbol{z}}_{0,2}^{(s)}(I)=(2.97,2)\right\}.
\end{equation}

\end{example}

For some extreme density (even if bounded), Problem \eqref{opt0} may be unsolvable, i.e. it is not possible to define the style that the store chooses. However, when $\rho_{jt}({\boldsymbol{z}}|\theta) \in \mathscr{P}_{jt},$ the following proposition guarantees the existence of the solution.
\begin{proposition}
\label{pro32}
	If $\rho_{jt}({\boldsymbol{z}}|\theta) \in \mathscr{P}_{jt},$ then $ \#(\mathscr{Z}_t(I))\geq 1.$
\end{proposition}
\begin{proof}
    In \ref{append2}.
\end{proof}
Due to the presence of different types of consumers, even stores open in the same time and with the same initial investment are likely to be dispersed in terms of style and price, which is consistent with the reality. However, among these styles there is generally a mainstream, usually designed to target sales volume to specific types of consumers. To capture this, and also for simplicity, we assume that the store owner makes decisions based on the importance of the purchase probability for each type of consumer.
\begin{assumption}[probability ordering rule]
	\label{ass1}
	The store that opens at time $t$ with initial investment $I$ will choose style
	\begin{align}
		{\boldsymbol{z}}_{t}^{(s)}(I) \in S_m,\  \text{where}&\  S_j = \underset{{\boldsymbol{z}}\in S_{j-1}}{\arg\max} \ \mathbf{P}^{(j)}\rho_{jt}\left({\boldsymbol{z}}|\theta\left({\boldsymbol{z}}\right)\right),\ \text{with}\  S_0 = \mathscr{Z}_t(I).
	\end{align}
	Note that $\#(S_j)$ may not be $1.$
\end{assumption}

The implication of Assumption \ref{ass1} is that with homogeneous cash flow, the store will make decisions based on the purchase probability of each type of consumers. From the first to the $m$-th type of consumers, their purchase probability is of decreasing importance to the store. That is, stores are more inclined to cater to the preference of the first few types of consumers, thereby greatly maximizing sales volume on them, rather than the latter types. Proposition \ref{unique} shows that ${\boldsymbol{z}}_{t}^{(s)}(I)$ is unique for a broad class of purchase probability density forms that we concerned.
\begin{proposition}
	\label{unique}
	If there exists $j_0\in\{1,2\ldots,m\}$ s.t. given $\theta$ the maximum point of $\rho_{j_0t}({\boldsymbol{z}}|\theta)$ on $D=\left\{{\boldsymbol{z}}\in\mathbb{R}^{p+q}:G_I(\boldsymbol{z})\leq I\right\}$ is unique, then $\#(S_m)\leq1.$ 
\end{proposition}
\begin{proof}
    Suppose ${\boldsymbol{z}_1},{\boldsymbol{z}_2}\in S_m,$ we have 
    \begin{equation}
        \rho_{jt}({\boldsymbol{z}}_1|\theta({\boldsymbol{z}}_1)) = \rho_{jt}({\boldsymbol{z}}_2|\theta({\boldsymbol{z}}_2)),\ 1\leq j \leq m.
    \end{equation}
   plus since ${\boldsymbol{z}_1},{\boldsymbol{z}_2}\in \mathscr{Z}_t(I),$ it leads to
    \begin{equation}
        \theta({\boldsymbol{z}}_1) = \theta({\boldsymbol{z}}_2).
    \end{equation}
    Then due to 
    \begin{equation}
\rho_{j_0t}({\boldsymbol{z}}_1|\theta({\boldsymbol{z}}_1)) = \rho_{j_0t}({\boldsymbol{z}}_2|\theta({\boldsymbol{z}}_1)),
   \end{equation}
   we know that ${\boldsymbol{z}_1}={\boldsymbol{z}_2}.$
\end{proof}
\begin{corollaryP}
    If  $\rho_{jt}(\boldsymbol{z}|\theta)\in \mathscr{P}_{jt},$ then $\#(S_m)=1.$
\end{corollaryP}
\begin{proof}
  When $\rho_{jt}(\boldsymbol{z}|\theta)\in \mathscr{P}_{jt},$ given $\theta,$
\begin{equation}
    \max_{\boldsymbol{z}:G(\boldsymbol{z})\leq I} \phi\left(C_j - \frac{\lambda_j}{2}\left\Vert\boldsymbol{z}-\hat{\boldsymbol{z}}_{j,t} \right\Vert^2\right)\cdot e^{ -\gamma_j \theta}. 
\end{equation}
is equivalent to
\begin{equation}
	\label{exa}
\min_{\boldsymbol{z}:G(\boldsymbol{z})\leq I} \left\Vert\boldsymbol{z}-\hat{\boldsymbol{z}}_{j,t} \right\Vert^2,
\end{equation}
which is a convex optimization with a strict convex objective and its solution is unique; see \cite{boyd2004convex}.
\end{proof}
Therefore, under the rule of Assumption \ref{ass1}, each store can always determine unique style and price. We denote the optimization based on this rule as $\widetilde{\max}.$ Now we are going to introduce some propositions of $\widetilde{\max}.$

\begin{proposition}
	If $m=1$ and $\rho_{1t}({\boldsymbol{z}}|\theta) \in \mathscr{P}_{1t},$ then
	\begin{equation}
		\label{case1}
		\underset{\left(\boldsymbol{z},\theta\right):G_I(\boldsymbol{z}) \leq I}{\arg\widetilde{\max}}\quad\theta\rho_{1t}({\boldsymbol{z}}|\theta)=
		\underset{\left(\boldsymbol{z},\theta\right):G_I(\boldsymbol{z}) \leq I}{\arg{\max}}\quad\theta \rho_{1t}({\boldsymbol{z}}|\theta).
	\end{equation}
\end{proposition}
\begin{proof}
	It is sufficient to show that the right-hand side of (\ref{case1}) is unique. For all given $\boldsymbol{z}=\{\boldsymbol{z}\in\mathbb{R}^{p+q}:G(\boldsymbol{z}) \leq I\},$ $\rho_{1t}$ reaches its maximum at $\theta^*=\frac{1}{\gamma_1}.$ Then \begin{equation}
		\underset{\left(\boldsymbol{z},\theta\right):G(\boldsymbol{z})}{\max} \rho_{1t}\left(\boldsymbol{z}|\theta^*\right)
	\end{equation}
	is equivalent to \eqref{exa}, of which the solution is unique; see \cite{boyd2004convex}. 
\end{proof}
\begin{proposition}
	\label{pro5}
	If $\rho_{jt}({\boldsymbol{z}}|\theta) \in \mathscr{P}_{jt},$ denote that 
	\begin{equation}
		\rho_{jt}({\boldsymbol{z}}|\theta) = K_{jt} \phi\left(C_j - \frac{\lambda_j}{2}\left\Vert\boldsymbol{z}-\hat{\boldsymbol{z}}_{j,t} \right\Vert^2\right)\cdot e^{ -\gamma_j \theta}.
	\end{equation} There exists bounded positive function $B_t(\cdot)$ such that on
	\begin{equation}
		\Psi_{B_t} =\left\{(K_{1t},\ldots,K_{mt}):K_{1t}>0,0<K_{jt}<B_t(K_{1t})\right\},
	\end{equation}
	the decisions of store owners
	\begin{equation}
		\label{p5}
		\underset{\left(\boldsymbol{z},\theta\right):G_I(\boldsymbol{z}) \leq I}{\arg\widetilde{\max}}\quad\theta\sum_{j=1}^m\mathbf{P}^{(j)}\rho_{jt}({\boldsymbol{z}}|\theta) = \underset{\left(\boldsymbol{z},\theta\right):G_I(\boldsymbol{z}) \leq I}{\arg{\max}}\quad\theta\sum_{j=1}^m\mathbf{P}^{(j)}\rho_{jt}({\boldsymbol{z}}|\theta) 
	\end{equation}
	and are smooth w.r.t. $(K_{1t},\ldots,K_{mt}).$
\end{proposition}
\begin{proof}
    In \ref{append2}.
\end{proof}
In this case, Stores will cater to the preference of the first type of consumers, due to their relatively higher level of purchase probability. We say that the first type of consumers dominate the demand market, or that the first type of consumers are the dominant consumers.

\begin{example}
	We still consider illustrative example in Example \ref{ex1}. Under Assumption \ref{ass1}, we show below the plot of the optimal amount of salt in the store for different $\frac{K_{10}}{K_{20}},$ which is unique.
\begin{figure}[h]

	\centering
	\begin{minipage}[b]{0.3\textwidth}
		\centering
		\includegraphics[width=\textwidth]{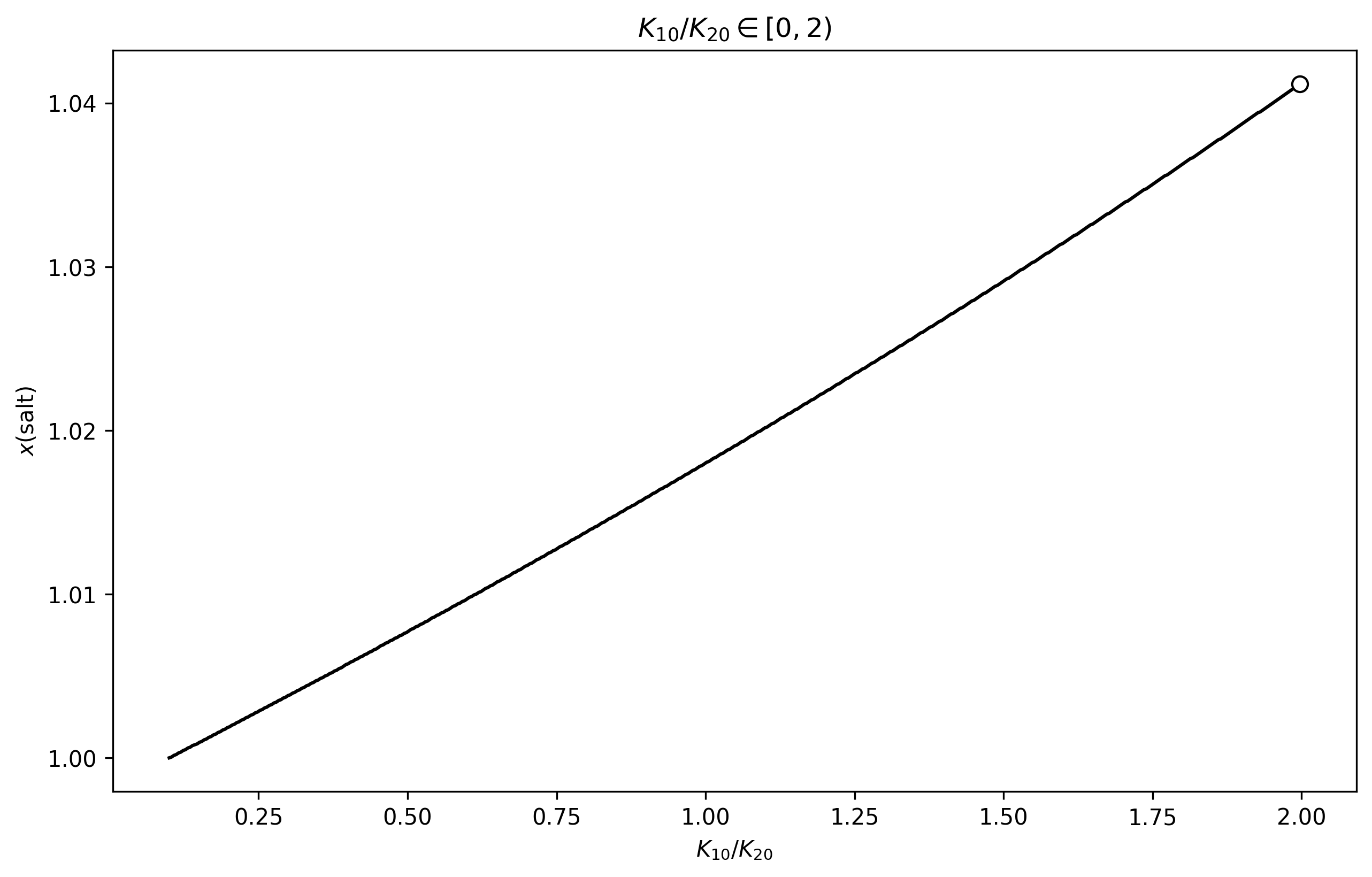}
		\raisebox{1.5ex}{(1)}
		\label{fig:image1}
	\end{minipage}
	\hfill 
	\begin{minipage}[b]{0.3\textwidth}
		\centering
		\includegraphics[width=\textwidth]{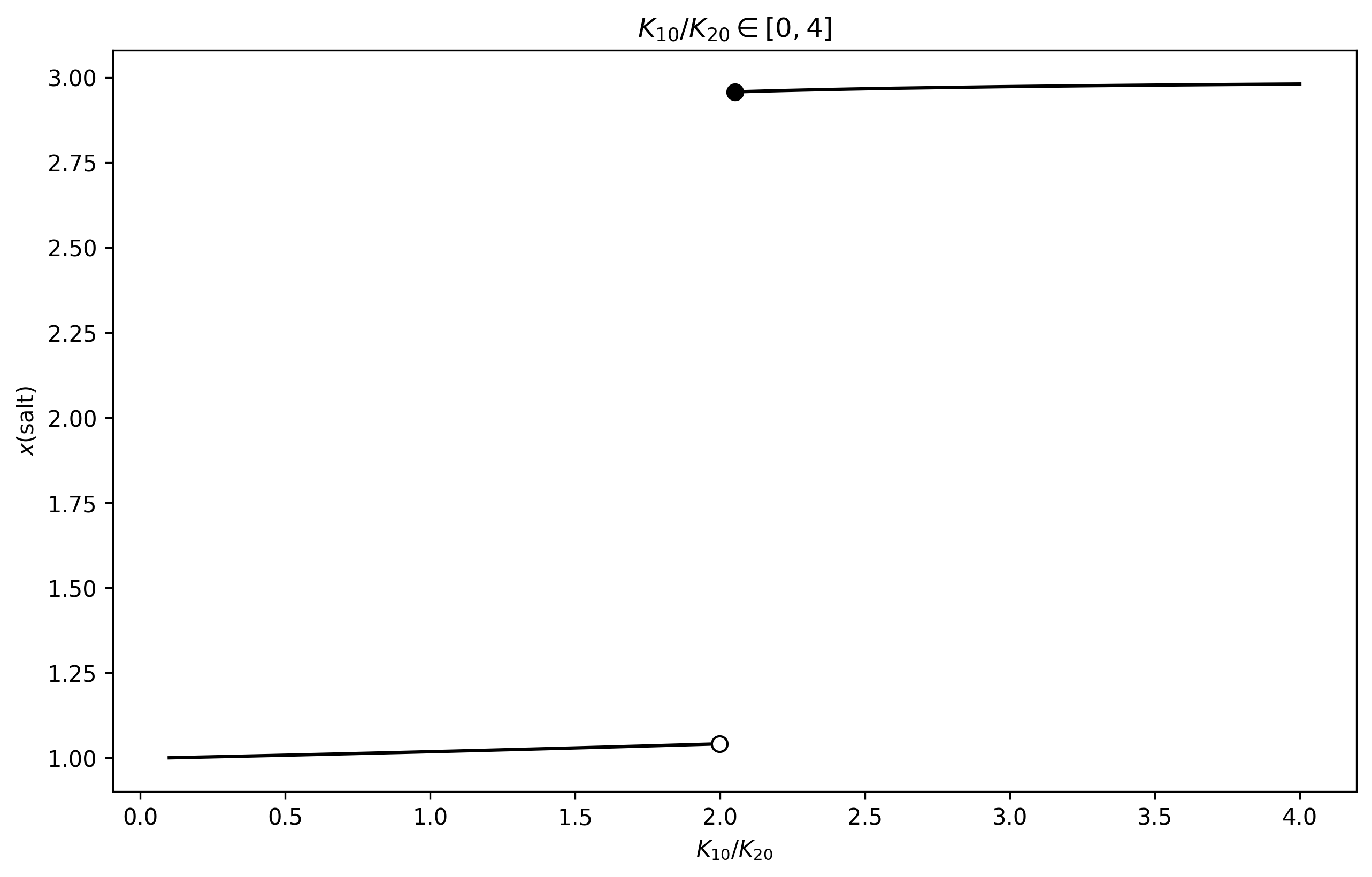}
		\raisebox{1.5ex}{(2)}
		\label{fig:image2}
	\end{minipage}
	\hfill 
	\begin{minipage}[b]{0.3\textwidth}
		\centering
		\includegraphics[width=\textwidth]{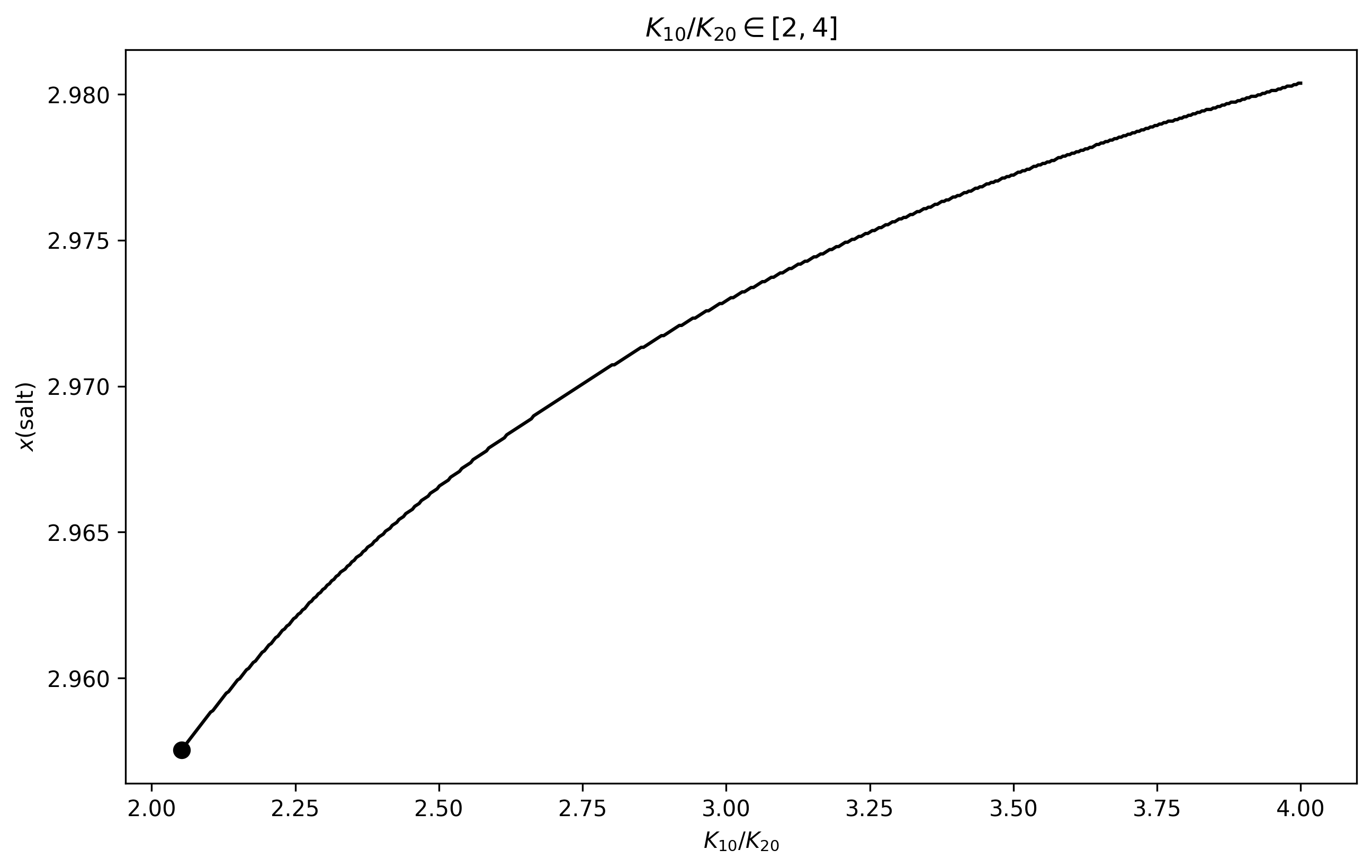}
		\raisebox{1.5ex}{(3)}
		\label{fig:image3}
	\end{minipage}
	\caption{Optimal Strategy for Store with Different $\frac{K_{10}}{K_{20}}$}
	\label{x_react}
\end{figure}
There is discontinuous variation in store decisions when the ratio of $K_{10}$ and $K_{20}$ spans $2$, but decisions are continuously smooth when the first type of consumer dominates, i.e., when the ratio is greater than $2$. Notice that when the ratio is exactly $2,$ by Assumption \ref{ass1}, the store owners will choose a higher amount of salt because they are more focused on sales to the young people.
 
\end{example}

\section{Part II: Equilibrium}
\label{Part II: Equilibrium}
\subsection{Preference shifting}
After introducing the strategies of store owners and consumers, we naturally wish to find an equilibrium, i.e. the demand-supply matching situation. Before providing an equilibrium, we need to add an important assumption.
\begin{assumption}[preference shifting]
	\label{a1}
	The product traditional preference style is linearly shifting 
	\begin{equation}
		\bar{\boldsymbol{x}}_t = \bar{\boldsymbol{x}}_0 +t\frac{\partial\bar{\boldsymbol{x}}_s}{\partial s}\bigg|_{s=0} = \bar{\boldsymbol{x}}_0 +t {\boldsymbol{c}},
	\end{equation}
	while the storefront traditional preference style does not change over time. Thus
	\begin{equation}
		\left(\begin{array}{c}
			\bar{\boldsymbol{x}}_t  \\
			\bar{\boldsymbol{\xi}}_t 
		\end{array}\right) = \left(\begin{array}{c}
			\bar{\boldsymbol{x}}_0  \\
			\bar{\boldsymbol{\xi}}_0 
		\end{array}\right)+t {\boldsymbol{d}},
	\end{equation}
	where $\boldsymbol{d}=(\boldsymbol{c}',\boldsymbol{0}')'.$
\end{assumption}

\cite{simonson1992choice} illustrates that consumers' preferences for products are influenced by information, advertising, and other market factors. Although the shift in product preferences may be non-linear, it is slower compared to the lifespan of consumer stores. Our assumption uses a first-order approximation of the shift. The following proposition states that under this Assumption \ref{a1}, the scoring function for each type of consumers has time translation invariance.

\begin{proposition}
	\label{pro6}
	Given type $j,$ the optimal preference style is also linearly shifting:
	\begin{equation}
		\hat{\boldsymbol{z}}_{j,t} = \hat{\boldsymbol{z}}_{j,0}+t{\boldsymbol{d}}.
	\end{equation}
	Thus the scoring function satisfies
	\begin{equation}
		F_{jt}({\boldsymbol{z}},\theta) = F_{j0}({\boldsymbol{z}}-t{\boldsymbol{d}},\theta).
	\end{equation}
\end{proposition}
\begin{proof}
    From (\ref{z}) we have
    \begin{align}
        \hat{\boldsymbol{z}}_{j,t}  &=\left(\begin{array}{c}\frac{1}{\lambda_j} \boldsymbol{a}_j + \bar{\boldsymbol{x}}_t \\
			\frac{1}{\lambda_j} \boldsymbol{b}_j + \bar{\boldsymbol{\xi}}_t
		\end{array}\right) \\
  &=\left(\begin{array}{c}\frac{1}{\lambda_j} \boldsymbol{a}_j + \bar{\boldsymbol{x}}_0 \\
			\frac{1}{\lambda_j} \boldsymbol{b}_j + \bar{\boldsymbol{\xi}}_0
		\end{array}\right) + t{\boldsymbol{d}}\\
  &= \hat{\boldsymbol{z}}_{j,0}+ t{\boldsymbol{d}}.
    \end{align}
    Furthermore, from (\ref{Score_simple}) we know that
    \begin{align}
        F_{jt}(\boldsymbol{z},\theta) &= \phi\left(C_j - \frac{\lambda_j}{2}\left\Vert\boldsymbol{z}-\hat{\boldsymbol{z}}_{j,t} \right\Vert^2\right) e^{- \gamma_j \theta} \\
        &=\phi\left(C_j - \frac{\lambda_j}{2}\left\Vert(\boldsymbol{z}-t{\boldsymbol{d}})-\hat{\boldsymbol{z}}_{j,0} \right\Vert^2\right)e^{ - \gamma_j \theta} \\
        &=F_{j0}(\boldsymbol{z}-t{\boldsymbol{d}},\theta).
    \end{align}
\end{proof}
The following proposition is an extension of Proposition \ref{pro5}. It demonstrates that there may be even more interesting phenomena in the presence of dominant consumers.
\begin{proposition}
	\label{pro7}
	Let $\rho_{jt}({\boldsymbol{z}}|\theta) = K_{jt} \phi\left(C_j - \frac{\lambda_j}{2}\left\Vert\boldsymbol{z}-\hat{\boldsymbol{z}}_{j,t} \right\Vert^2\right)\cdot e^{ -\gamma_j \theta}\in \mathcal{P}_{jt}.$ There exists bounded positive function $\tilde{B}_t(\cdot)\leq B_t(\cdot)$ such that on
	\begin{equation}
		\Psi_{\tilde{B}_t} =\left\{(K_{1t},\ldots,K_{mt}):K_{1t}>0,0<K_{jt}<\tilde{B}_t(K_{1t})\right\}\subseteq \Psi_{B_t},
	\end{equation}
	there are 
	\begin{equation}
		\label{eq62}
		\frac{\partial}{\partial K_{jt}}\int_t^{t'}K_{jt}F_{j\tau}\left({\boldsymbol{z}}_t^{(s)}(I),\theta\left({\boldsymbol{z}}_t^{(s)}(I)\right)\right) \mathrm{d}\tau
		> \sum_{2\leq i\leq m,i\neq j}\left|\frac{\partial}{\partial K_{it}}\int_t^{t'}K_{jt} F_{j\tau}\left({\boldsymbol{z}}_t^{(s)}(I),\theta\left({\boldsymbol{z}}_t^{(s)}(I)\right)\right) \mathrm{d}\tau\right|>0.
	\end{equation}
	for $j=2,3,\ldots,m,$  $t' \in (t+t_{\epsilon},\infty)$, some $t_{\epsilon}>0$ and $I\in\mathcal{I}$
\end{proposition}
\begin{proof}
    In \ref{append2}.
\end{proof}
Proposition \ref{pro7} illustrates that when the level of purchase probability of dominant consumers is relatively high to a certain point, the interaction effect between other different types of consumers becomes low.
\begin{example}
    We still consider illustrative example in Example \ref{ex1}. But now half of the old people are younger (called middle-aged) with heavier tastes, and they prefer $5$g of salt. We envision that the tastes of these three types of consumers change over time and that their most preferred amount of salt increased by $1$g per year. Their scoring functions at time $t$ are
    \begin{equation}
        F_{1t}(x,\xi,\theta) = \exp\left(1-(x-3-t)^2-\left(\xi-5\right)^2-\theta\right),
	\end{equation}\begin{equation}
	    F_{2t}(x,\xi,\theta) = \exp\left(1-(x-1-t)^2-\left(\xi-5\right)^2-\theta\right),
	\end{equation}\begin{equation}
	    F_{2t}(x,\xi,\theta) = \exp\left(1-(x-5-t)^2-\left(\xi-5\right)^2-\theta\right),
	\end{equation} respectively. We consider its integral with respect to time
 \begin{equation}
     \check{F}_{j}(x,\xi,\theta)=\int_0^{\infty} F_{jt}(x,\xi,\theta)\mathrm{d}t,\  j=1,2,3.
 \end{equation}
Consider $K_{j0}\check{F}_{j}\left(x_0^{(s)}(20),\xi_0^{(s)}(20),\theta\left(\boldsymbol{z}_0^{(s)}(20)\right)\right)$ as an approximate long-run cumulative value of the generalized purchase probability density. When $K_{10}$ is relatively large compared to $K_{20}$ and $K_{30}$, the first type of consumers become the dominant consumers, and we are going to look at the interaction between the other two types of consumers in this case. Figure \ref{ex4} shows the case where we fix $K_{10} = 1,$ and $K_{20}$ and $K_{30}$ vary in a range less than $1$, respectively. We can see that the purchase probability levels of the two have little cross-influence on each other.
 \begin{figure}[h]
	\centering
	\begin{minipage}[b]{0.4\textwidth}
		\centering
		\includegraphics[width=\textwidth]{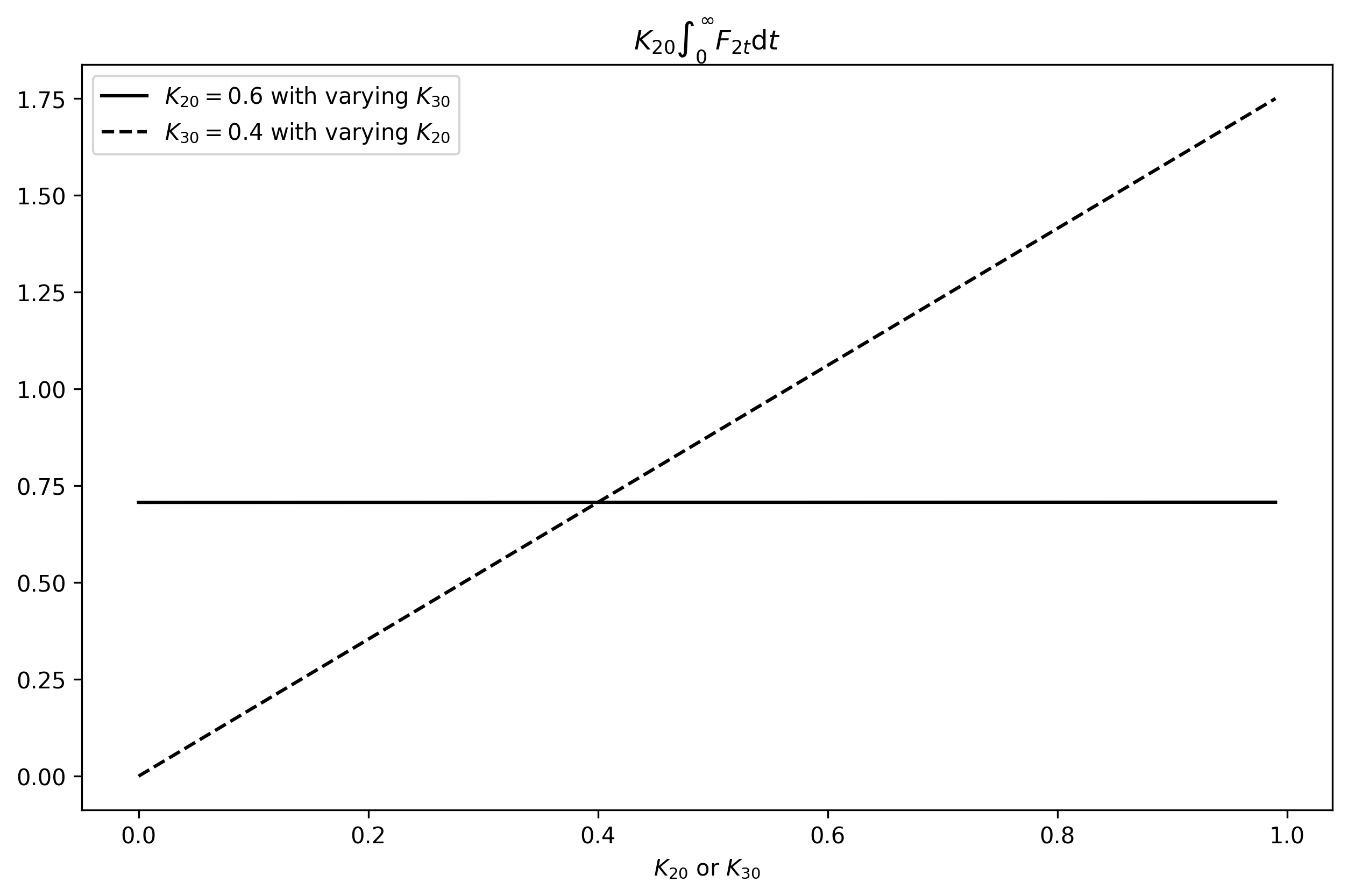}
		\raisebox{1.5ex}{(1)}
		\label{fig:image1}
	\end{minipage}
	\hspace{0.1\textwidth}
	\begin{minipage}[b]{0.4\textwidth}
		\centering
		\includegraphics[width=\textwidth]{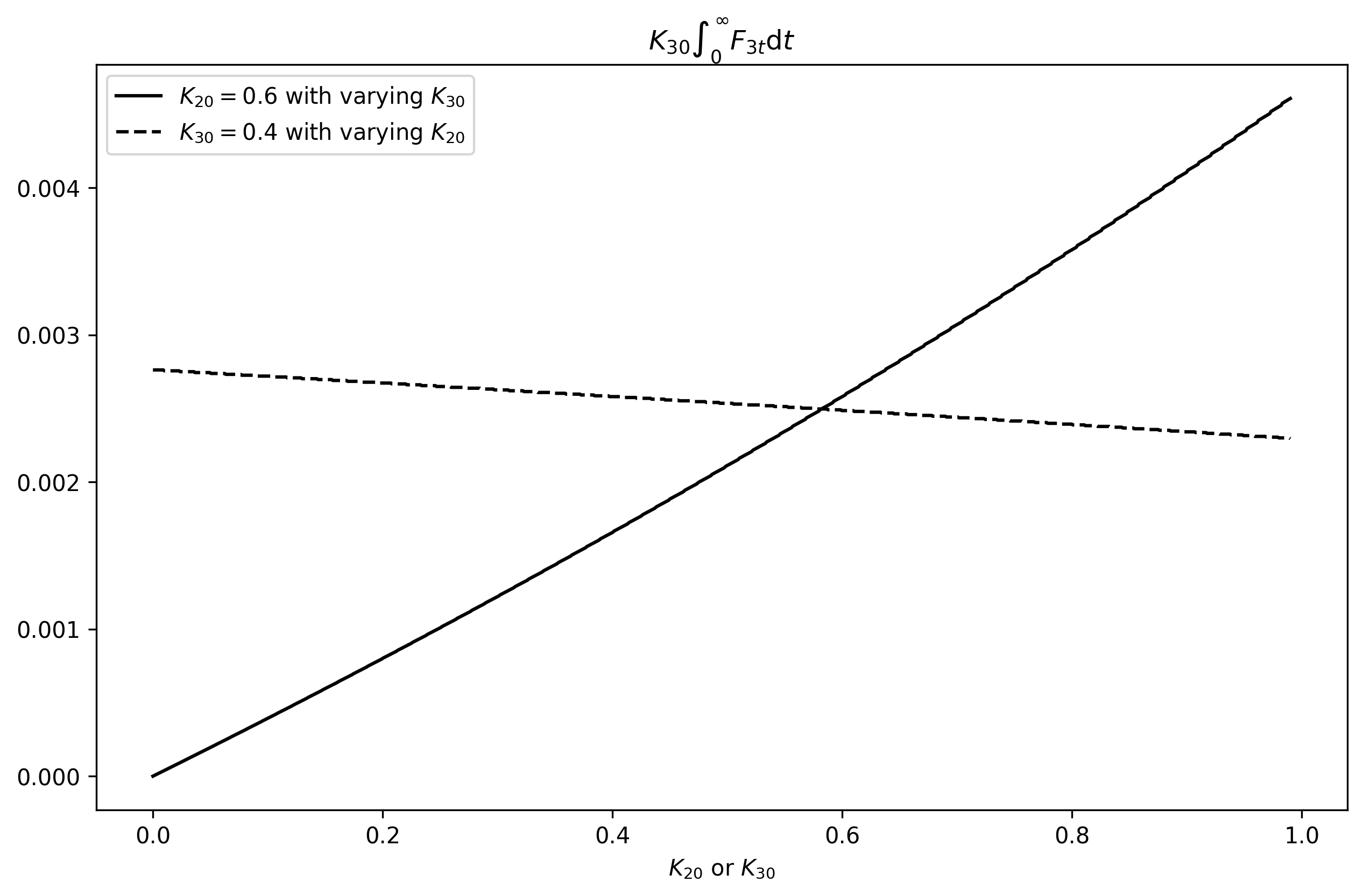}
		\raisebox{1.5ex}{(2)}
		\label{fig:image2}
	\end{minipage}
	\caption{Interaction between $2$th and $3$th consumers}
	\label{ex4}
\end{figure}
\end{example}

For the story told in Propositions \ref{pro5} and \ref{pro7} to happen, it is not in fact necessary that the level of the generalized purchase probability of the dominant consumers is really high. As Proposition {\ref{pro8}} shows, when a type of consumers is a large proportion of the population, even a low level of generalized purchase probability can make stores cater to their preference. 

\begin{proposition}
	\label{pro8}
	For any $b>0$ and $M_0>0,$ when the dominant consumers are a sufficiently large proportion, i.e. $\mathbf{P}^{(1)}$ is sufficiently large,
	$\tilde{B}_t(\cdot)$ in Proposition \ref{pro7} can be chosen such that
	\begin{equation}
		\inf_{K_{1t}\geq b }\tilde{B}_t(K_{1t}) \geq M_0
	\end{equation}
	holds.
\end{proposition}
\begin{proof}
    In \ref{append2}.
\end{proof}
	
	\subsection{Equilibrium under preference shifting}
From the perspective of game theory, we characterize this economy as a complete information game. There are two parts of players in the game. There are infinitely many (uncountable) store owners denoted as $\mathcal{T}\times\mathcal{I}$. For consumers, since the decisions at time $t_1$ and $t_2\ (t_1 \neq t_2)$ are independent (maximizing the utility at time $t_1$ and $t_2,$ respectively), consumers are also set to be infinitely many (uncountable) denoted as $\mathcal{T}\times\{1,2,\ldots,m\}.$ The action set of the store owner $(t,I)$ is $\left\{(\boldsymbol{z},\theta)\in\mathbb{R}^{p+q}\times\mathbb{R}^+:G_I(\boldsymbol{z})\leq I\right\}$, and the action set of the consumer $(t,j)$ is $\left\{\rho:\int_{\mathbb{Z}_t}\rho\mathrm{d}\boldsymbol{z}=1\right\}.$ The payoff for the store owner $(t,I)$ is $\theta\sum_{j=1}^m \mathbf{P}^{(j)} \rho_{jt}({\boldsymbol{z}}|\theta)\cdot \mathds{1}_{\{(\boldsymbol{z},\theta)\in C_t\}},$ and the payoff for the consumer $(t,j)$ is $U_{jt}(\rho).$

This game has an uncountable number of players and the action set of consumers at time $t$ depends on the strategies of the store owners in the history prior to $t.$ Also, a store owner's optimal strategy exists but is discontinuous with respect to the consumers' action (see Example \ref{ex1}). These lead us not to assert the existence of an equilibrium. However, Theorem \ref{t1} provides a Nash equilibrium (\cite{nash1951non}) for this model.

 We denote the style, price, provided style set, and consumer purchase probability density under equilibrium as ${\boldsymbol{z}}_t^*(I), \theta({\boldsymbol{z}}_t^*(I)), \mathbb{Z}^*_t, \rho^*_{jt},$ respectively.

	\begin{theorem} \label{t1}
		When $\mathbf{P}^{(1)}$ is sufficiently large, there exists $\bar{r}>0$ such that as long as $0<r(I)\leq \bar{r}(I\in\mathcal{I})$ holds, the system reaches a Nash equilibrium when
		\begin{equation}
			{\boldsymbol{z}}_t^*(I) = {\boldsymbol{z}}_0^*(I)+ t{\boldsymbol{d}}, 
		\end{equation}
		\begin{equation}
			\theta({\boldsymbol{z}}_t^*(I)) = \theta({\boldsymbol{z}}_0^*(I)), 
		\end{equation}
            \begin{equation}
		\mathbb{Z}^*_t = \mathbb{Z}^*_0 + t{\boldsymbol{d}}=\{\boldsymbol{z}+t\boldsymbol{d}:\boldsymbol{z}\in\mathbb{Z}^*_0\},
            \end{equation}
		\begin{equation}
			\rho^*_{jt}({\boldsymbol{z}}) = K_j \cdot F_{jt}({\boldsymbol{z}},\theta({\boldsymbol{z}}))\cdot\mathds{1}_{\{\boldsymbol{z}\in \mathbb{Z}^*_t\}},
		\end{equation}
		where ${\boldsymbol{z}}_0^*(I),$ $\theta({\boldsymbol{z}}_0^*(I)), \mathbb{Z}^*_0$ and $K_j(j=1,2,\ldots,m)$ satisfy
		\begin{equation}
			\label{xx1}
			\left({\boldsymbol{z}}_0^*(I),\theta({\boldsymbol{z}}_0^*(I))\right) = \underset{\left(\boldsymbol{z},\theta\right):G_I(\boldsymbol{z}) \leq I}{\arg\widetilde{\max}}\quad\theta\sum_{j=1}^m \mathbf{P}^{(j)}K_j \cdot F_{j0}({\boldsymbol{z}},\theta),
   \end{equation}
   \begin{equation}
   \label{xx7}
			\mathbb{Z}^*_0 = \left\{\boldsymbol{z}_{\tau}^{(s)}(I)| I\in \mathcal{I}, -T(I)\leq\tau\leq 0 \right\},
   \end{equation}
   \begin{equation}
			\label{xx}
			K_j^{-1} = \int_{\mathcal{I}} \int_{0}^{T(I)}F_{j0}({\boldsymbol{z}}_{0}^*(I)-\tau{\boldsymbol{d}},\theta({\boldsymbol{z}}_{0}^*(I)))\mathrm{d}\tau\mathrm{d}I,
		\end{equation}
		where
		\begin{equation}
			\label{xx2}
			T(I) = \sup\left\{\tau>0\bigg| \theta({\boldsymbol{z}}_{0}^*(I))\sum_{j=1}^m \mathbf{P}^{(j)}K_j \cdot F_{j0}({\boldsymbol{z}}_{0}^*(I)-\tau{\boldsymbol{d}},\theta({\boldsymbol{z}}_{0}^*(I)))\geq r(I)\right\}.
		\end{equation}
	\end{theorem}
    \begin{proof} The proof is divided into three steps. The first two steps respectively demonstrate the optimality of the supply and demand sides. The last step proves the existence of the equilibrium.
        
        \textbf{Step1.} When $\rho^*_{jt}({\boldsymbol{z}}) = K_j \cdot F_{jt}({\boldsymbol{z}},\theta({\boldsymbol{z}}))\cdot\mathds{1}_{\{\boldsymbol{z}\in \mathbb{Z}^*_t\}}.$ The store owner's optimization goal at time $0$ is \begin{align}
        \label{opt_local}
		\underset{{{\boldsymbol{z}},\theta}}{\widetilde{\max}} \quad &\theta\sum_{j=1}^m \mathbf{P}^{(j)}K_j \cdot F_{j0}({\boldsymbol{z}},\theta)\cdot\mathds{1}_{\{\boldsymbol{z}\in \mathbb{Z}^*_0\}}\\
		&\text{s.t.\quad} G_I(\boldsymbol{z}) \leq I,\notag\\
     & \qquad \ ({{\boldsymbol{z}},\theta}) \in C_0, \notag
	\end{align}
where
\begin{equation}
    C_0 = \left\{(\boldsymbol{z},\theta):\exists\delta_1>0,\sum_{j=1}^m \mathbf{P}^{(j)}K_j \int_{\{\Vert\boldsymbol{z}_1-\boldsymbol{z}\Vert\leq \delta_1\}\cap \{\boldsymbol{z}_1\in \mathbb{Z}^*_0\}}  F_{j0}({\boldsymbol{z}},\theta)\mathrm{d}\boldsymbol{z}_1>0\right\}.
\end{equation}
Note that $C_0= \mathbb{Z}^*_0$ and $\left({\boldsymbol{z}}_0^*(I),\theta({\boldsymbol{z}}_0^*(I))\right)\in \mathbb{Z}^*_0$. Thus, $\left({\boldsymbol{z}}_0^*(I),\theta({\boldsymbol{z}}_0^*(I))\right)$ defined as (\ref{xx1}) is also the unique solution of (\ref{opt_local}), i.e. the optimal strategy for a store that opens at time $0$ with initial investment $I.$ 
Noting that $F_{jt}({\boldsymbol{z}},\theta)= F_{j0}({\boldsymbol{z}}-t\boldsymbol{d},\theta),\mathbb{Z}^*_t = \mathbb{Z}^*_0 + t\boldsymbol{d}$ and $C_t = C_0 + t\boldsymbol{d},$ The store owner's optimization goal at time $t$ is equivalent to \begin{align}
        \label{opt_local}
		\underset{{{\boldsymbol{z}},\theta}}{\widetilde{\max}} \quad &\theta\sum_{j=1}^m \mathbf{P}^{(j)}K_j \cdot F_{j0}({\boldsymbol{z}}-t\boldsymbol{d},\theta)\cdot\mathds{1}_{\{\boldsymbol{z}-t\boldsymbol{d}\in \mathbb{Z}^*_0\}}\\
		&\text{s.t.\quad} G_I(\boldsymbol{z}-t\boldsymbol{d}) \leq I,\notag\\
     & \qquad \ ({{\boldsymbol{z}-t\boldsymbol{d}},\theta}) \in C_0. \notag
	\end{align}
We know that $\boldsymbol{z}^*_t(I) = \boldsymbol{z}^*_0(I)+t\boldsymbol{d},$ and $\theta\left(\boldsymbol{z}^*_t(I)\right)=\theta\left(\boldsymbol{z}^*_0(I)\right).$ After some algebra, we get \eqref{xx7} and \eqref{xx2}.  

        \textbf{Step2.} For type $j$-th consumers, the optimal purchase probability density satisfies 
	\begin{equation}
		\rho^*_{jt}({\boldsymbol{z}}) = K_{jt} \cdot F_{jt}({\boldsymbol{z}},\theta({\boldsymbol{z}}))\cdot\mathds{1}_{\{\boldsymbol{z}\in \mathbb{Z}^*_t\}},
	\end{equation}
where 
	\begin{align}
		K_{jt}^{-1} &= \int_{\mathbb{Z}^*_t} F_{jt}({\boldsymbol{z}},\theta({\boldsymbol{z}}))\mathrm{d}{\boldsymbol{z}}\\
		& = \int_{\mathcal{I}} \int_{t-T(I)}^{t}F_{jt}({\boldsymbol{z}}_{\tau}^*(I),\theta({\boldsymbol{z}}_{\tau}^*(I)))\mathrm{d}\tau\mathrm{d}I\\
		&=\int_{\mathcal{I}} \int_{t-T(I)}^{t}F_{j0}({\boldsymbol{z}}_{0}^*(I)-(t-\tau){\boldsymbol{d}},\theta({\boldsymbol{z}}_{0}^*(I)))\mathrm{d}\tau\mathrm{d}I\\
		&= \int_{\mathcal{I}} \int_{0}^{T(I)}F_{j0}({\boldsymbol{z}}_{0}^*(I)-\tau{\boldsymbol{d}},\theta({\boldsymbol{z}}_{0}^*(I)))\mathrm{d}\tau\mathrm{d}I.
	\end{align}
	is a constant of $t$, which is exactly $K_j^{-1}$ in \eqref{xx}. So we have 
	\begin{equation}
		\rho^*_{jt}({\boldsymbol{z}}) = K_{j} \cdot F_{jt}({\boldsymbol{z}},\theta({\boldsymbol{z}}))\cdot\mathds{1}_{\{\boldsymbol{z}\in \mathbb{Z}^*_t\}}.
	\end{equation}
        
        \textbf{Step3. } For given $\boldsymbol{K}=(K_1,\ldots,K_m),$ solve \eqref{xx1} to obtain
\begin{equation}
	{\boldsymbol{z}}_{0}^*(\boldsymbol{K},I),\ \theta({\boldsymbol{z}}_{0}^*(\boldsymbol{K},I)).
\end{equation}
Denote that
\begin{equation}
	\psi_j(\boldsymbol{K},t,I) = F_{j0}({\boldsymbol{z}}_{0}^*(\boldsymbol{K},I)-\tau{\boldsymbol{d}},\theta({\boldsymbol{z}}_{0}^*(\boldsymbol{K},I))),
\end{equation}
\begin{equation}
	T(\boldsymbol{K},I) = \sup\left\{\tau>0\bigg| \theta({\boldsymbol{z}}_{0}^*(\boldsymbol{K},I))\sum_{j=1}^m \mathbf{P}^{(j)}K_j \cdot F_{j0}({\boldsymbol{z}}_{0}^*(\boldsymbol{K},I)-\tau{\boldsymbol{d}},\theta({\boldsymbol{z}}_{0}^*(\boldsymbol{K},I)))\geq r(I)\right\}.
\end{equation}
and 
\begin{equation}
	V_j(\boldsymbol{K}) = \int_{\mathcal{I}} \int_{0}^{T(\boldsymbol{K},I)}\psi_j(\boldsymbol{K},t,I)\mathrm{d}\tau\mathrm{d}I.
\end{equation} 
Then we use backward induction to show that there exist  $\boldsymbol{K}=(K_1,\ldots,K_m),$ such that
\begin{equation}
    K_j V_j(\boldsymbol{K}) = 1, j=1,2,\ldots,m.
\end{equation}
Details are in the \ref{append2}.
    \end{proof}
    \begin{corollary}
The equilibrium in Theorem \ref{t1} is Pareto-efficient.
	\end{corollary}

	\begin{corollary}
		Even if each store is second-order rational and consumers are first-order rational, the equilibrium in Theorem \ref{t1} holds.
	\end{corollary}
	
	\begin{corollary}
		 		Even if there is cooperation between stores opening at the same time, the equilibrium in Theorem \ref{t1} holds.
	\end{corollary}
	
	From Theorem \ref{t1} we can see some characteristics of the equilibrium state. First, the styles provided by newly entered stores also shift linearly in the direction of preference shifting over time, while prices are only related to the initial investment of the stores. Second, the purchase probability density for each type of consumers can be expressed as a monotonic decreasing function of the distance of the style from the optimal preference style, which indicates that the probability of a consumer purchasing from a fixed store will decrease to zero over time. Finally, this equilibrium is time stable and the lifespan of a store is only related to its initial investment and not to the time when it opens.

	\subsection{Exponential transform function and theoretical conversion rate in equilibrium}
	In this Nash equilibrium, we focus on the life of a fixed store. Assume that this store opens at time $0$. Under the constraint of $g_I(\boldsymbol{\boldsymbol{\xi}})\leq I,$ it chooses the style of ${\boldsymbol{z}}^*_0(I)$ and the price of $\theta({\boldsymbol{z}}^*_0(I)).$ Then at time $t,$ its aggregated purchase probability density is
	\begin{align}
		\rho_t^*({\boldsymbol{z}}^*_0(I)) &=\sum_{j=1}^m \mathbf{P}^{(j)}\rho^*_{jt}({\boldsymbol{z}}^*_0(I))\\
		&= \sum_{j=1}^m \mathbf{P}^{(j)}K_j \cdot F_{j0}({\boldsymbol{z}}_{0}^*(I)-t{\boldsymbol{d}},\theta({\boldsymbol{z}}^*_0(I))) \\
		&=\sum_{j=1}^m \mathbf{P}^{(j)}K_j \cdot \phi\left(C_j-\frac{\lambda_j}{2}\Vert{\boldsymbol{z}}^*_0(I)-\hat{\boldsymbol{z}}_{j0}\Vert^2+\lambda_j	\langle{{\boldsymbol{z}}^*_0(I)-\hat{\boldsymbol{z}}_{j0},\boldsymbol{d}}\rangle t-\frac{\lambda_j}{2}\Vert{\boldsymbol{d}}\Vert^2t^2-\gamma_j\theta({\boldsymbol{z}}^*_0(I))\right).
	\end{align}
	It can be seen that as time increases, the purchase probability density of type $j$ consumers decreases with the order of $\phi(-t^2),$ which results in a decrease over time in the aggregated purchase probability density. 
	
	The purchase probability of the type $j$ consumers should be calculated as an integral of the density near the style
	\begin{equation}
		\tilde{\beta}_{jt}({\boldsymbol{z}}^*_0(I)) = \int_{\Delta(\boldsymbol{z}_0(I),\epsilon)} \rho^*_{jt}({\boldsymbol{z}}_1)\mathrm{d}{\boldsymbol{z}_1}
	\end{equation}
	for some unobservable $\epsilon$ which is not easy to estimate in reality. But we can boil it down to the initial purchase probability
	\begin{equation}\label{eq:ini_pur_prob}
		\begin{aligned}
			\tilde{\beta}_{jt}({\boldsymbol{z}}^*_0(I)) &= \tilde{\beta}_{j0}({\boldsymbol{z}}^*_0(I)) \frac{\int_{\Delta_t(\boldsymbol{z}_0(I),\epsilon)} \rho^*_{jt}({\boldsymbol{z}}_1)\mathrm{d}{\boldsymbol{z}}_1}{\int_{\Delta_0(\boldsymbol{z}_0(I),\epsilon)} \rho^*_{j0}({\boldsymbol{z}}_1)\mathrm{d}{\boldsymbol{z}}_1}.
		\end{aligned}
	\end{equation}
	By simply assuming that $\Gamma_{\epsilon}(\mathbb{Z}^*_t) = \Gamma_{\epsilon}(\mathbb{Z}^*_0)+t\boldsymbol{d},$ which implies that $|\Delta_t(\boldsymbol{z}_0(I),\epsilon)|=|\Delta_0(\boldsymbol{z}_0(I),\epsilon)|,$ we further get
 \begin{equation}
     \tilde{\beta}_{jt}({\boldsymbol{z}}^*_0(I)) = \tilde{\beta}_{j0}({\boldsymbol{z}}^*_0(I)) \frac{\rho^*_{jt}({\boldsymbol{z}}^*_0(I))}{\rho^*_{j0}({\boldsymbol{z}}^*_0(I))} +o(1)\ (\epsilon \rightarrow 0).
 \end{equation}
 Furthermore, if we choose the transform function as $\phi(\cdot) = \exp(\cdot)$, we have
	$$
	\rho_{jt}^*({\boldsymbol{z}}^*_0(I)) = K_j \exp\left\{ c_j - \frac{1}{2} \lambda_j || {\boldsymbol{z}}^*_0(I) - \hat{\boldsymbol{z}}_{j0} ||^2 \right\} \exp\left\{ \lambda_j	\langle{{\boldsymbol{z}}^*_0(I)-\hat{\boldsymbol{z}}_{j0},\boldsymbol{d}}\rangle t-\frac{\lambda_j}{2}\Vert{\boldsymbol{d}}\Vert^2t^2 \right\}.
	$$
	We further omit $o(1)$ in \eqref{eq:ini_pur_prob} and have
	$$
	\tilde{\beta}_{jt}({\boldsymbol{z}}^*_0(I)) = \tilde{\beta}_{j0}({\boldsymbol{z}}^*_0(I)) \frac{\rho^*_{jt}({\boldsymbol{z}}^*_0(I))}{\rho^*_{j0}({\boldsymbol{z}}^*_0(I))} = \tilde{\beta}_{j0}({\boldsymbol{z}}^*_0(I)) \exp\left\{ \lambda_j	\langle{{\boldsymbol{z}}^*_0(I)-\hat{\boldsymbol{z}}_{j0},\boldsymbol{d}}\rangle t-\frac{\lambda_j}{2}\Vert{\boldsymbol{d}}\Vert^2t^2 \right\}.
	$$
	Hence, we simplify $\tilde{\beta}_{jt}({\boldsymbol{z}}^*_0(I))$ as 
	\begin{equation}
		\tilde{\beta}_{t}(I) = \sum_{j=1}^{m}\mathbf{P}^{(j)} \tilde{\beta}_{j0}(I)\exp\left(\mu_j t- \nu_j t^2\right),
	\end{equation}
	where $\mu_j=\lambda_j	\langle{{\boldsymbol{z}}^*_0(I)-\hat{\boldsymbol{z}}_{j0},\boldsymbol{d}}\rangle$ and $\nu_j=\frac{\lambda_j}{2}\Vert{\boldsymbol{d}}\Vert^2.$
	
	The initial purchase probability represents the probability of consumers who are completely unfamiliar with the store making purchases. This value is determined by the unique nature of the store itself, which depends on the initial investment of the store.
	
	When $t$ is sufficiently large, $\tilde{\beta}_t$ monotonically decreases with respect to $t$ and approaches $0$. This is mainly due to the product style being eliminated by the market due to the style-shifting. However, due to the dispersion of preferences among different types of people, $\tilde{\beta}_t$ may experience some fluctuations before decreasing. Figure \ref{Conversion_Rate_Curves} shows the conversion rate curves of several different shapes. The following theorem demonstrates that the $\tilde{\beta}_t$ undergoes at most one growth before decreasing to $0$ when the product preferences of different types are not dispersed. That is to say, only (1) and (2) in Figure \ref{Conversion_Rate_Curves} will appear when people's product preferences are sufficiently concentrated. In other words, the $\tilde{\beta}_t$ curve has the shape of a rainbow or is a part of it.
	
	\begin{figure}[h]
		\centering
		\begin{minipage}[b]{0.3\textwidth}
			\centering
			\includegraphics[width=\textwidth]{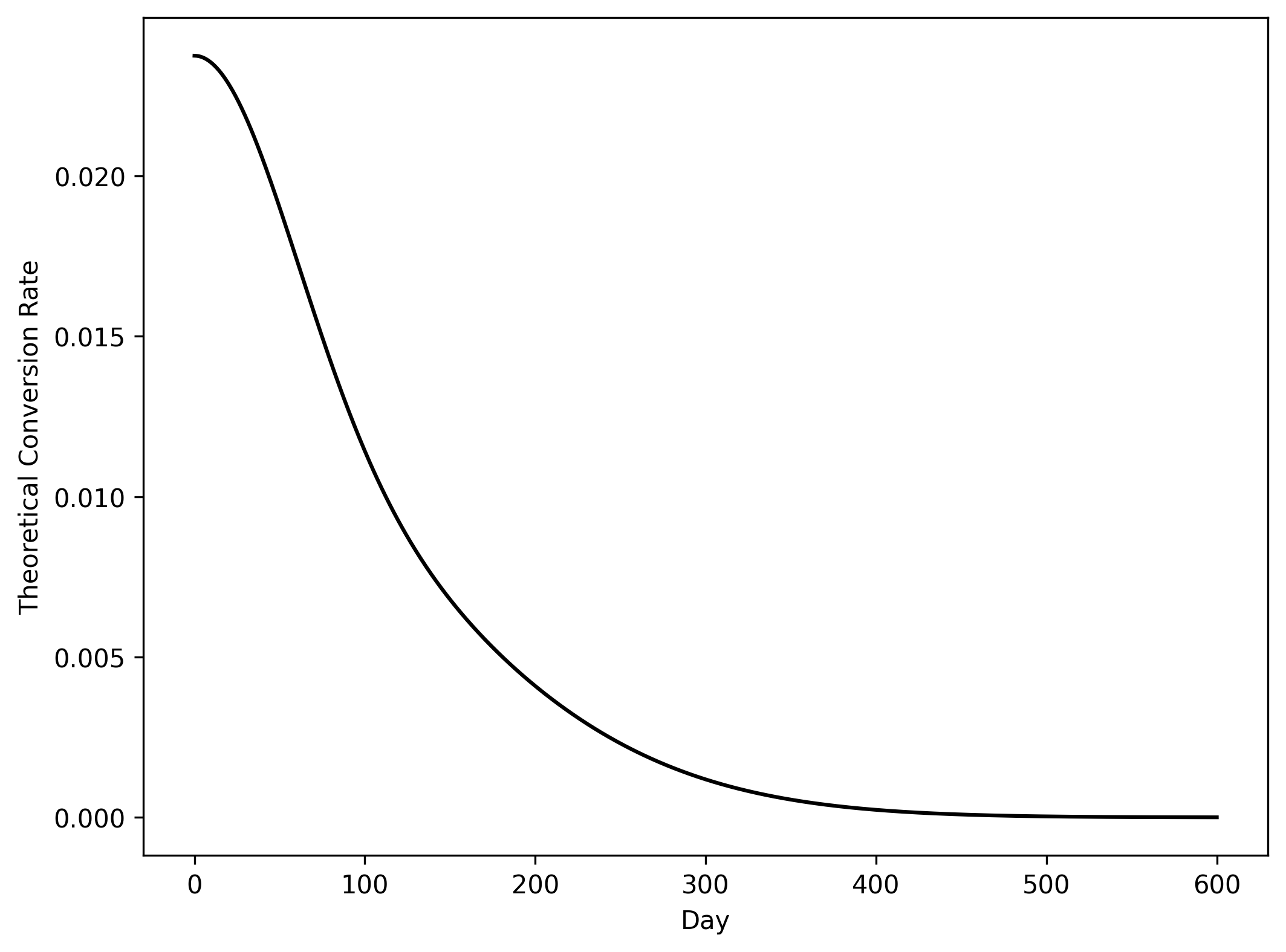}
			\raisebox{1.5ex}{(1)}
			\label{fig:image1}
		\end{minipage}
		\hfill 
		\begin{minipage}[b]{0.3\textwidth}
			\centering
			\includegraphics[width=\textwidth]{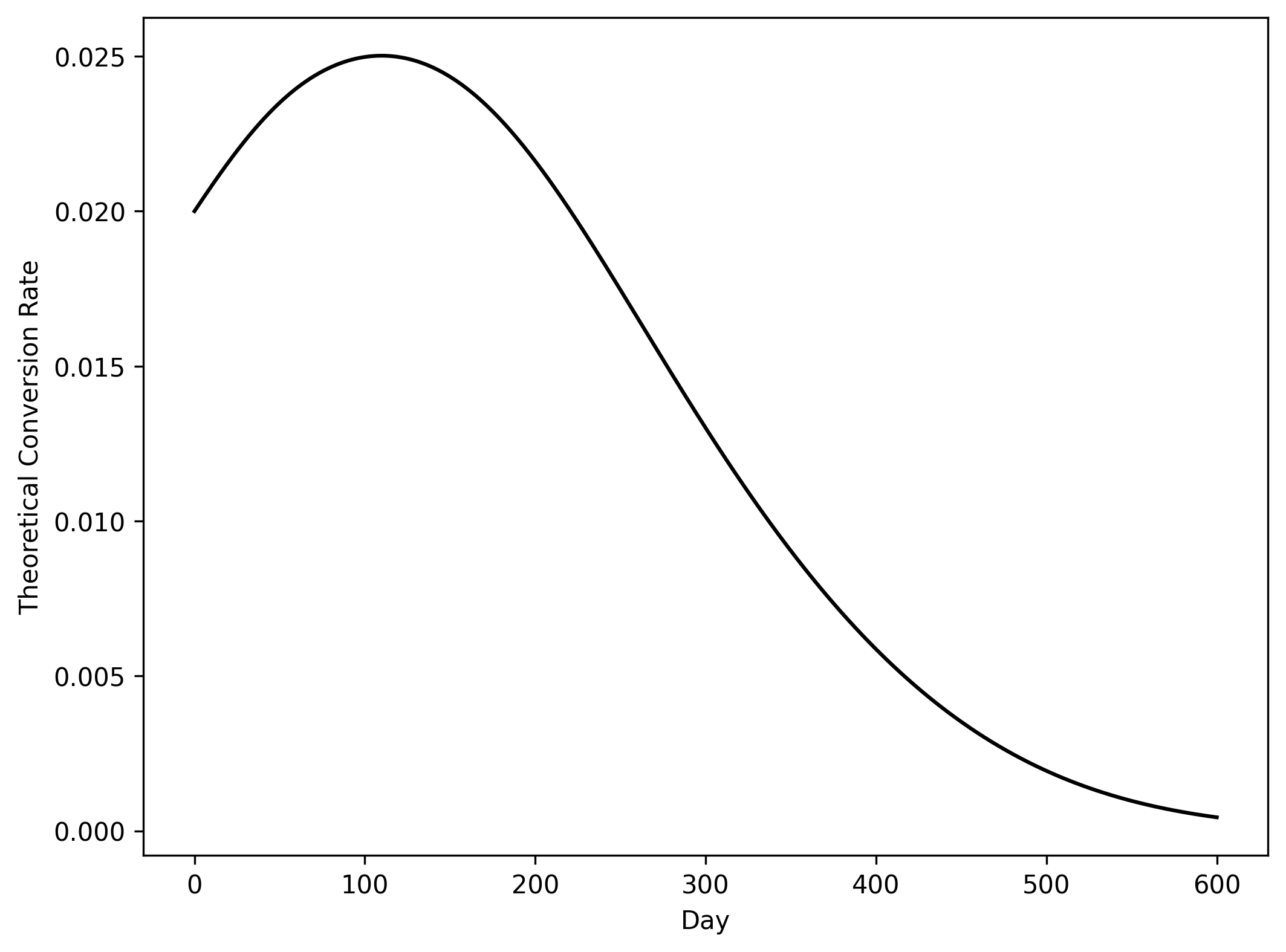}
			\raisebox{1.5ex}{(2)}
			\label{fig:image2}
		\end{minipage}
		\hfill 
		\begin{minipage}[b]{0.3\textwidth}
			\centering
			\includegraphics[width=\textwidth]{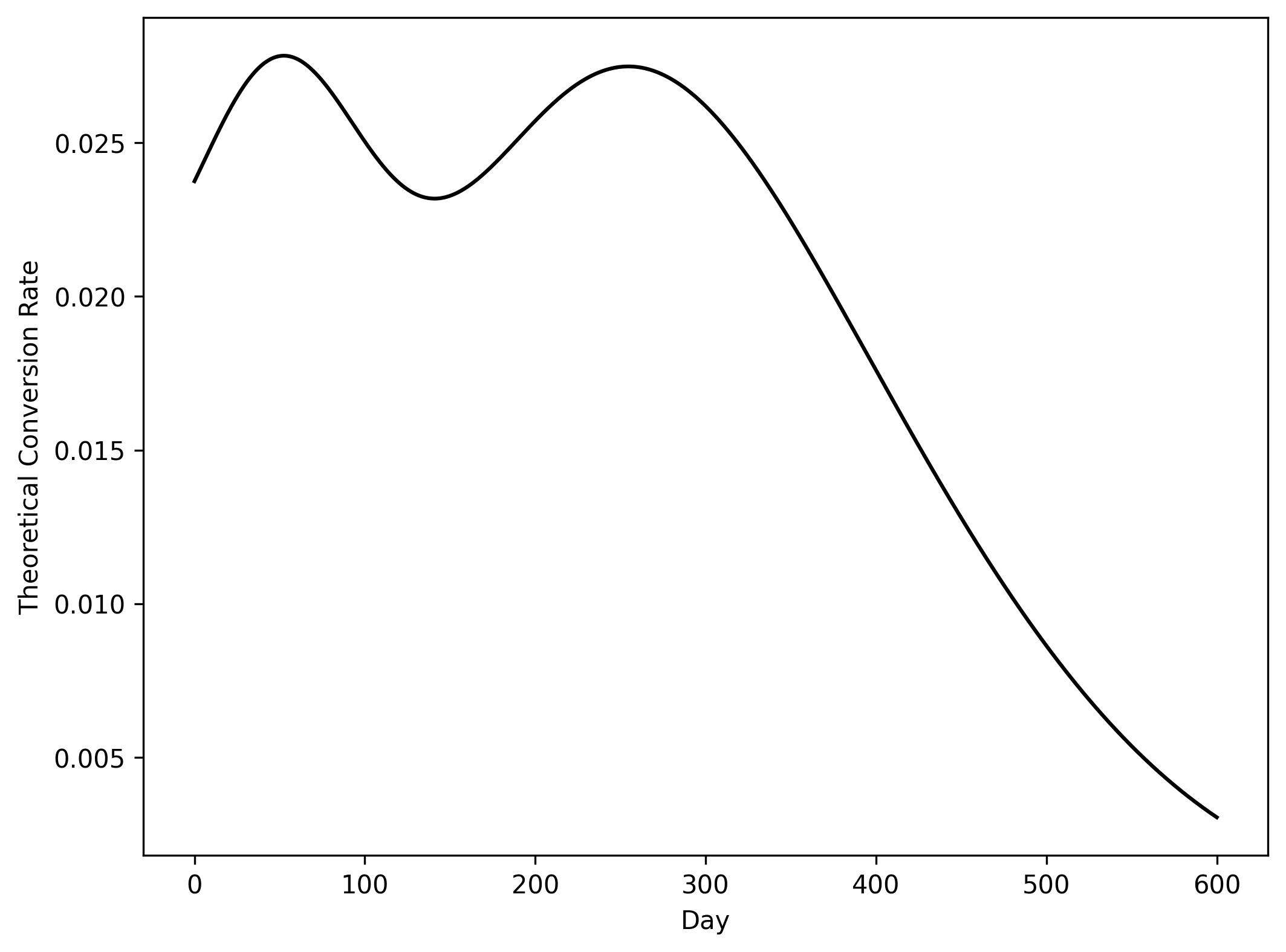}
			\raisebox{1.5ex}{(3)}
			\label{fig:image3}
		\end{minipage}
		\caption{Different Types of Conversion Rate Curves}
		\label{Conversion_Rate_Curves}
	\end{figure}
	
	\begin{proposition} \label{pro9} 
 Fix $I$ and assume that $\lambda_j \leq M$ for some $M>0.$ For any $\boldsymbol{d} \in \mathcal{R}^{p+1},$ there exists $\epsilon(\boldsymbol{d}) >0$ s.t. when $\sup_{1\leq i,j\leq m }\left\Vert \frac{\boldsymbol{a}_i}{\lambda_i}-\frac{\boldsymbol{a}_j}{\lambda_j}\right\Vert<\epsilon,$ $\tilde{\beta}_{t}(I)$ has at most one maximum point on $[0,\infty).$
	\end{proposition}
\begin{proof}
    In \ref{append2}.
\end{proof}
\subsection{Style updating}
\label{chasing}
In reality, although most stores will close down and go out of business for a limited period of time, there are still a handful of centuries-old stores. Their cash flow decreases significantly slower than other stores in the same industry or their cash flow remains stable for a long period of time. 

Let us consider that in the equilibrium in Theorem \ref{t1} a store which opens at time $t$ with initial investment $I$ has the ability to update the style. There is a cost to update style, and the store budgets $S>0$ for the cost of style updating. From time $0$ to time $t > 0$, the store can change its style from $\boldsymbol{z}^*_{0}(I)$ to $\boldsymbol{z}^*_{0\rightarrow t}(I)$ smoothly. However, it needs to satisfy constraint \begin{equation}
    \left\Vert\frac{\partial}{\partial t} \boldsymbol{z}^*_{0\rightarrow t}(I) \right\Vert \leq g_S(S),
\end{equation} where $g_S(\cdot)$ is the updating efficiency of the style-chasing budget, which is a positive increasing function.
In order to have a longer lifespan, the store's optimization goal is 
\begin{align}
    \max_{{z}^*_{0\rightarrow t}(I)} \quad&\inf\left\{t>0\bigg| \theta({z}^*_{0\rightarrow t}(I))\sum_{j=1}^m \mathbf{P}^{(j)}K_j F_{jt}({z}^*_{0\rightarrow t}(I),\theta({z}^*_{0\rightarrow t}(I)))\cdot\mathds{1}_{\{{z}^*_{0\rightarrow t}(I)\in \mathbb{Z}^*_t\}}< r(I)\right\} \\
    &\text{s.t.}\ \left\Vert\frac{\partial}{\partial t} \boldsymbol{z}^*_{0\rightarrow t}(I) \right\Vert \leq g_S(S), \notag\\ 
    &\qquad \ G_I(\boldsymbol{z}^*_{0\rightarrow t}(I)) \leq I,\notag\\
    &\qquad \left.\boldsymbol{z}^*_{0\rightarrow t}(I)\right|_{t=0} = \boldsymbol{z}^*_{0}(I).\notag
\end{align}
The existence of a finite number of stores with this ability does not affect the equilibrium in Theorem \ref{t1}.

\textbf{When the updating budget is adequate.} When $g_S(S)\geq \Vert\boldsymbol{d}\Vert,$ the optimal strategy of the store is 
\begin{equation}
    {z}^*_{0\rightarrow t}(I) = {z}^*_{t}(I) = {z}^*_{0}(I) + t \boldsymbol{d}.
\end{equation}
In this case, the purchase probability density for this store at time $t$ is
\begin{align}
    \sum_{j=1}^m \mathbf{P}^{(j)}K_j F_{jt}({z}^*_t(I),\theta({z}^*_t(I))) &= \sum_{j=1}^m \mathbf{P}^{(j)}K_j F_{j0}({z}^*_t(I)-t\boldsymbol{d},\theta({z}^*_t(I))) \\
    &=\sum_{j=1}^m \mathbf{P}^{(j)}K_j F_{j0}({z}^*_0(I),\theta({z}^*_0(I))).
\end{align}
Consumers' interest in this store remains stable over time, and thus the store's cash flow does not decrease.

\textbf{When the updating budget is inadequate.} When $g_S(S)< \Vert\boldsymbol{d}\Vert,$ the optimal strategy of the store is 
\begin{equation}
    {z}^*_{0\rightarrow t}(I) ={z}^*_{\check{t}(t,S)}(I),
\end{equation}
where 
$\check{t}(t,S) = \frac{g_S(S)}{\Vert\boldsymbol{d}\Vert}\cdot t.$ 
In this case, the purchase probability density for this store at time $t$ is
\begin{align}
    \sum_{j=1}^m \mathbf{P}^{(j)}K_j F_{jt}\left({z}^*_{\check{t}(t,S)},\theta\left({z}^*_{\check{t}(t,S)}(I)\right)\right) &= \sum_{j=1}^m \mathbf{P}^{(j)}K_j F_{j0}\left({z}^*_{\check{t}(t,S)}-t\boldsymbol{d},\theta\left({z}^*_{\check{t}(t,S)}(I)\right)\right) \\
    &=\sum_{j=1}^m \mathbf{P}^{(j)}K_j F_{j0}({z}^*_0(I)-t\omega(S)\boldsymbol{d},\theta({z}^*_0(I))),
\end{align}
where $\omega(S) = \left(1-\frac{g_S(S)}{\Vert\boldsymbol{d}\Vert}\right).$ After the same operations as in the previous subsection, we get
    \begin{equation}
		\tilde{\beta}_{t}(I) = \sum_{j=1}^{m}\mathbf{P}^{(j)} \tilde{\beta}_{j0}(I)\exp\left(\omega(S)\mu_j t- \omega^2(S)\nu_j t^2\right),
	\end{equation}
where $\mu_j=\lambda_j	\langle{{\boldsymbol{z}}^*_0(I)-\hat{\boldsymbol{z}}_{j0},\boldsymbol{d}}\rangle$ and $\nu_j=\frac{\lambda_j}{2}\Vert{\boldsymbol{d}}\Vert^2.$ The coefficient of $t^2$ is $\omega^2(S)\nu_j<\nu_j$, indicating that in the case of style updating, the consumers' purchase probability decreases more slowly and the store has a longer lifespan. Figure \ref{stylechasing} shows the impact that style updating can have on the store life cycle, where the length of the red arrows represents $\Vert\boldsymbol{d}\Vert$, and the length of the blue arrows represents $g_S(S).$
 \begin{figure}[h]
	\centering
	\includegraphics[width = 0.6\textwidth]{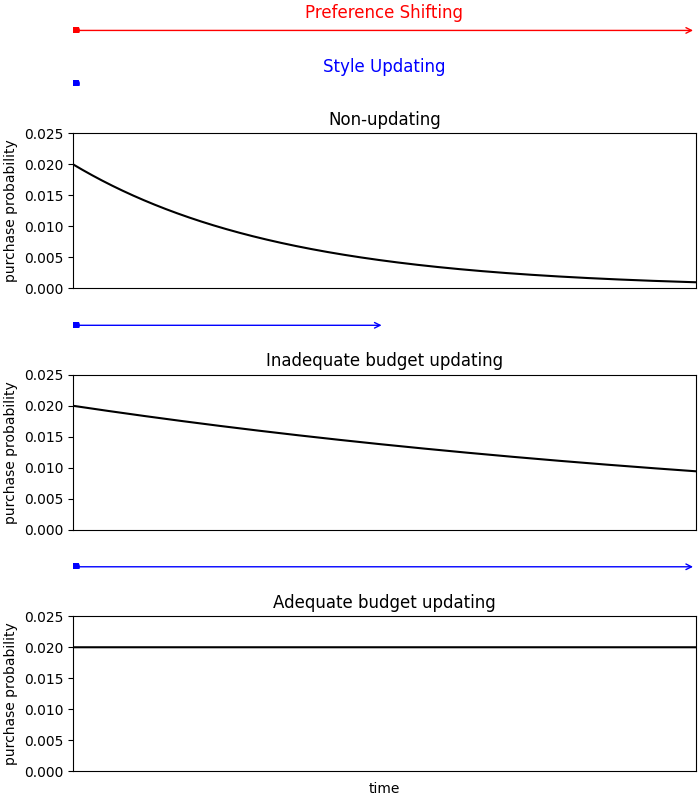}
	\caption{Style updating}   
	\label{stylechasing}
\end{figure}

\section{Part III: Potential customers} 
\label{Part III: Potential customers}

In real life, oftentimes, we are unwilling to travel long distances to purchase from a store. This naturally suggests that willingness to purchase is negatively correlated with distance. We set the proportion of people willing to enter the store to buy to decrease exponentially with distance. The store attracts people in the coordinate $\mathbf{x}$ at most in proportion to $e^{-\delta\Vert\mathbf{x}\Vert},$ where $\delta$ is the distance attenuation coefficient. 

Moreover, at the beginning period of the store opening, there are few people who know of its existence, so it takes time to broaden its visibility. From the opening date, the store's visibility broadens around at a speed of $k.$ Therefore, we indicate that the store will interact with a circular range less than $kt$ from the store at time $t,$ which is naturally expanding due to the broadening of the store's visibility.
Define the proportion of uncompetitive attraction as
\begin{equation}
	q_t(\mathbf{x}) = \mathds{1}_{\{\Vert\mathbf{x}\Vert\leq kt\}} e^{-\delta\Vert\mathbf{x}\Vert}.
\end{equation}


Furthermore, within a real economy, populations tend to cluster unevenly, and an area around landmarks typically has a higher concentration of people. Despite the intricate and usually non-analytical nature of density functions (depicted as population heatmaps), for the sake of simplicity, we assume uniform foot traffic distribution near store locations, represented by a probability density $u$. This assumption is reasonable as individuals further away from a store location are often less likely to become an active customer, thereby minimizing the influence of density on our results. Additionally, based on the mean value theorem, there exists an equivalent theoretical uniform density function that parallels the unknown and nonlinear actual distribution.


Overall, the potential customers around the store are expressed as
\begin{align}
	N_t &= \iint_{\mathbb{R}^2} q_t(\mathbf{x})u \mathrm{d}\mathbf{x}\\
	&= \iint_{\mathbb{R}^2} \mathds{1}_{\{\Vert\mathbf{x}\Vert\leq kt\}} e^{-\delta\Vert\mathbf{x}\Vert}u \mathrm{d}\mathbf{x}.\label{ll1}
\end{align}
After simple calculations in terms of polar coordinates, we obtain
\begin{align}
	N_t &= \int_0^{2\pi} \int_0^{\infty} \mathds{1}_{\{r\leq kt\}} e^{-\delta r} u r \mathrm{d} r\mathrm{d}\theta\\
	&= 2\pi \frac{u}{\delta^2}(1-(\delta kt+1)e^{-\delta kt}).
\end{align}
This equation states that due to the broadening of visibility, the flow of potential customers increases in the initial stage. After a period of time, the flow of potential customers increases to the upper limit of $2\pi\frac{u}{\delta^2},$ which mainly depends on the foot traffic around the location.

In reality, competition exists not only in terms of style but also in terms of other stores in spatial proximity. An area will usually have some homogeneous or similarly styled stores whose stylistic difference is less than the indistinguishable difference $\epsilon$. \cite{singh2006market} uses a natural experiment to demonstrate spatial competition between stores. In our model, this type of competition is reflected in the dilution of potential customers rather than a decrease in the conversion rate. 
The following provides a more granular modeling of the flow of the potential customers. 

Consider a store that opens at time $t_0=0$ in $\mathbf{x}_0=\mathbf{0}$ coordinates with visibility broadening speed $k_0.$ At moment $t,$ let it be surrounded by $m$ stores with similar styles, with location coordinates $\mathbf{x}_1,\mathbf{x}_2,\ldots,\mathbf{x}_m,$ opening times $t_1,t_2,\ldots,t_m$ and visibility broadening speed $k_1,k_2,\ldots,k_m$ respectively. Similar to (\ref{ll1}), define the proportion of uncompetitive attraction of the $j$-th store to the coordinate $\mathbf{x}$ at time $t$ as
\begin{equation}
	q_{jt}(\mathbf{x}) = \mathds{1}_{\{\Vert\mathbf{x}-\mathbf{x}_j\Vert\leq k_j(t-t_j)\}} e^{-\delta\Vert\mathbf{x}-\mathbf{x}_j\Vert},\  j=0,1,2,\ldots,m.
\end{equation}
Our model assumes that foot traffic in $\mathbf{x}$ coordinates becomes potential customers at the highest attraction proportion of all stores, which is then distributed as potential customers for each store according to the attraction proportion of each store. The proportion of competitive attraction is 
\begin{equation}
	\tilde{q}_{jt}(\mathbf{x}) = \max_{0\leq i \leq m }{q}_{it}(\mathbf{x}) \cdot\frac{{q}_{jt}(\mathbf{x})}{\sum_{i=0}^m {q}_{it}(\mathbf{x})}.
\end{equation}
The following two propositions of $\tilde{q}_{jt}$ show that the presence of spatial competition makes the attraction proportion lower and its continuity with respect to space, respectively
\begin{proposition}
	For any $0\leq j \leq m,t\geq 0$ and $\mathbf{x}\in\mathbb{R}^2,$ $\tilde{q}_{jt}(\mathbf{x})\leq {q}_{jt} (\mathbf{x}).$
\end{proposition}
\begin{proof}
    Note that
    \begin{equation}
    \frac{ \max_{0\leq i \leq m }{q}_{it}(\mathbf{x})}{\sum_{i=0}^m {q}_{it}(\mathbf{x})} \leq 1,
\end{equation}and it naturally follows that the proposition is proven.
\end{proof}
\begin{proposition}
	\label{continu}
	For any $0\leq j \leq m,t\geq 0,$ $\tilde{q}_{jt}$ is continuous on $\cap_{i=0}^{m}\{\mathbf{x}\in\mathbb{R}^2:\Vert \mathbf{x} -\mathbf{x}_i\Vert\leq k_i (t-t_i)\}$ with respect to $\mathbf{x}.$
\end{proposition}
\begin{proof}
Let $c_j$ be function from $\mathbb{R}^{m+1}$ to $\mathbb{R},$ of the form  
    \begin{equation}
        c_j(x_0,x_1,\ldots,x_m) = \max_{0\leq i \leq m }x_i \cdot\frac{x_j}{\sum_{i=0}^m x_i}.
    \end{equation}
Obviously $c_j$ is continuous on $\mathbb{R}^{m+1}.$
Note that for any $0\leq j \leq m,t\geq 0,$ ${q}_{jt}$ is continuous on $\{\mathbf{x}\in\mathbb{R}^2:\Vert \mathbf{x} -\mathbf{x}_j\Vert\leq k_j (t-t_j)\}$ with respect to $\mathbf{x}.$ Thus,
\begin{equation}
    \tilde{q}_{jt}(\mathbf{x}) = c_j({q}_{j0}(\mathbf{x}),{q}_{j1}(\mathbf{x}),\ldots,{q}_{jm}(\mathbf{x}))
\end{equation}
is continuous on $\cap_{i=0}^{m}\{\mathbf{x}\in\mathbb{R}^2:\Vert \mathbf{x} -\mathbf{x}_i\Vert\leq k_i (t-t_i)\}$ with respect to $\mathbf{x}.$
\end{proof}
It can be seen from the proof that the continuity demonstrated by proposition \ref{continu} only exists in area where all stores are visible. That is to say, newly opened stores will cause an unexpected decrease in the attraction proportion of other stores during the process of broadening their visibility.

The flow of the potential customers of the $j$th store is modelled as 
\begin{align}
	N_{jt}&=\iint_{\mathbb{R}^2}  \tilde{q}_{jt}(\mathbf{x})u \mathrm{d}\mathbf{x}\\
	\label{w1}
	&= \iint_{\mathbb{R}^2}  Q_t(\mathbf{x})\mathds{1}_{\{\Vert\mathbf{x}\Vert\leq kt\}} e^{-\delta\Vert\mathbf{x}\Vert}u \mathrm{d}\mathbf{x}
\end{align}
where $Q_t(\mathbf{x})=\frac{\max_{0\leq i \leq m }{q}_{it}(\mathbf{x})}{\sum_{i=0}^m {q}_{it}(\mathbf{x})}\leq 1.$ Comparing (\ref{ll1}) and (\ref{w1}), it can be seen that the presence of other stores in the proximity narrows the attraction proportion of the store for customers, which leads to a decrease of the flow of potential customers.

\begin{example}
   Consider spatial competition between two stores. One store opens at time $0,$ and another store opens at time $t$ located at a certain distance from the first store. 
    Figure \ref{hhh} shows how the first store's flow of potential customers is affected by the new store. (1) and (2) show the case where the new store opens at $t = 150$ and $t = 300$, respectively. The different curves represent the distance of the new store from the first store, where the highest curve represents the situation where no new store is opened. 

\begin{figure}[h]
	\centering
	\begin{minipage}[h]{0.4\textwidth}
		\centering
		\includegraphics[width=\textwidth]{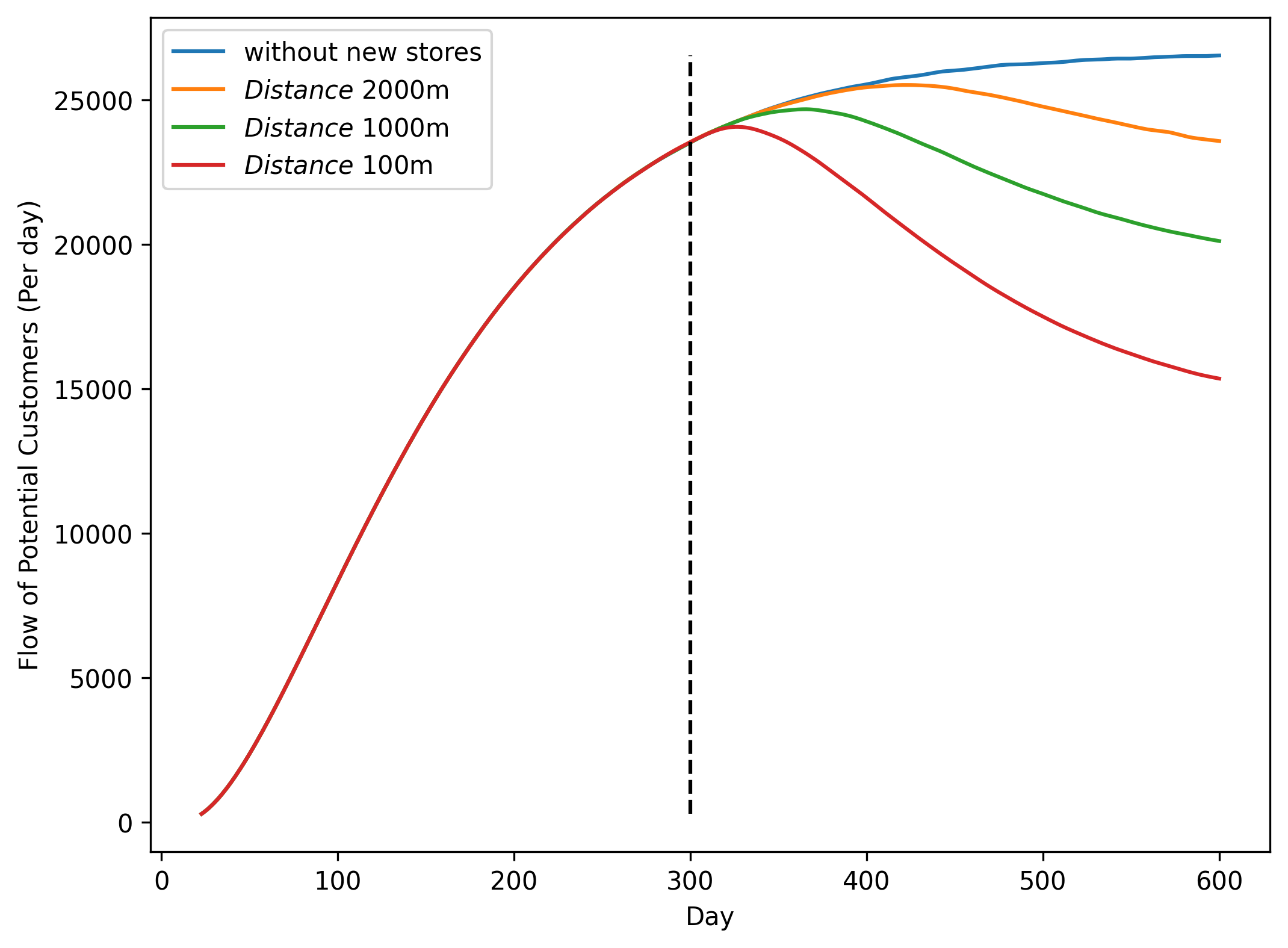}
		\raisebox{1.5ex}{(1)}
		\label{fig:image1}
	\end{minipage}
	\hspace{0.1\textwidth}
	\begin{minipage}[h]{0.4\textwidth}
		\centering
		\includegraphics[width=\textwidth]{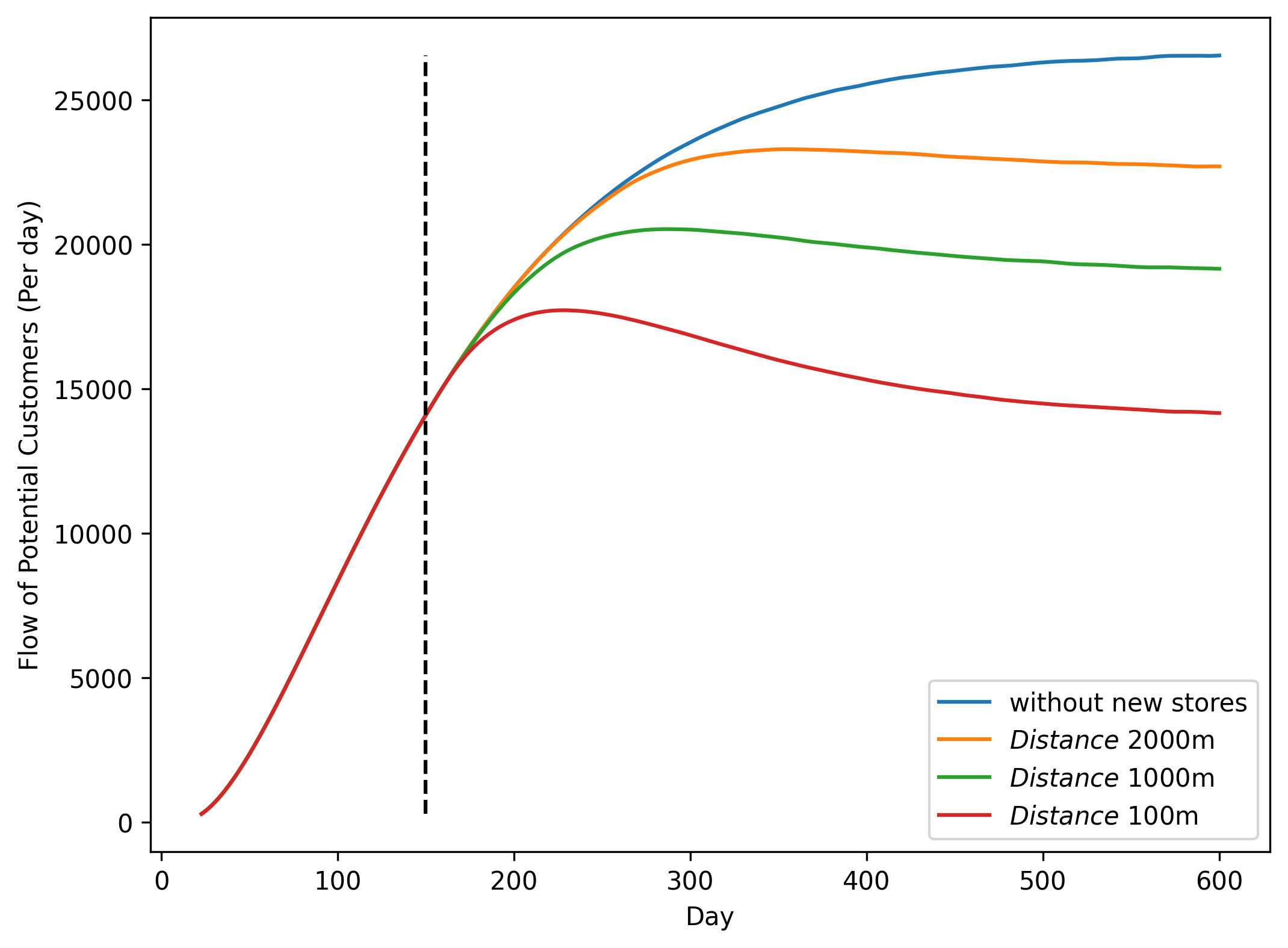}
		\raisebox{1.5ex}{(2)}
		\label{fig:image2}
	\end{minipage}
	\caption{Impact of Spatial Competition}
	\label{hhh}
\end{figure}

\end{example}

As an explicit expression for $N_t$ is unavailable, Figure \ref{hhh} is generated through Monte Carlo simulation (see \cite{metropolis1953equation}). 
This limitation leads to the complexity of calibrating parameters in empirical studies, which gives motivation for us to introduce Theorem \ref{jj} and the competing equivalent foot traffic density.

\begin{theorem}
	\label{jj}
	There exists a competing equivalent foot traffic density $u'$ for a given store (say $0$th store) such that
	\begin{equation}
		\iint_{\mathbb{R}^2}  {q}_{0t}(\mathbf{x})u '\mathrm{d}\mathbf{x} = \iint_{\mathbb{R}^2}  \tilde{q}_{0t}(\mathbf{x})u \mathrm{d}\mathbf{x} + o(1)\ (t\rightarrow \infty)
	\end{equation} holds. And $u'$ can be calculated as
	\begin{equation}
		\label{calcu}
		u' =  u  \frac{\delta^2} {2\pi} \iint_{\mathbb{R}^2} {e^{- \delta\Vert\mathbf{x}-\mathbf{x}_0\Vert-
				\delta d(\mathbf{x})}}\left({\sum_{j=0}^m e^{-\delta\Vert\mathbf{x}-\mathbf{x}_j\Vert}
		}\right)^{-1} \mathrm{d}\mathbf{x},
	\end{equation}
	where $d(\mathbf{x}) = \inf_{1\leq j \leq m} \Vert \mathbf{x}-\mathbf{x}_j\Vert.$
\end{theorem}
\begin{proof}
Notice that when $t\rightarrow \infty$
\begin{equation}
    \tilde{q}_{0t}(\mathbf{x}) \rightarrow \tilde{q}_{0}(\mathbf{x}) = {e^{- \delta\Vert\mathbf{x}-\mathbf{x}_0\Vert-
\delta d(\mathbf{x})}}\left({\sum_{j=0}^m e^{-\delta\Vert\mathbf{x}-\mathbf{x}_j\Vert}
}\right)^{-1},
\end{equation}
for $ \mathbf{x}\in \mathbb{R}^2,$ and 
\begin{equation}
   \iint_{\mathbb{R}^2}  {q}_{0t}(\mathbf{x})\mathrm{d}\mathbf{x}= \frac{2\pi}{\delta^2}(1-(\delta kt+1)e^{-\delta kt}) \rightarrow  \frac{2\pi}{\delta^2}.
\end{equation}
It is sufficient to show
\begin{equation}
    \lim_{t\rightarrow \infty}\iint_{\mathbb{R}^2}  \tilde{q}_{0t}(\mathbf{x})u \mathrm{d}\mathbf{x} =\iint_{\mathbb{R}^2} \lim_{t\rightarrow \infty} \tilde{q}_{0t}(\mathbf{x})u \mathrm{d}\mathbf{x}.
\end{equation}
Noting that \begin{equation}
    \tilde{q}_{0t} \leq {q}_{0t} \leq e^{- \delta\Vert\mathbf{x}-\mathbf{x}_0\Vert},
\end{equation}the proposition holds by the dominated convergence theorem, see \cite{folland1999real}.
\end{proof}

\section{Part IV: Full model of cash flow}
\label{Part IV: Full model of cash flow}
Choose transform function as $\exp(\cdot)$, we obtain the full model of cash flow in uncompetitive case as follows: 
\begin{equation}
	\text{CF}_t = 2\pi \frac{u}{\delta^2}(1-(\delta kt+1)e^{-\delta kt})\cdot\theta(I)\sum_{j=1}^{m}\mathbf{P}^{(j)} \tilde{\beta}_{j0}(I)\exp\left(\mu_j t- \nu_j t^2\right).
\end{equation}
For the case where the spatial competition exists, replace $u$ with competing equivalent foot traffic density $u'$ to get an approximation as
\begin{equation}
	\text{CF}_t \approx 2\pi \frac{u'}{\delta^2}(1-(\delta kt+1)e^{-\delta kt})\cdot\theta(I)\sum_{j=1}^{m}\mathbf{P}^{(j)} \tilde{\beta}_{j0}(I)\exp\left(\mu_j t- \nu_j t^2\right).
\end{equation}
Figure \ref{Cash_Flow_simulation_1} shows a typical cash flow curve in the case of 10 consumer types. Here, $t=1$ is set as one day. The monetary unit of the cash flow $\text{CF}_t$ is Chinese yuan/ day. In Figure \ref{Cash_Flow_simulation_1}, we can see the cash flow has a rainbow-shaped curve. In the early stage, the cash flow of the store gradually increases, mainly due to the increase in the flow of potential customers, which is essentially caused by the visibility broadening. The decrease in cash flow in the later period is due to a decrease in the conversion rate caused by style-shifting.

\begin{figure}[h]
	\centering
	\includegraphics[width = 0.8\textwidth]{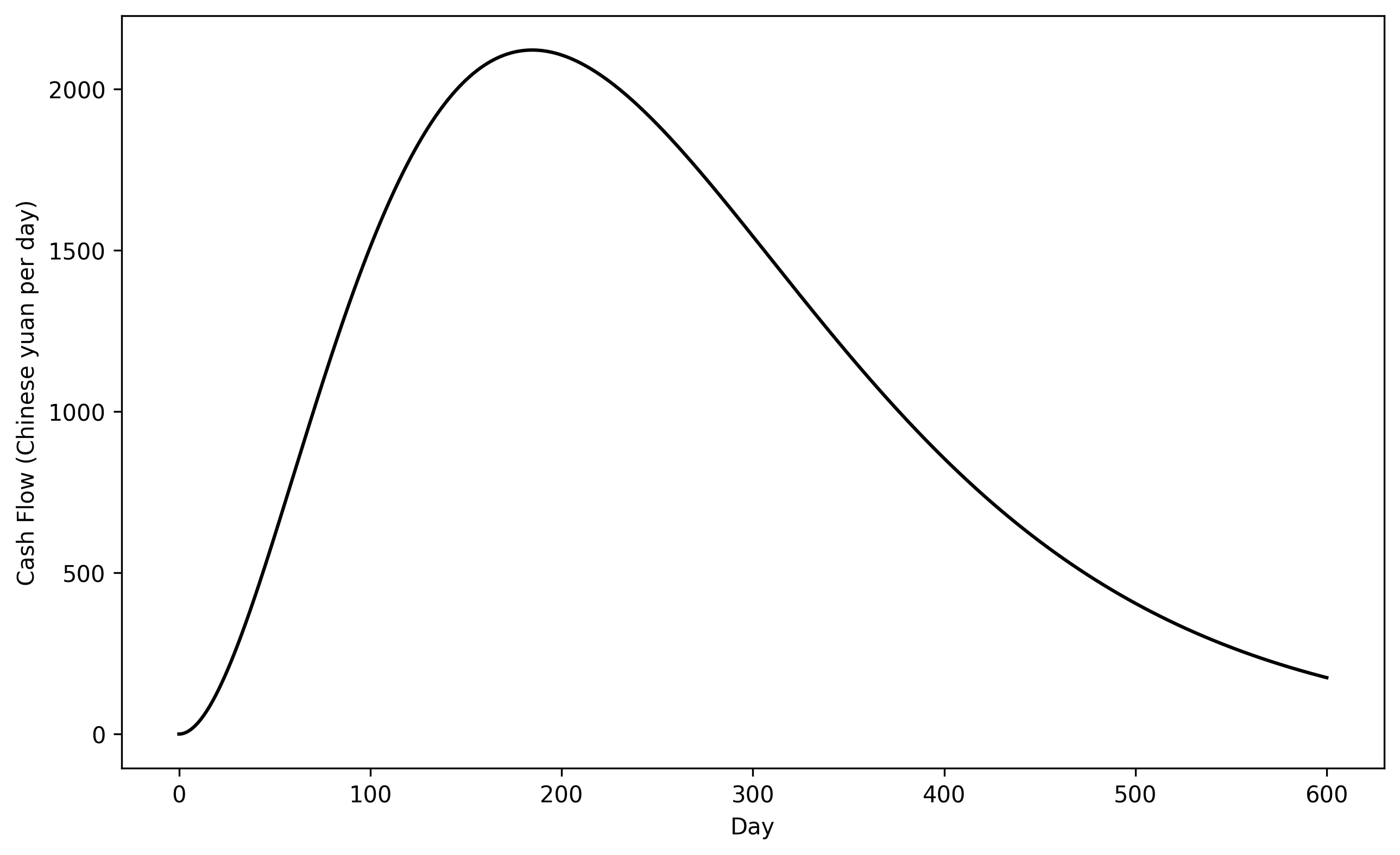}
	\caption{Cash Flow Curve}   
	\label{Cash_Flow_simulation_1}
\end{figure}
Subsequently, we set $m = 1$ and assess curve performance under varying settings: (1) population flow density $u$, (2) decrease coefficient $\nu$, and (3) visibility broadening speed $k$. Our examination centers on two key aspects: the peak of the cash flow's rainbow-shaped curve and the point at which it decreases by 95\% from the peak (referred to as the closing point). The former represents the store's maximum theoretical cash flow, while the latter signifies the duration of the store's lifecycle. Figure \ref{hh} illustrates cash flow trends over time for a representative store utilizing our comprehensive model.

Above all, Figure \ref{hh} signifies the inherent lifecycle of a store has a rainbow-shaped curve. Second, a higher foot traffic density $u$ results in a larger peak cash flow yet has no bearing on the store's lifecycle. This outcome challenges conventional beliefs that a store situated in an area with low foot traffic experiences a shorter lifecycle. But our results find that each type of stores has its own viability with a response to different environments. 

Third, the decrease coefficient $\nu$ impacts both the lifecycle and the peak of the cash flow. A faster preference shift results in a shorter lifecycle and a lower peak of the cash flow. It is explained that for a store located in a shopping mall, the operational risk is larger and the ACP is also higher; Otherwise, the same investments cannot yield a similar shape.

Notably, the visibility broadening speed $k$ minimally affects the lifecycle but significantly influences peak performances. A higher $k$ leads to an earlier and higher peak, indicating the significance of the opening ceremony in the long-term performance of a store. In conclusion, our model serves as a robust tool for evaluating store performance.

\begin{figure}[h]
	\centering
	\begin{minipage}[b]{0.3\textwidth}
		\centering
		\includegraphics[width=\textwidth]{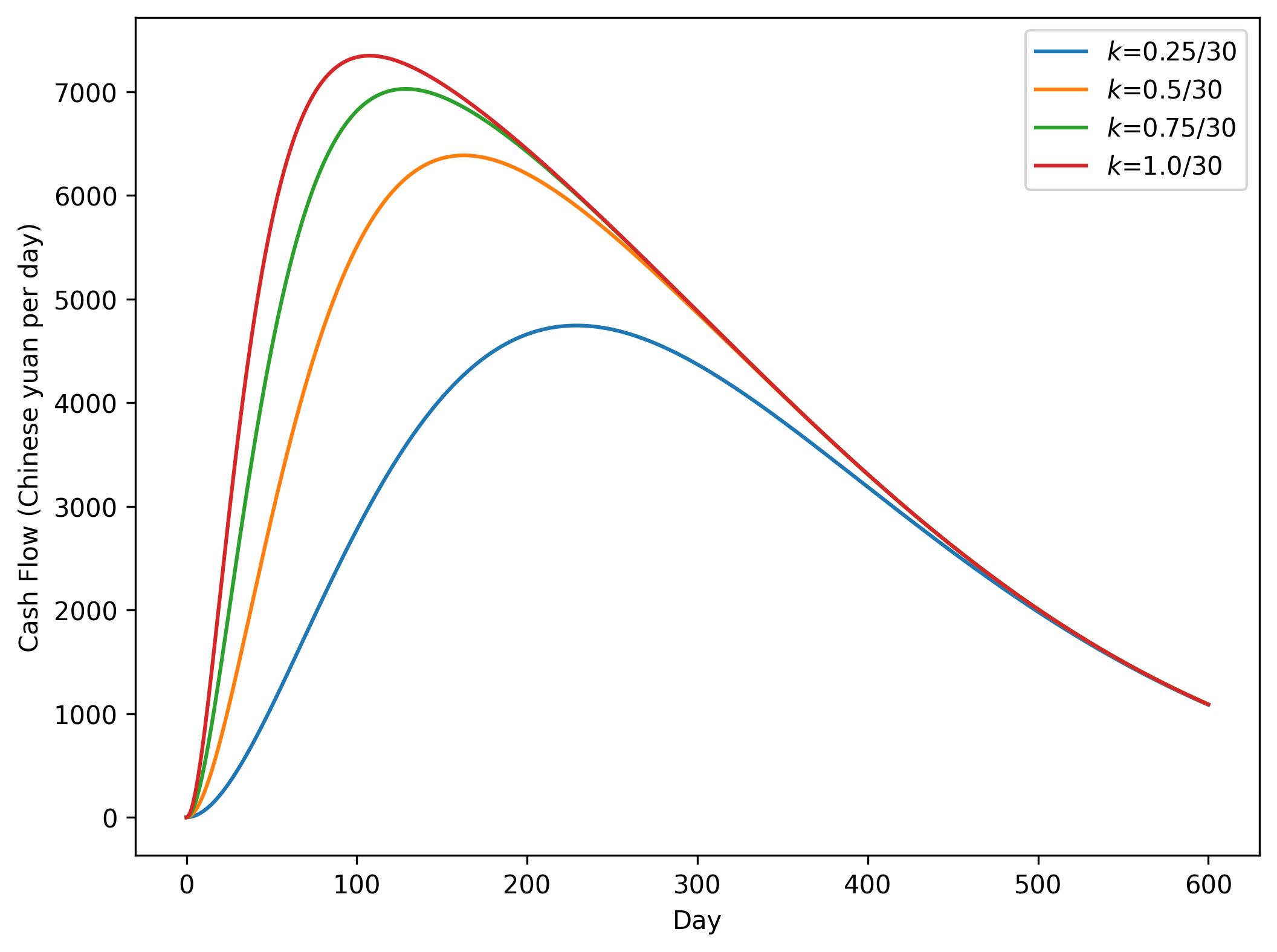}
		\raisebox{1.5ex}{(1)}
		\label{fig:image1}
	\end{minipage}
	\hfill 
	\begin{minipage}[b]{0.3\textwidth}
		\centering
		\includegraphics[width=\textwidth]{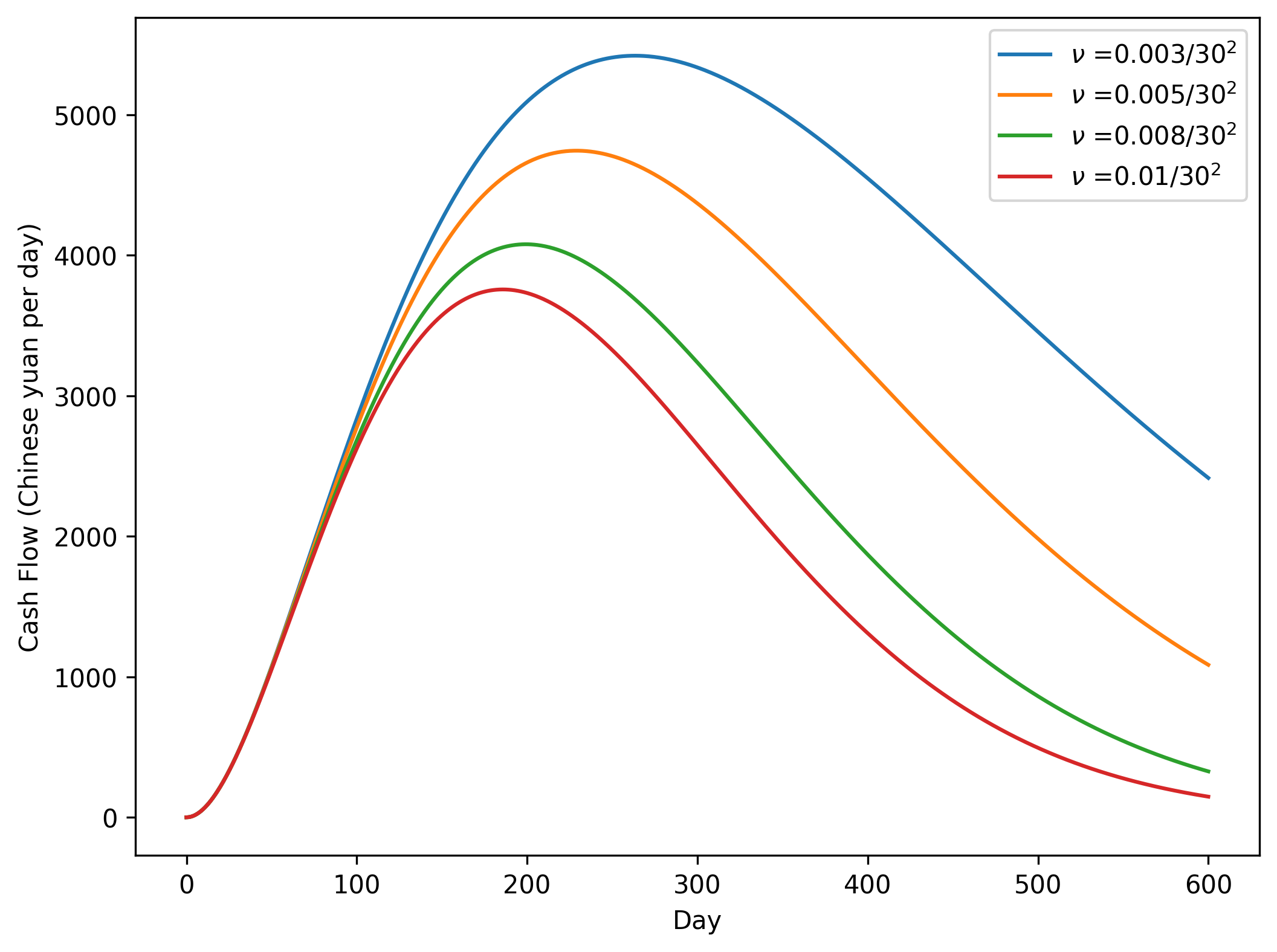}
		\raisebox{1.5ex}{(2)}
		\label{fig:image2}
	\end{minipage}
	\hfill 
	\begin{minipage}[b]{0.3\textwidth}
		\centering
		\includegraphics[width=\textwidth]{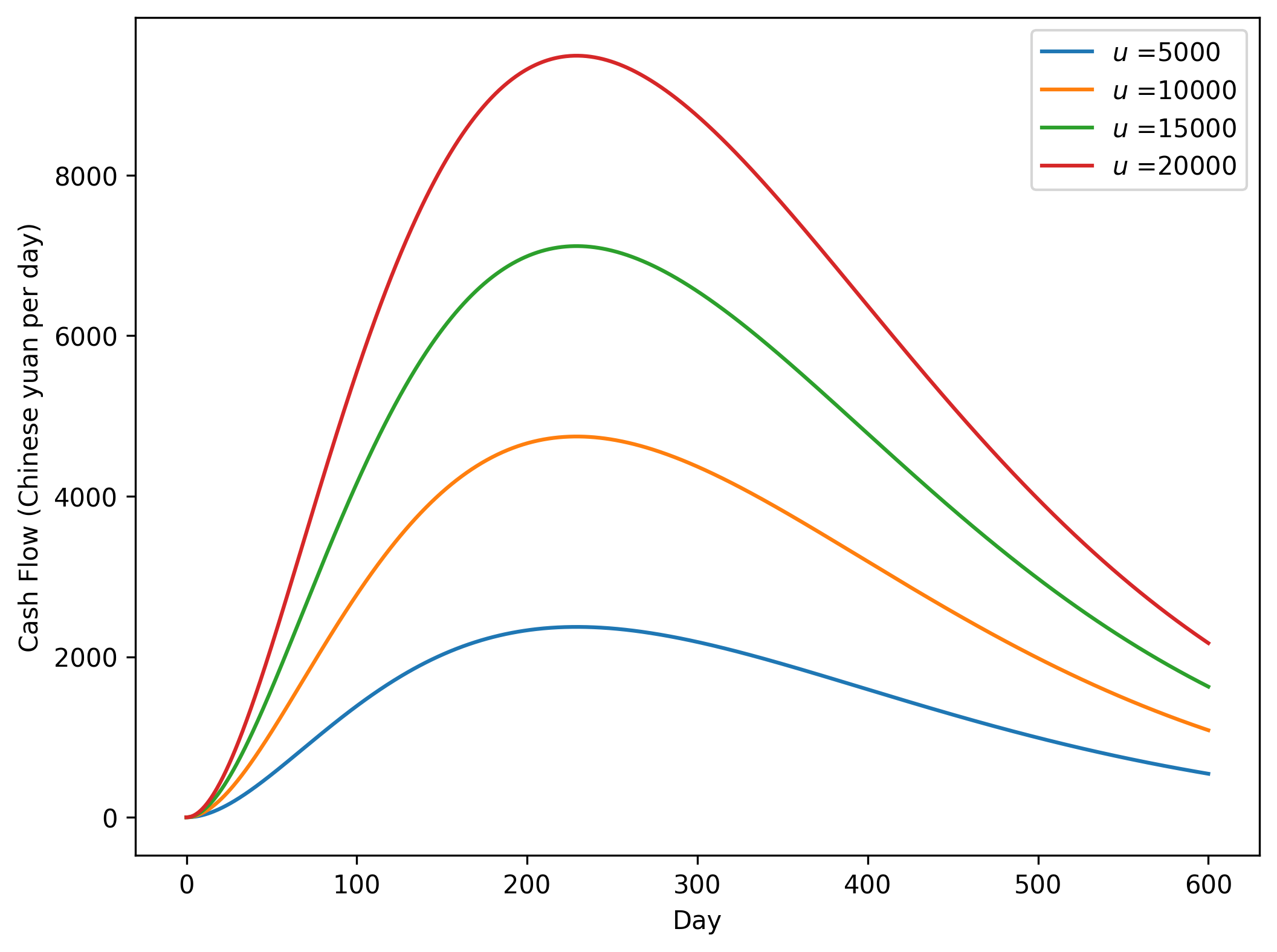}
		\raisebox{1.5ex}{(3)}
		\label{fig:image3}
	\end{minipage}
	\caption{Cash Flow with Varying Parameters}
	\label{hh}
\end{figure}

\section{Empirical Evidence}
\label{Empirical Evidence}
\
In this section, for simplicity in empirical research, we consider that there is only one type of people in the population structure, i.e. $m=1$. In actuality, the methods and conclusions can be extended to more complex population structures. Moreover, spatial competition is also considered, so we use the competitive equivalent foot traffic density. Under this setting, we use the following model for empirical research: 
\begin{equation}
	\text{CF}_t = 2\pi \frac{u'}{\delta^2}(1-(\delta kt+1)e^{-\delta kt})\cdot\theta(I) \tilde{\beta}_{0}(I)\exp\left(-\nu t^2\right)+\epsilon_t,
	\label{formula_cf_empirical}
\end{equation}
where $\epsilon_t$ is the random residual term. We denote (\ref{formula_cf_empirical}) as $\text{CF}_t=h(t:\delta,u',\theta,k,\nu,\beta_0)+\epsilon_t.$ In this model, $t=1$ equals one day, and the monetary unit of $\text{CF}_t$ is Chinese yuan/day.

\subsection{Cash flow data and basic information of stores}
\begin{table}[h]
	\centering
	\caption{The summary statistics.\\
		This table reports the summary statistics for the daily cash flow of three representative stores from three industries in China's real economics: the retail industry (store A), the restaurant industry (store B), and the service industry (store C). Panel A displays the descriptive statistics of cash flow, while Panel B shows basic information. In particular, the ramp-up period comprises the dates that cash flow increases from zero to the optimal cash flow.
	}
	\label{tab_summary}
	\scalebox{1.0}{
		\begin{tabular}{cccc}
			\hline
			\multicolumn{4}{c}{Panel A: The summary statistics for cash flow }\\
			\hline
			& Store A & Store   B  & Store C  \\
			\hline
			Mean                                           &     $4236.45	$    &      $1141.75$      &     $767.75$     \\
			STD                                            &    $1648.13$     &      $311.34$      &      $757.89$    \\
			Min                                            &    $12.00$     &     $110.00$       &     $25.00$\\
			Max                                            &     $10182.51$    &     $2958.00$       &      $5454.90$\\
			Skewness                                       &     $-0.04$    &         $1.49$   &      $2.70$\\
			Kuitosis                                       &    $0.15$     &      $5.18$      &     $10.76$     \\
			AR(1)                                          &     $0.84$    &      $0.53$      &    $-0.02$      \\
			Opening date                                   &    2022-01-11     &       2022-01-10     &    2022-08-30      \\
			Closing date                                   &   2023-09-19      &       2023-07-25     &      2023-10-25    \\
			Ramp up period                                 &    $325$     &   $22$         &      $77$    \\
			OBS                                            &     $616$    &    $556$        &      $157$    \\
			\hline
			&         &            &          \\
			\hline
			\multicolumn{4}{c}{Panel B: The basic information} \\
			\hline
			& Store A & Store B    & Store C  \\
			\hline
			ACP                                            &    $19.24$     &      $17.49$       &   $81.16$      \\
			Industry& Retail  & Restaurant & Service  \\
			Location                                       & Yuxi, Yunnan Province   &Yuxi, Yunnan Province          & Yuxi, Yunnan Province         \\
			\hline
	\end{tabular}}
\end{table}

First of all, we collect the daily cash flow data of three stores from three representative industries of China's real economics: the retail industry (store A), the restaurant industry (store B) and the service industry (store C) \footnote{In fact, we collects the order-level data and sum them into daily cash flow.}. The data is authorized for use by the store owners. Note that the stores are usually closed on Chinese holidays (such as Lunar New Year). Therefore, we interpolate the missing value with the mean of two cash flow values on the same weekday of the preceding and subsequent weeks. Table \ref{tab_summary} displays the summary statistics for these three stores. The average cash flow of store A is 4236.45 Chinese yuan with a standard deviation of 1648.13 Chinese yuan, while that of store C is 767.75 Chinese yuan with a standard deviation of 757.89 Chinese yuan. It shows that cash flow from the service industry is more volatile than that of the retail industry, which is expected because the commodities of the retail industry are generally the same throughout and the demands of consumers are stable. However, the ACP of store A (19.24 Chinese yuan) is lower than that of store C (81.16 Chinese yuan). This result suggests that the commodity of the service industry is lower frequency but higher gross profit rate.


Furthermore, through real-life experiences, we calibrate an important parameter: the distance attenuation coefficient $\delta$. Not many people are willing to travel more than three kilometers (km) away to make a purchase\footnote{For example, see https://www.canyin88.com/zixun/2019/07/05/73882.html for more details.}. Therefore, the distance attenuation coefficient $\delta$ is set as $\frac13\ln\frac{1}{0.01}\approx 1.535$ i.e. $\hat{\delta} = 1.535$, which means the attraction of the store to the foot traffic decreases to less than $1\%$ at a distance of $3$km.

Estimation of traffic density by using Theorem \ref{jj}
We collect traffic heat-maps from \href{https://map.baidu.com/}{Baidu} map, including the absolute values of the population per 100 $m^2$ in Yuxi City, Yunnan Province for two years, 2022 and 2023. By weighted averaging, we obtain the absolute density of the population within a three-kilometer radius of each store. In order to get a more accurate idea of the consumer density near each store, we also collect the number of comparable stores that belong to the same category (or sub-industry) with our sample stores from \href{https://www.dianping.com/}{Dianping}, a China's online platform where the public writes reviews about different businesses. However, the ``similar style" means that the distance on the style vector of stores is sufficiently small, which does not equate to the ``same category". To avoid data manipulation, we simply use the ``same category" to substitute ``similar styles", a concept that is defined in our paper, resulting in some problems. We will provide further discussion in \ref{Additional Discussions}.

Therefore, we use the formula \ref{jj} to estimate the competitive equivalent $\hat{u}'$. Table \ref{tab_empirical_parameters} reports the specific county, category, number of competitors, and real foot traffic density, as well as the competitive equivalent foot traffic density estimated by the formula \eqref{jj}. First, store A is in Xinping County, while stores B and C are in Hongta District, showing reasonable results that the population density (above 50000) in the district is larger than that (6246.85) in the county. Second, only two stores exist in the same category as store C (from the service industry), while 454 stores exist that are similar to store B (from the restaurant industry). The results suggest that the demand of the service industry is relatively lower than that of the rigid demand of the restaurant industry. Finally, taking the competition into account, the competitive equivalent foot traffic density of store A in the county is surprisingly highest (3470.17). This means that not only the absolute population but also the intensity of competition matters. Overall, our paper provides a strong tool to consider both the distribution of populations and the intensity of competition.


\begin{table}[h]
	\centering
	\caption{The calibration of foot traffic competitive equivalent density $u'$.\\
		This table reports the estimation results of traffic density by the following formula which is derived in Theorem \ref{jj},\\
		\begin{minipage}{\textwidth}
			\[u' =  u  \frac{\delta^2} {2\pi} \iint_{\mathbb{R}^2} {e^{- \delta\Vert\mathbf{x}-\mathbf{x}_0\Vert-
					\delta d(\mathbf{x})}}\left({\sum_{j=0}^m e^{-\delta\Vert\mathbf{x}-\mathbf{x}_j\Vert}
			}\right)^{-1} \mathrm{d}\mathbf{x},\]
		\end{minipage}
		where $d(\mathbf{x}) = \inf_{1\leq j \leq m} \Vert \mathbf{x}-\mathbf{x}_j\Vert.$
		In particular, we present the specific county of store A, store B, and store C as well. Furthermore, this table shows the number of competitors in the same category (or sub-industry) of the same county, where the categories are defined as the same as that in the \href{https://www.dianping.com/}{Dianping} website, the well-known China’s online platform where the public reviews stores and businesses. Furthermore, this table shows the average real foot traffic density from both 2022 and 2023 combined. The population density data stems from the \href{https://map.baidu.com/}{Baidu} map, the famous online map provider in China. Our sample period spans from January 2022 to December 2023.
	}
	\label{tab_empirical_parameters}
	\begin{tabular}{>{\centering\arraybackslash}
			m{1.5cm}
			>{\centering\arraybackslash}m{2.5cm}
			>{\centering\arraybackslash}m{2.5cm}
			>{\centering\arraybackslash}m{2cm}
			>{\centering\arraybackslash}m{2cm}
			>{\centering\arraybackslash}m{4cm}}
		\hline
		\multirow{2}{*}{{Store}}& \multirow{2}{*}{{County}}& \multirow{2}{*}{{Category}}& \multirow{2}{*}{{Competitors}}& \multicolumn{2}{c}{{Foot traffic density}} \\
		\cline{5-6}
		&                 &        &                              & {Real data $\hat{u}$}& {Competitive equivalent $\hat{u}'$}\\
		\hline
		Store A        & Xinping County  & Convenience & 2        & 6246.85    & 3470.17 \\
		Store B        & Hongta District & Simple food          & 454        & 59538.79     & 692.41 \\
		Store C        & Hongta District & Barber          & 289          & 95950.60       & 3014.02 \\
		\hline
	\end{tabular}
\end{table}



\subsection{The model speculation: non-linear least squares regression }

Here our main objective is to estimate and test the three parameters (visibility broadening speed $k$, initial conversion rate $\beta_0,$ and decrease coefficient $\nu$). In doing so we observe on the one hand whether there is a significant increase and decrease in the cash flow of the store over its life cycle, and on the other hand to estimate the true initial purchase probability and the theoretical length of the store's life cycle. We use a special case of Generalized Method of Moments estimation (GMM): Nonlinear Least Squares regression (NLS) to estimate the parameters $(k,$  $\beta_0,$  $\nu)$ using the cash flow data from each store, see \cite{hansen1982large} and \cite{ruckstuhl2010introduction}. The estimator is formulated as 
\begin{equation}
	(\hat{k},\hat{\nu},\hat{\beta_0}) = \underset{(k,\nu,\beta_0)}{\arg\max} \sum_{t=1}^T\left\Vert \text{CF}_t-h(t:\hat{\delta},\hat{u}',\bar{\theta},k,\nu,\beta_0)\right\Vert^2.
\end{equation}
Considering the heteroskedasticity and autocorrelation of the residuals, we report the \cite{newey1986simple} adjusted standard errors. The \cite{newey1986simple} adjusted estimated variance is 
\begin{equation}
	\widehat{\text{\rm{Var}}} \left(\hat{k},\hat{\nu},\hat{\beta_0}\right) = T\left(J' J\right)^{-1} \cdot\hat{\Omega}\cdot\left(J' J\right)^{-1},
\end{equation}
where
\begin{equation}
	J = \left.\left(\begin{array}{ccc}
		\frac{\partial h(1)}{\partial k} & \ldots & \frac{\partial h(T)}{\partial k} \\
		\frac{\partial h(1)}{\partial \nu} & \ldots & \frac{\partial h(T)}{\partial \nu} \\
		\frac{\partial h(1)}{\partial \beta_0} & \ldots & \frac{\partial h(T)}{\partial \beta_0} \\
	\end{array}\right)\right|_{\left(\hat{k}, \hat{\nu}, \hat{\beta_0}\right)}
\end{equation}
and 
\begin{equation}
	\hat{\Omega} = \frac1T\left(\sum_{t=1}^T\hat{\epsilon}_t^2J_tJ_t'+\sum_{l=1}^L\sum_{t=l+1}^T \omega_l \hat{\epsilon}_t\hat{\epsilon}_{t-l}(J_tJ_{t-l}'+J_{t-l}
	J_t')\right),\ \omega_l = 1-\frac{l}{L+1},
\end{equation}
where $\hat{\epsilon}$ is the regression residual and the maximum autocorrelation order $L$ is set as $7.$



\subsection{Empirical results}

\begin{table}[ht]
	\centering
	\captionof{table}{The Non-linear Least Squares regression (NLS) results in daily frequency.\\
		This table reports the NLS results of the following model:\\
		\begin{minipage}{\textwidth}
			\[\text{CF}_t = 2\pi \frac{u'}{\delta^2}(1-(\delta kt+1)e^{-\delta kt})\cdot\theta(I) \tilde{\beta}_{0}(I)\exp\left(-\nu t^2\right)+\epsilon_t,\]
		\end{minipage}
		where $\epsilon_t$ is the random residual term. 
		We set $t=1$ as one day, and the monetary unit of $\text{CF}_t$ is Chinese yuan/day. Furthermore, we present three representative stores distributed in three industries of China's real economics: the retail industry (model 1), the restaurant industry (model 2), and the service industry (model 3). Moreover, we present the estimates of $k$, $\nu$, and $\beta_0$ by times $10^6$, $10^6$, $10^2$, respectively, because the scale is too small in daily frequency. The standard error adjusted by \cite{newey1986simple} is enclosed in parentheses. The sample period spans from January 2022 to September 2023, during which China's digital business rapidly grows.
	}
	
	\label{tab_regression_results}
	\scalebox{1.0}{
		\begin{tabular}{cccc}
			\hline
			\multicolumn{4}{c}{Panel A: NLS results}\\
			\hline
			& Model 1: Store A            & Model 2: Store B & Model 3: Store C  \\
			\hline
			$k(\times 10^2)$         & 2.50&         24.65&       9.41\\
			& (0.26)&         (15.64)&           (2.57)\\
			
			$\nu (\times 10^6)$         & 2.88&         0.83&       57.11\\
			& (0.35)&         (0.22)&           (12.40)\\
			
			$\beta_0 (\times 10^2)$  & 3.59&          3.46&       0.17\\
			& (0.14)&          (0.13)&       (0.02)\\
			&                           &                   &                 \\
			OBS        & 616&            556&       139\\
			F-statistics          & 244.71&         53.73&     63.46\\
			Adj. $R^2$    & 44.40\%&         16.15\%&          29.59\%\\
			\hline
			&                           &                   &                 \\
			\hline
			\multicolumn{4}{c}{Panel B: The Initial conversion rate and the theoretical life span}\\
			\hline
			& Store A            & Store B           & Store C  \\
			\hline
			The initial conversion rate           &            3.59\%&           3.46\%&        0.17\%\\
			Theoretical life span (in days)           &            1043&           1896&        233\\
			\hline
	\end{tabular}}
\end{table}

Panel A of Table \ref{tab_regression_results} shows the NLS results of (\ref{formula_cf_empirical}). First, the decrease coefficients of store A, store B, and store C are $2.88 \times 10^{-6}$ (with standard error = $0.15 \times 10^{-6}$ ), $0.83 \times 10^6$ (with standard error = $0.22 \times 10^{-6}$ ) and $57.11 \times 10^{-6}$ (with standard error = $12.40 \times 10^{-6}$ ), respectively. It means that the decrease coefficients are all statistically significant at a 1\% confidence level. The decreased speed of cash flow is shown as an exponential form, meaning that the decreased speed will increase as time goes by. For instance, it means that the cash flow of store C decreases above 1\% after 68 days ( $exp(- 57.11\times10^{-6} \times 68^2) - exp(- 57.11\times10^{-6} \times 67^2)$). This parameter is the key factor that impacts the life cycle of stores. More importantly, the results in Panel B of Table \ref{tab_regression_results} means that store C from the service industry has the longer theoretical life span (233 days), while store B from the restaurant industry has the shortest life span (1896 days). Second, store A from the retail industry has the smallest visibility broadening speed ($2.50 \times 10^{-2}$ with a standard error =$0.26 \times 10^{-2}$), while store B from the restaurant industry has the largest ($24.65 \times 10^{-2}$ with a standard error =$15.64 \times 10^{-2}$), indicating that our model performs a significant estimate at a 1\% confidence level. It suggests that the people around the store know quickly the restaurant store, possibly because of its larger monetary investment in its opening activities. Finally, the initial conversion rate of store A from the retail industry is the largest (3.59\% with a standard error = 0.14\%), that of store C from the service industry is the smallest (0.17\% with standard error = 0.02\%), and that of store B from restaurant industry is the middle one (3.46\% with a standard error =0.01\%). It shows that the parameters are all significant at 1\% confidence level, corresponding with the real-life experience. In our daily lives, product difference among retail stores is ignored, while the service store highly relies on the service and technique of the owners. It leads to the reasonable conclusion that store A from retail owns the highest initial conversion rate, while store C owns the lowest initial conversion rate.

Overall, we provide a strong empirical tool to evaluate the key parameters of the stores, including the competition-adjusted potential consumers $\hat{u}'$, the visibility broadening speed $k$, the initial conversion rate $\beta_0$ and the decrease coefficient $\nu$. These parameters capture the life cycle of the stores.

\subsection{Visualizations of cash flow}
Figure \ref{Fitted} shows the raw cash flow and the fitted values of our model. Intuitively, our model performs well in capturing the trend of increase and decrease in cash flow. More precisely, different industries have different rules in terms of frequency. First, the cash flow of the retail store C fits well in all three frequencies, suggesting that the consumer may not have a special preference for one specific store, meaning that the cash flow of the retail store is nearly independently and identically distributed. However, the restaurant store B performs well in weekly frequencies, while the service store C fits well in monthly frequencies. The results may indicate that these two industries have a strong seasonality. For example, we may change our taste every week, while we cut our hair monthly.

\begin{figure}[ht]
	\centering
	\begin{minipage}{.3\textwidth}
		\centering
		\includegraphics[width=\linewidth]{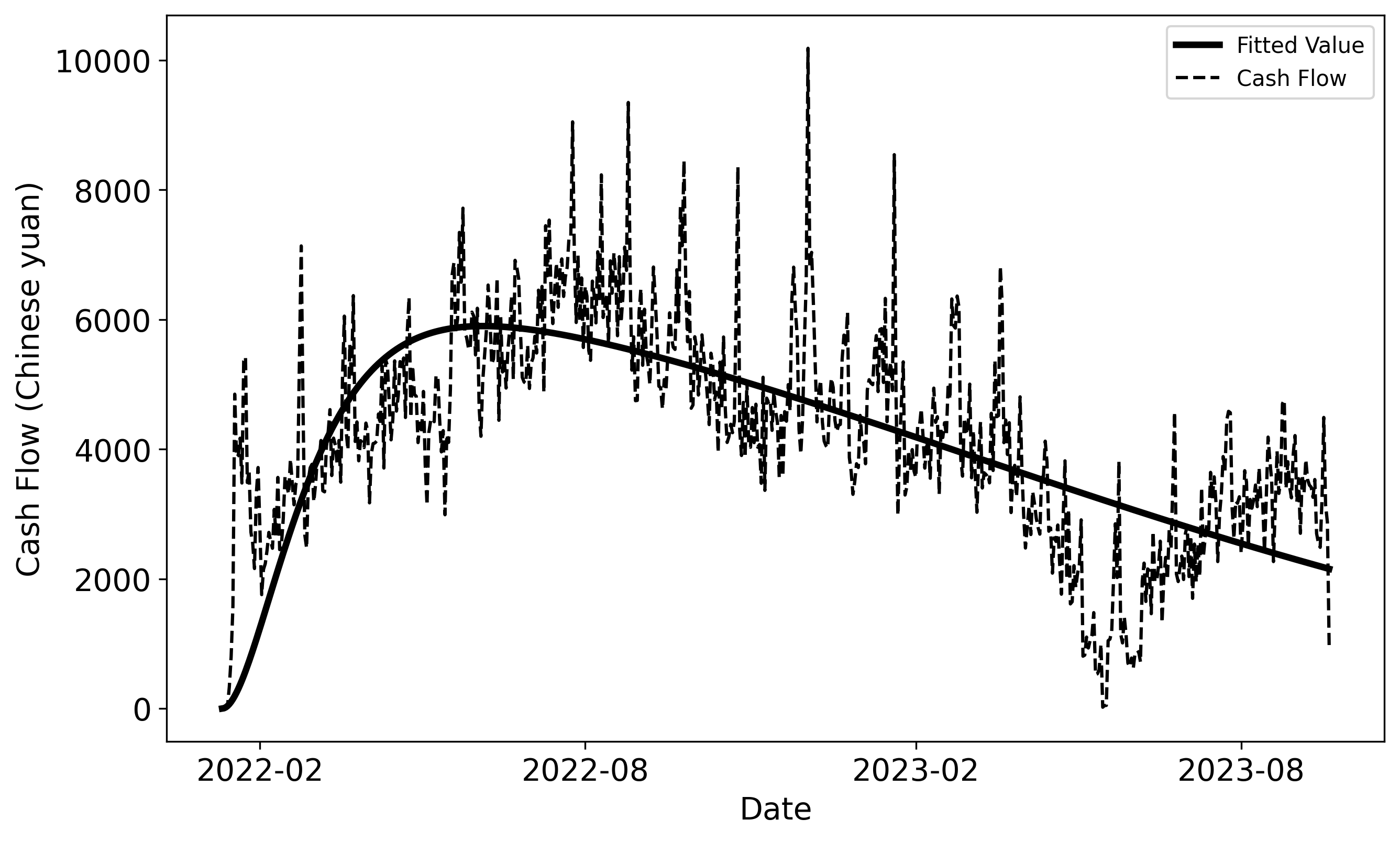}
		\raisebox{1.5ex}{Store A: Daily}
		\label{fig:test1}
	\end{minipage}%
	\hfill 
	\begin{minipage}{.3\textwidth}
		\centering
		\includegraphics[width=\linewidth]{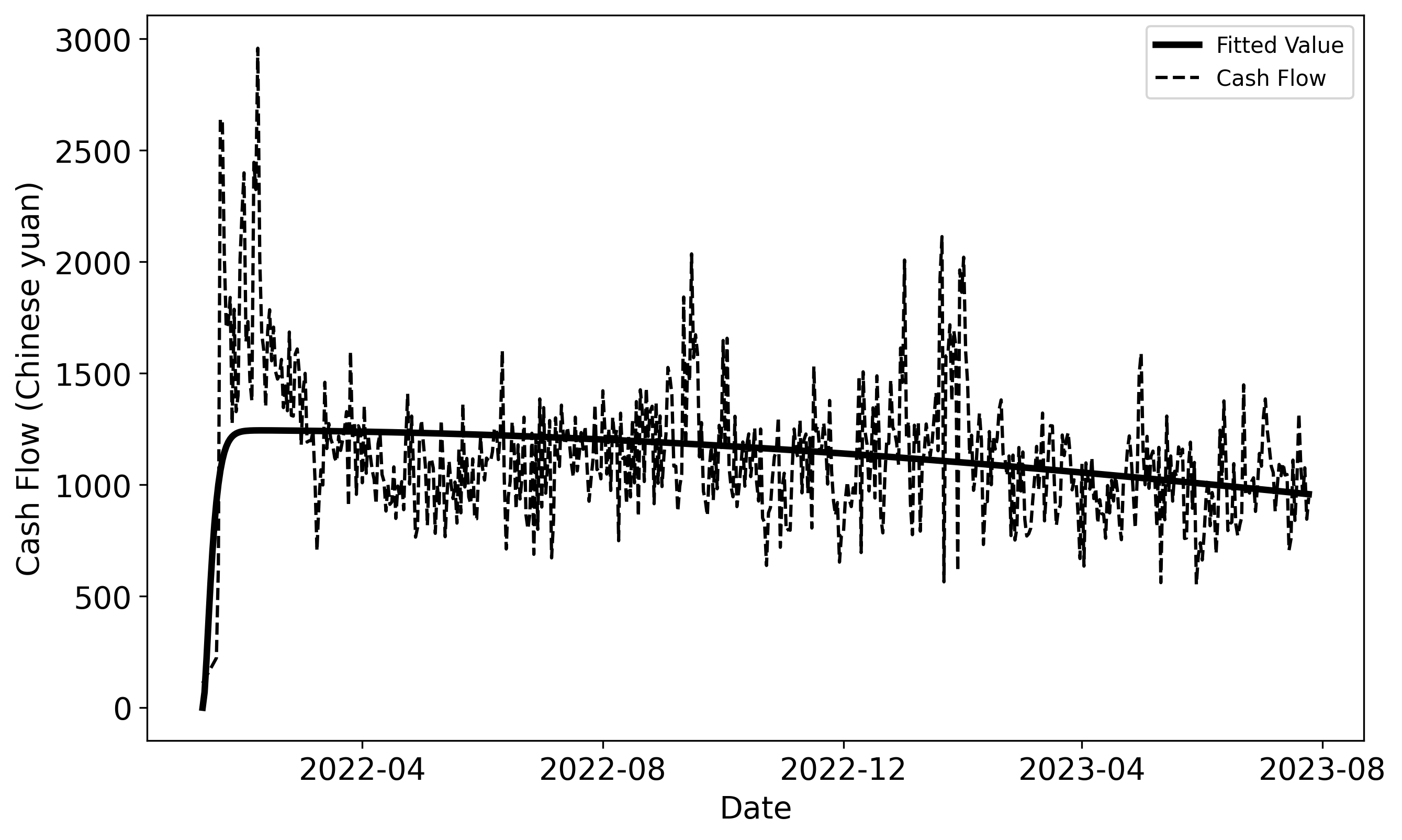}
		\raisebox{1.5ex}{Store B: Daily}
		\label{fig:test2}
	\end{minipage}%
	\hfill 
	\begin{minipage}{.3\textwidth}
		\centering
		\includegraphics[width=\linewidth]{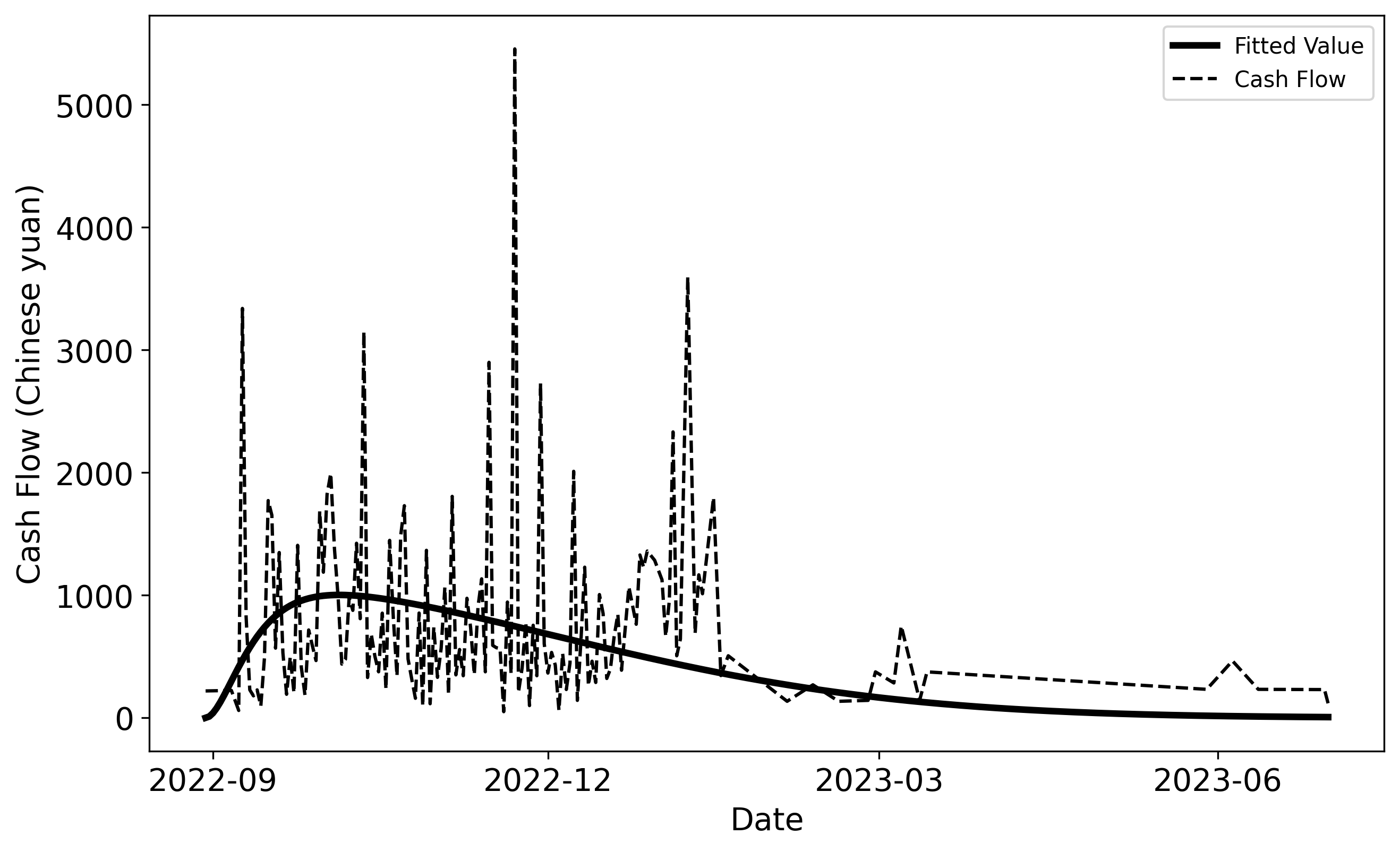}
		\raisebox{1.5ex}{Store C: Daily}
		\label{fig:test3}
	\end{minipage}
	
	\vspace{1em} 
	
	\begin{minipage}{.3\textwidth}
		\centering
		\includegraphics[width=\linewidth]{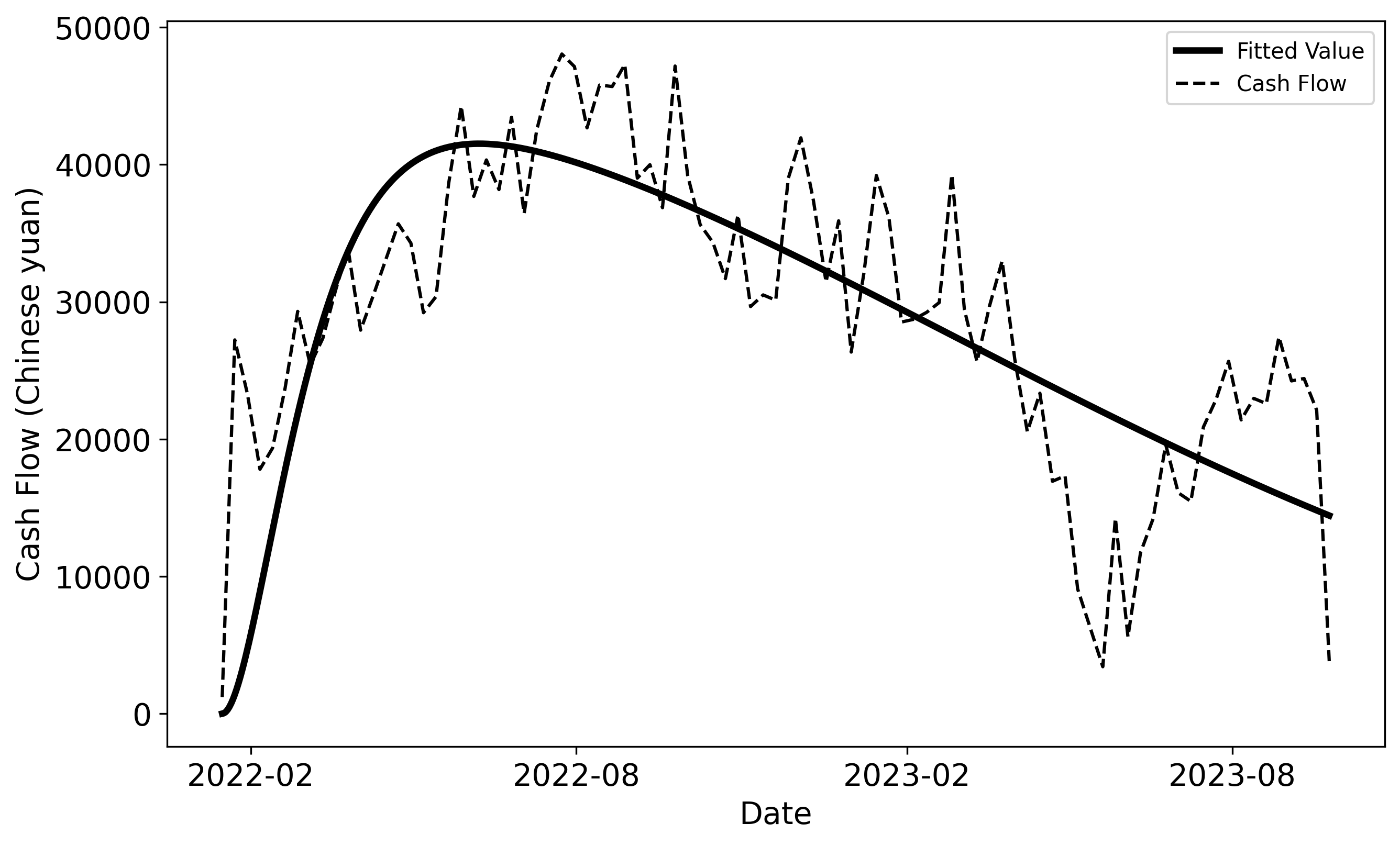}
		\raisebox{1.5ex}{Store A: Weekly}
		\label{fig:test4}
	\end{minipage}%
	\hfill 
	\begin{minipage}{.3\textwidth}
		\centering
		\includegraphics[width=\linewidth]{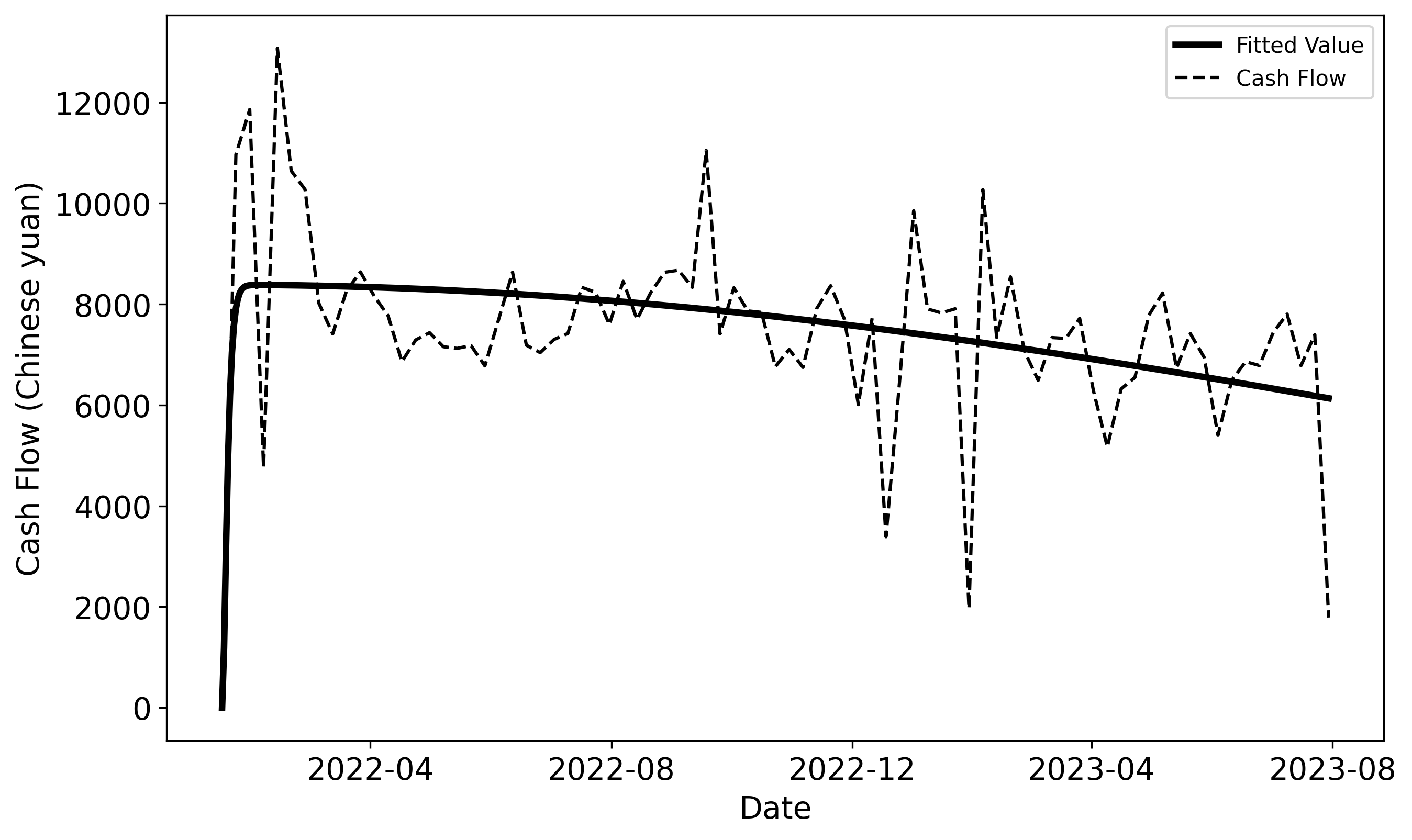}
		\raisebox{1.5ex}{Store B: Weekly}
		\label{fig:test5}
	\end{minipage}%
	\hfill 
	\begin{minipage}{.3\textwidth}
		\centering
		\includegraphics[width=\linewidth]{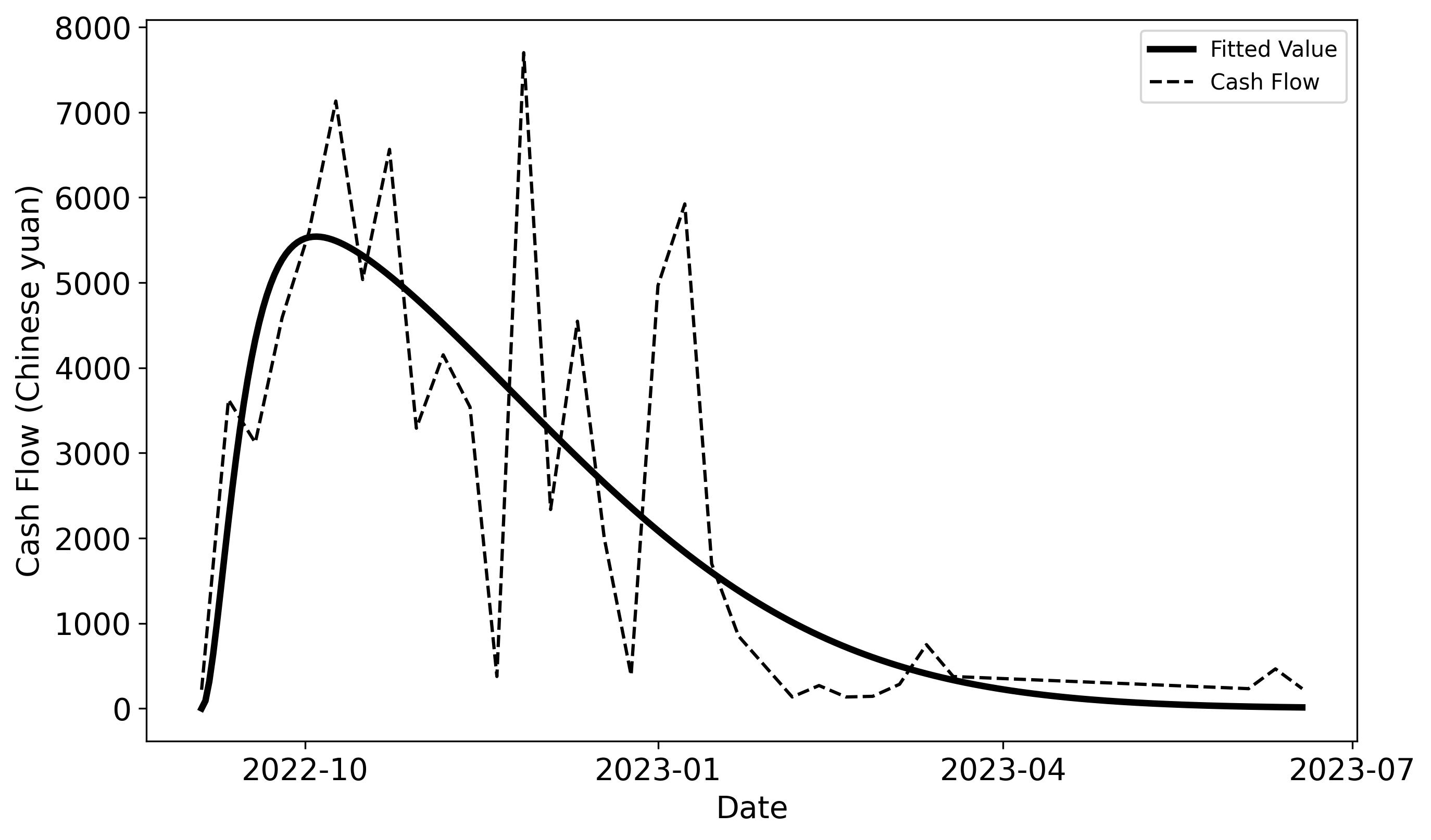}
		\raisebox{1.5ex}{Store C: Weekly}
		\label{fig:test6}
	\end{minipage}
	
	\vspace{1em} 
	
	\begin{minipage}{.3\textwidth}
		\centering
		\includegraphics[width=\linewidth]{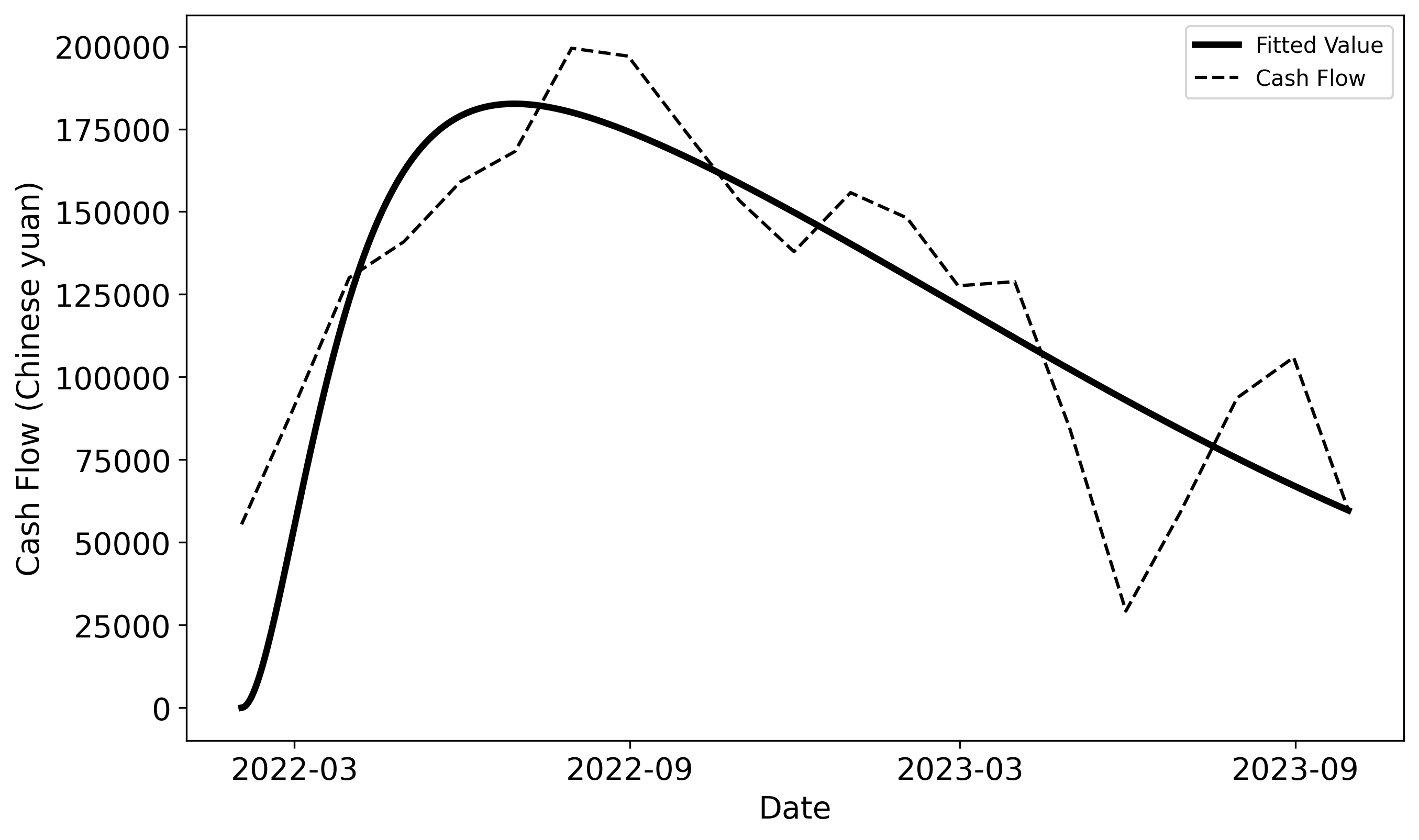}
		\raisebox{1.5ex}{Store C: Monthly}
		\label{fig:test7}
	\end{minipage}%
	\hfill 
	\begin{minipage}{.3\textwidth}
		\centering
		\includegraphics[width=\linewidth]{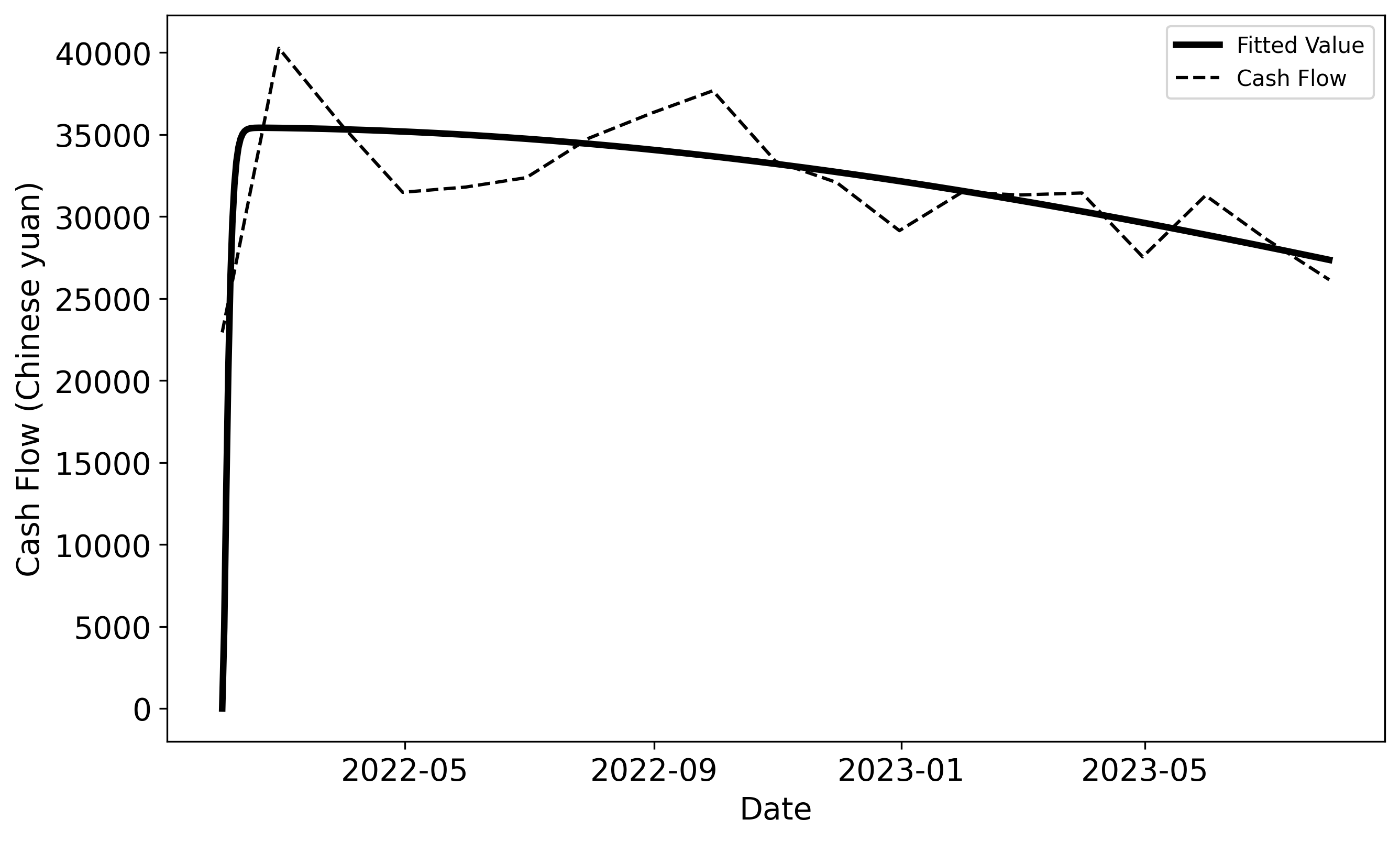}
		\raisebox{1.5ex}{Store C: Monthly}
		\label{fig:test8}
	\end{minipage}%
	\hfill 
	\begin{minipage}{.3\textwidth}
		\centering
		\includegraphics[width=\linewidth]{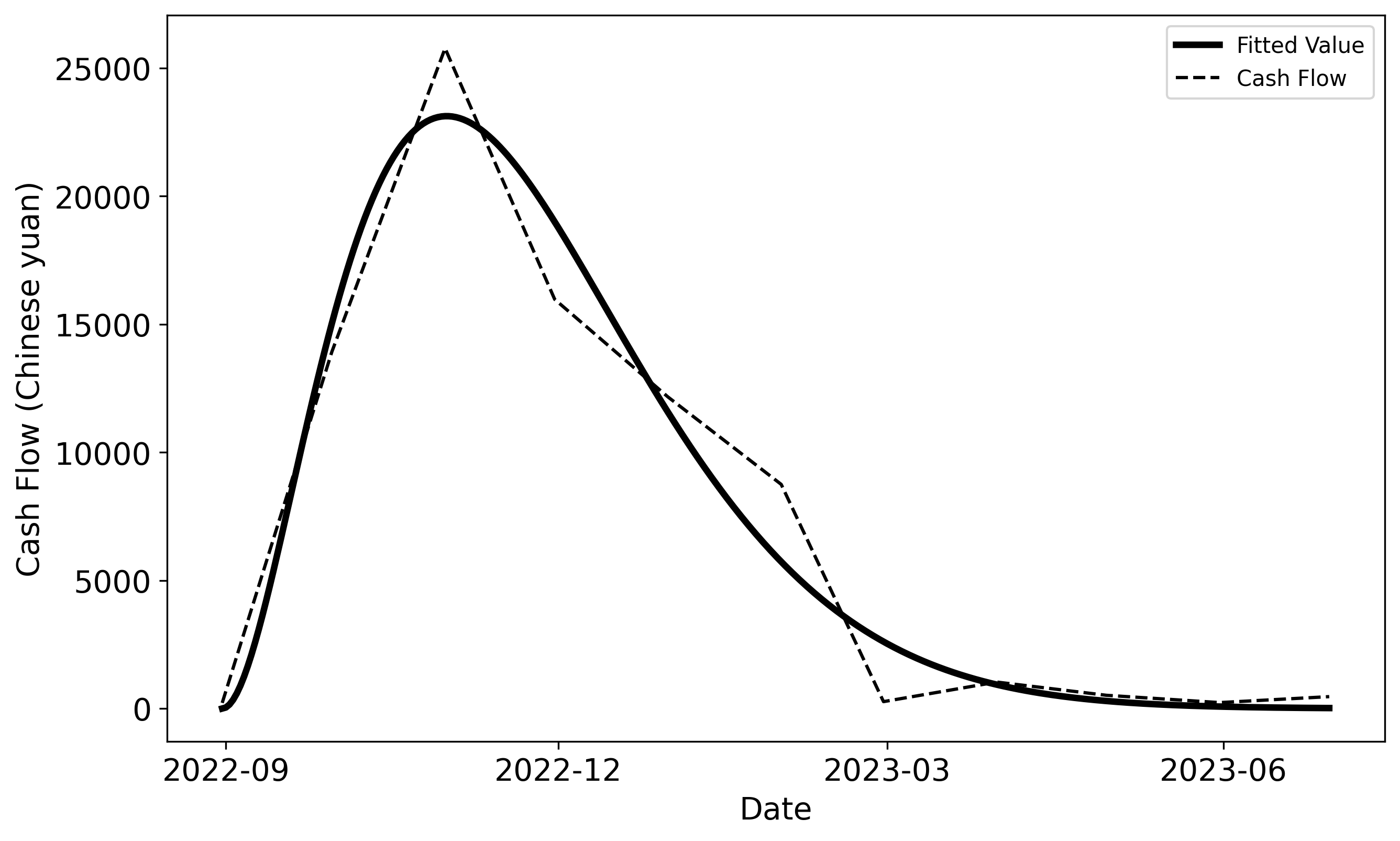}
		\raisebox{1.5ex}{Store C: Monthly}
		\label{fig:test9}
	\end{minipage}
	\caption{The figure illustrates the raw cash flow and fitted values of our model on three kinds of frequency: daily, weekly and monthly. The dashed line plots the raw cash flow, while the solid one presents the fitted values of our model. The sample period spans from January 2022 to September 2023.}
	\label{Fitted}
\end{figure}


Overall, our model well captures the trend of cash flow of China's stores over time in both qualitative and quantitative ways.

\section{Additional Discussions}
\label{Additional Discussions}
In this section, we will discuss insights, shortcomings, and further research directions. First and foremost, our model leads to some interesting implications, such as the diversely shaped cash flow curve with the diversity of consumer preferences, providing strong explanations to our real life. Then, we present some challenges faced in the empirical experiments. For instance, how we distinguish the store's style and the industry affects the estimation of potential consumers. Finally, we indicate the pros and cons of some basic settings in our model.

\subsection{Some interesting implications}
Our structural model provides micro-foundations for the data generation process of a store's cash flow, from which we derive some interesting implications. First, in the real economy, though most of the stores show a rainbow-shaped cash flow, some may present as an “M” shape. Our model implies that this is due to the diversity of consumer preferences. In the beginning, the style of the store may be outdated for some consumers but may be ahead of time for others, which means cash flow may experience increases and decreases. Second, Figure \ref{hh} shows that faster broadening of visibility not only shortens the ramp-up period but also increases the height of peak cash flow. This means that we should encourage the owner to increase monetary investment into opening promotion activities. Third, the decrease coefficient $\nu$ is the key parameter to determine the life cycle of a store. To avoid the shifting of consumer preference, the stores should frequently update their products to keep up with style. Stores that have higher budgets for product innovation will have an advantage in style updating and thus their life cycle will be longer. Similarly, policymakers should encourage the development of new product research departments. For example, a Chinese milk tea store, NAIXUE, listed in the Hong Kong exchange market updates its products per month. 

\subsection{Challenges in the empirical design}
Several challenges are encountered during the empirical testing of our model. First, in Section \ref{Empirical Evidence}, we use the category (or the sub-industry) to filter the competitors, implying that the stores in the same category own a similar style. However, this may result in some problems. For example, the stores of Miishien and Hangzhou soup dumpling stores belong to the same category of ``simple food", but it is not reliable that they are indistinguishable styles for consumers. Therefore, we may have overestimated the number of competitors for each store in our empirical research, thus underestimating the potential customers and overestimating the initial conversion rate. In empirical results, the initial conversion rate of store B from the restaurant industry is only slightly lower than that of store A from the retail industry may be because of the overestimation of store B's competitors and underestimation of the potential consumers. 

In practical scenarios, certain stores allocate significant financial resources to promotional endeavors during their opening period. This proactive approach serves to enhance awareness among the local population, thereby accelerating the initial growth. Within our framework, the parameter $k$, denoting the visibility broadening speed, governs the duration of this ramp-up period. A higher value for $k$ means a shorter ramp-up period for the store. In our empirical investigation, we observed that the parameter $k$ for store B yielded statistically insignificant results, implying a negligible ramp-up period for this particular store. Our model also accounts for these unusual cases.


Furthermore, the daily cash flow exhibits significant periodicity, particularly in the context of retail stores. For example, stores located within office buildings often experience reduced cash flow during weekends. This predictable nature presents another influential factor affecting cash flow. It unfortunately leads to a problem of heteroscedasticity and autocorrelation. Despite this challenge, our model effectively captures the primary trends in cash flow dynamics.



Finally, in real-life scenarios, when a store's cash flow dips below the break-even point, owners may not promptly opt for the closure of their store. Figure \ref{hh} depicting store C highlights the owner's persistent efforts to bolster cash flow, which ultimately proved unsuccessful. In practice, owners typically dedicate an extended period to enhancing cash flow, often through methods like discounted promotions. This period is usually three months in China, because of quarterly rent obligations. 

In conclusion, 
while our framework is comprehensive in evaluating key parameters and in validating the explanatory power, these challenges still require future consideration and resolution.

\subsection{Comments on the model settings}
There are also some settings of our theoretical model that are worth emphasizing and discussing in detail. 
First, in Proposition \ref{unique}, we show that even stores that open with the same initial monetary investment at the same time may choose different styles, which is consistent with real life. However, we assume a probability ordering rule to simplify to derive the Nash Equilibrium. Even if it is not assumed, we suspect that there should be an equilibrium in this style market, possibly a mixed equilibrium. Second, one might argue that our assumption of population structure does not change over time, but in fact, the assumption of preference shifting implies the explanation. What is more likely to happen in reality is that both population structure and consumers' preferences change, but we can always encapsulate common preference groups that are stable over time into more specific types (in the most specific case, each individual is a type). The only thing we need to focus on is how the preferences of the specific types change. Third, the preference shifting of consumers is depicted by the classic Cobb-Douglas utility function in Assumption \ref{uf}. Nevertheless, the theory of behavioral economics has become more and more popular and is shown to be more realistic; see \cite{BARBERIS20031053}. E.g., the cumulative prospect theory (\cite{tversky1992advances}) and related S-shaped utility functions with new variants (\cite{liang2020classification,liang2024unified}) could be explored to calibrate the customers' preferences. 
Lastly, the Nash Equilibrium derived in this paper relies on the existence of a type of dominant consumers with a high proportion. If this is not the case, the store's decisions with respect to consumers' behavior may be discontinuous. Some counterexamples prove that the equilibrium in Theorem \ref{t1} is not guaranteed under the existing conditions. The main reason for this is that it is hard to ensure that the optimal purchase probability density for each type of consumer is simultaneously normalized.

\section{Concluding remarks}
\label{Conclusion}

We develop a theoretical framework to characterize a store's daily cash flow and explore its data production mechanism. Within this model, we establish a Nash equilibrium to elucidate the impact of the shifting consumer preferences; this results in a gradual decrease trend of cash flow. Additionally, our theory highlights that competitive intensity among stores serves as another reason for unexpected store closures.
Using real data from three distinct industries in China (retail, restaurant, and service), we apply the Non-linear Least Squares regression (NLS) method to estimate initial conversion rates and decrease speeds, validating our model. Our empirical findings exhibit robust qualitative and quantitative performance, demonstrating a monthly cash flow decrease by a maximum of 5.01\%, corroborating survey insights from the China State Statistics Bureau.

The structural model we present offers insight into the cash flow generation processes at the micro level, providing explanations applicable to practical scenarios. Notably, while most stores exhibit a standard rainbow-shaped cash flow pattern, anomalies such as an ``M" shape can be accounted for by considering diverse consumer preferences in our model. Besides, we find that a faster visibility broadening speed not only accelerates ramp-up periods but also amplifies cash flow peaks (Figure \ref{hh}). This suggests that store owners should be encouraged to increase promotional activities during store openings.
Furthermore, the decreasing coefficient $\nu$ emerges as a pivotal parameter shaping a store's life cycle dynamics, wherein regular menu updates or periodic location changes can mitigate the shifting effect of consumer preference. China's bubble tea stores, such as NAIXUE listed in the Hong Kong exchange market, update their products per month. In this regard, policymakers may consider supporting innovative products to foster sustained market competitiveness. To conclude, this paper provides an answer to the 
open question of why the typical store usually has such a short life cycle. 


The originality of store modeling, along with the practical challenges inherent in empirical studies, presents promising avenues for future research. Numerous unresolved questions persist regarding specific configurations and data accessibility. For example, in our investigation, we employ a non-random scoring function, diverging from random events in real operations of a store, which could be a subject for subsequent exploration. Furthermore, given the focal point on preferences in our study, we presume the average consumer purchase (ACP) to be constant. To enhance accuracy, the ACP may be modeled by a deterministic function or a stochastic process. 
Finally, a comprehensive examination of rainbow-shaped curve structures, particularly in circumstances of censored data, remains on the precipice of interest for further study.

\quad 

\noindent
{\bf Acknowledgements.}
The authors thank Zhuo Chen, Zherui Fan, Chengxin Gong, Yangbo He, Zhiyi Song, Ningxin Zhang, Zhen Zhou, Elizabeth Zhu, and members of the research group on financial mathematics and risk management at CUHK-Shenzhen for their helpful discussion and important feedback. 
Y. Liu acknowledges financial support from the research startup fund at The Chinese University of Hong Kong, Shenzhen (Grant No. UDF01003336). 

\newpage
\appendix
\section{Details in definition}
\label{append1}
\subsection{Definition of purchase probability density.}
As long as $\text{dim}(\mathbb{Z}_t)>0,$ we can define the purchase probability density. Denote the Lebesgue measure (or Hausdorff measure) on $\mathbb{Z}_t$ as $\sigma_{\mathbb{Z}_t},$ The type $j$th consumers' purchase probability measure on $\mathbb{Z}_t$ is $\mathbb{P}_{jt},$ which satisfies that $\mathbb{P}_{jt}\ll \sigma_{\mathbb{Z}_t}.$ Then we define 
\begin{equation}
    \rho_{jt}(\boldsymbol{z}) = \tilde{\rho}_{jt}(\boldsymbol{z}) = \frac{\mathrm{d}\mathbb{P}_{jt}}{\mathrm{d}\sigma_{\mathbb{Z}_t}}
\end{equation}
and the integrals with respect to it is defined as \begin{equation}
    \int_{\mathbb{Z}_t} f(\boldsymbol{z},\rho_{jt}(\boldsymbol{z})) \mathrm{d}\boldsymbol{z} = \int_{\mathbb{Z}_t} f(\boldsymbol{z},\tilde{\rho}_{jt}(\boldsymbol{z})) \mathrm{d}\sigma_{\mathbb{Z}_t}.
\end{equation}
See more in \cite{cohn2013measure} and \cite{bertsekas2008introduction}.

\section{Proofs}
\label{append2}
\subsection{Missing proof of Theorem \ref{t1}}
According to Proposition \ref{pro7}, there exists $\tilde{B}_0(\cdot)$ such that on $\Psi_{\tilde{B}_0}$ we have $\psi_j$ are smooth w.r.t. $\boldsymbol{K}$ and
\begin{equation}
	\frac{\partial}{\partial K_j} \int_0^{t'}\psi_j(\boldsymbol{K},t,I)\mathrm{d}\tau > \sum_{2\leq i \leq m, i\neq j}\frac{\partial}{\partial K_i}\left|\int_0^{t'}\psi_j(\boldsymbol{K},t,I)\mathrm{d}\tau\right|>0
\end{equation}
for $j=2,\ldots,m$ and $t \in (t_{\epsilon},\infty)$ for some $t_{\epsilon}>0.$ It is evident that there exists $b_0>0$ and such that
\begin{equation}
	\sup_{\boldsymbol{K}\in \Psi_{\tilde{B}_0}|_{K_{1}\geq b_0 } }b_0\int_{\mathcal{I}}\int_0^{T(\boldsymbol{K},I)} \psi_1(\boldsymbol{K},t,I)\mathrm{d}\tau\mathrm{d}I\leq b_0\int_{\mathcal{I}}\int_0^{\infty} \sup_{\boldsymbol{K}\in \Psi_{\tilde{B}_0}|_{K_{1}\geq b_0 } } \psi_1(\boldsymbol{K},t,I)\mathrm{d}\tau\mathrm{d}I < 1.
\end{equation}
By Proposition \ref{pro8} when $\mathbf{P}^{(1)}$ is sufficiently large, we can simultaneously ensure that
\begin{align}
	&\inf_{\boldsymbol{K}\in \Psi_{\tilde{B}_0}|_{K_{1}\geq b_0 } }\tilde{B}_0(b_0)\int_{\mathcal{I}}\int_0^{T(\boldsymbol{K},I)} \psi_j(\boldsymbol{K},t,I)\mathrm{d}\tau\mathrm{d}I \\
	\geq &
	\tilde{B}_0(b_0)\int_{\mathcal{I}}\int_0^{\underset{{\boldsymbol{K}\in \Psi_{\tilde{B}_0}|_{K_{1}\geq b_0 } }}{\inf}T(\boldsymbol{K},I)} \underset{{\boldsymbol{K}\in \Psi_{\tilde{B}_0}|_{K_{1}\geq b_0 } }}{\inf}\psi_j(\boldsymbol{K},t,I)\mathrm{d}\tau\mathrm{d}I >1
\end{align}
for $j=2,\ldots,m$ and 
\begin{equation}
	\check{\Psi} = [b_0,\infty)\times\left[0,\tilde{B}_0(b_0)\right]^{m-1}\subseteq \Psi_{\tilde{B}_0}.
\end{equation}
From now on, we restrict $\boldsymbol{K}$ to $\check{\Psi}$ to explore the existence. 
Take $r_{m}>0$ such that when $0<r(I)\leq r_{m}(I\in\mathcal{I}),$
\begin{equation}
	\underset{{\boldsymbol{K}\in \Psi_{\tilde{B}_0}|_{K_{1}\geq b_0 } }}{\inf}T(\boldsymbol{K},I) > t_\epsilon,\ I\in\mathcal{I}.
\end{equation} Furthermore, there exists $W(\boldsymbol{K},I)>0$ such that $\psi_j(\boldsymbol{K},t,I)$ decreases to $0$ w.r.t. $t$ on $[W(\boldsymbol{K},I),\infty)$ for any $j=1,2,\ldots,m.$ So $\sum_{j=1}^m\mathbf{P}^{(j)}K_j\psi_j(\boldsymbol{K},t,I)$ also decreases to $0$ w.r.t. $t$ on $[W(\boldsymbol{K},I),\infty).$ Therefore, $T(\boldsymbol{K},I)$ is the unique zero of
\begin{equation}
	\theta({\boldsymbol{z}}_{0}^*(\boldsymbol{K},I))\sum_{j=1}^m \mathbf{P}^{(j)}K_j \cdot F_{j0}({\boldsymbol{z}}_{0}^*(\boldsymbol{K},I)-\tau{\boldsymbol{d}},\theta({\boldsymbol{z}}_{0}^*(\boldsymbol{K},I)))- r(I)
\end{equation} 
on $[W(\boldsymbol{K},I),\infty).$ By the implicit function theorem (see \cite{rudin1964principles}), $T(\boldsymbol{K},I)$ is smooth w.r.t. $\boldsymbol{K}.$ Then we propose the following lemma. Thus,
$
	v(\boldsymbol{K},I),\ V(\boldsymbol{K})
$
are both smooth w.r.t. $\boldsymbol{K}.$

\begin{lemma}
	\label{lemma1}
Restrict $\boldsymbol{K}$ to $\check{\Psi},$ 
there exists $\bar{r}>0$ such that as long as $0<r(I)\leq \bar{r}(I\in\mathcal{I})$ holds, for $n=m-1,m-2,\ldots,1,0,$ given $\bar{K}_1,\bar{K}_2,\ldots,\bar{K}_n,$ there exist $K^*_{n+1},K^*_{n+2},\ldots,K^*_{m}$ such that when $\boldsymbol{K} = (\bar{K}_1,\bar{K}_2,\ldots,\bar{K}_n,K^*_{n+1},K^*_{n+2},\ldots,K^*_{m}),$ 
\begin{equation}
	1 = \int_{\mathcal{I}} \int_{0}^{T(\boldsymbol{K},I)}K^*_j\psi_j(\boldsymbol{K},t,I)\mathrm{d}\tau\mathrm{d}I,\ j=n+1,\ldots,m.
\end{equation}
hold. Moreover, when $n\geq 1,$ $K^*_j$ is smooth w.r.t. $\bar{K}_i$ for $n+1\leq j\leq m, 1\leq i \leq n,$ and
\begin{equation}
	\left|\frac{\partial K^*_j}{\partial \bar{K_i}}\right|<1.
\end{equation}
for $n+1\leq j \leq m, 2\leq i \leq n.$
\end{lemma}
\begin{proof}[proof of Lemma \ref{lemma1}]
We use backward mathematical induction for $n$ to prove.

\textbf{Basic Case: } When $n=m-1,$ given $\bar{K}_1,\bar{K}_2,\ldots,\bar{K}_{m-1},$ note that
\begin{equation}
	K_m V_m(\bar{K}_1,\bar{K}_2,\ldots,\bar{K}_{m-1},K_m)\rightarrow 0\ (K_m\rightarrow 0 )
\end{equation}
and 
\begin{equation}
	K_m V_m(\bar{K}_1,\bar{K}_2,\ldots,\bar{K}_{m-1},K_m)> 1 \ (K_m\rightarrow \tilde{B}_0(b_0) ).
\end{equation}
By continuity, there exists $K^*_1 \in \left[0,\tilde{B}_0(b_0)\right]$ such that
\begin{equation}
	K^*_m V_m(\bar{K}_1,\bar{K}_2,\ldots,\bar{K}_{m-1},K^*_m)= 1.
\end{equation}
Moreover, for $i = 2,\ldots,m,$
\begin{align}
	\frac{\partial }{\partial K_i}(K_mV_m(\boldsymbol{K})) &= K_m\int_{\mathcal{I}}\psi_m(\boldsymbol{K},T,I)\mathrm{d}I\cdot\frac{\partial }{\partial K_i}T(\boldsymbol{K},I) + \int_{\mathcal{I}}\int_0^{T}\frac{\partial }{\partial K_i} (K_m\psi_m(\boldsymbol{K},t,I))\mathrm{d}t\mathrm{d}I\\
	&= \int_{\mathcal{I}}\int_0^{T}\frac{\partial }{\partial K_i} (K_m\psi_m(\boldsymbol{K},t,I))\mathrm{d}t\mathrm{d}I + o(1)\ \left(\sup_{\mathcal{I}}r(I)\rightarrow 0 \right).
\end{align}
Thus there exists $0<r_{m-1}\leq r_{m},$ when $0<r(I)\leq r_{m-1}(I\in\mathcal{I}),$
\begin{equation}
	\frac{\partial }{\partial K_m}(K_mV_m(\boldsymbol{K})) > \sum_{i=2}^{m-1}\left|\frac{\partial }{\partial K_i}(K_mV_m(\boldsymbol{K}))\right|>0,\ I\in\mathcal{I}.
\end{equation}
So we know that $K_m V_m(\bar{K}_1,\bar{K}_2,\ldots,\bar{K}_{m-1},K_m)$ increases w.r.t. $K_m,$ thus $K^*_m$ is unique. By the implicit function theorem (see \cite{rudin1964principles}), $K^*_m$ is smooth w.r.t. $\bar{K_i},i=1,2,\ldots,m-1$ and  
\begin{equation}
	\left|\frac{\partial K^*_m}{\partial \bar{K_i}}\right| = \frac{\left|\frac{\partial }{\partial K_i}(K_mV_m(\boldsymbol{K}))\right|}{\frac{\partial }{\partial K_m}(K_mV_m(\boldsymbol{K}))}<1,
\end{equation}
for $i = 2,\ldots,m-1.$

\textbf{Inductive Step: }Assume the statement is true for $n\ (2\leq n\leq m-1),$ and below we show that it is true for $n-1.$ Given $\bar{K}_{1},\ldots,\bar{K}_{n-1},$ and arbitrary $K_n,$ by the inductive hypothesis there exist $K^*_{n+1}(K_n),\ldots,K^*_{m}(K_n)$ such that for $j=n+1,\ldots,m,$
\begin{equation}
	K^*_j(K_n)V_j(\boldsymbol{K}(K_n))  =1,
\end{equation}
where $\boldsymbol{K}(K_n)=(\bar{K}_{1},\ldots,\bar{K}_{n-1},K_n,K_{n+1}^*(K_n),\ldots,K_{m}^*(K_n))$ is smooth w.r.t. $K_n.$ Note that 
\begin{equation}
	K_n V_n(\boldsymbol{K}(K_n))\rightarrow 0\ (K_n\rightarrow 0 )
\end{equation}
and 
\begin{equation}
	K_n V_n(\boldsymbol{K}(K_n))> 1 \ (K_n\rightarrow \tilde{B}_0(b_0) ).
\end{equation}
By continuity, there exists $K_n^*\in\left[0,\tilde{B}_0(b_0)\right]$ such that
\begin{equation}
	K^*_n V_n(\boldsymbol{K}(K^*_n))= 1.
\end{equation}
Furthermore, for $i = 2,\ldots,n-1,$
\begin{align}
	&\frac{\partial }{\partial K_i}(K_nV_n(\boldsymbol{K}(K_n))) \\
	=& K_n\int_{\mathcal{I}}\psi_n(\boldsymbol{K}(K_n),T,I)\mathrm{d}I\cdot\frac{\partial }{\partial K_i}T(\boldsymbol{K}(K_n),I) + \int_{\mathcal{I}}\int_0^{T}\frac{\partial }{\partial K_i} (K_m\psi_m(\boldsymbol{K}(K_n),t,I))\mathrm{d}t\mathrm{d}I\\
	=& \int_{\mathcal{I}}\int_0^{T}\frac{\partial }{\partial K_i} (K_n\psi_n(\boldsymbol{K}(K_n),t,I))\mathrm{d}t\mathrm{d}I + o(1)\ \left(\sup_{\mathcal{I}}r(I)\rightarrow 0 \right).
\end{align}
and 
\begin{align}
	&\frac{\partial }{\partial K_n}(K_nV_n(\boldsymbol{K}(K_n))) \\
	=&K_n\int_{\mathcal{I}}\psi_n(\boldsymbol{K}(K_n),T,I)\mathrm{d}I\cdot\left(\frac{\partial }{\partial K_n}T(\boldsymbol{K}(K_n),I)+\sum_{j=n+1}^{m}\frac{\partial K^*_j}{\partial K_n}\frac{\partial }{\partial K_j}T(\boldsymbol{K}(K_j),I) \right)\notag\\
	 &+\left(\int_{\mathcal{I}}\int_0^{T}\frac{\partial }{\partial K_n} (K_n\psi_n(\boldsymbol{K}(K_n),t,I))\mathrm{d}t\mathrm{d}I+\sum_{j=n+1}^{m}\int_{\mathcal{I}}\int_0^{T}\frac{\partial K^*_j}{\partial K_n}\frac{\partial }{\partial K_j} (K_n\psi_n(\boldsymbol{K}(K_n),t,I))\mathrm{d}t\mathrm{d}I\right)\\
	 =& \int_{\mathcal{I}}\int_0^{T}\frac{\partial }{\partial K_n} (K_n\psi_n(\boldsymbol{K}(K_n),t,I))\mathrm{d}t\mathrm{d}I+\sum_{j=n+1}^{m}\int_{\mathcal{I}}\int_0^{T}\frac{\partial K^*_j}{\partial K_n}\frac{\partial }{\partial K_j} (K_n\psi_n(\boldsymbol{K}(K_n),t,I))\mathrm{d}t\mathrm{d}I + o(1),
\end{align}
where
\begin{align}
	&\sum_{j=n+1}^{m}\left|\int_{\mathcal{I}}\int_0^{T}\frac{\partial K^*_j}{\partial K_n}\frac{\partial }{\partial K_j} (K_n\psi_n(\boldsymbol{K}(K_n),t,I))\mathrm{d}t\mathrm{d}I\right|\\
	\leq &\sum_{j=n+1}^{m}\left|\frac{\partial K^*_j}{\partial K_n}\right|\left|\int_{\mathcal{I}}\int_0^{T}\frac{\partial }{\partial K_j} (K_n\psi_n(\boldsymbol{K}(K_n),t,I))\mathrm{d}t\mathrm{d}I\right| \\
	\leq & \sum_{j=n+1}^{m}\left|\int_{\mathcal{I}}\int_0^{T}\frac{\partial }{\partial K_j} (K_n\psi_n(\boldsymbol{K}(K_n),t,I))\mathrm{d}t\mathrm{d}I\right|
.\end{align}
Thus there exists $0<r_{n-1}\leq r_{n},$ when $0<r(I)\leq r_{n}(I\in\mathcal{I}),$
\begin{equation}
	\frac{\partial }{\partial K_n}(K_nV_n(\boldsymbol{K}(K_n))) > \sum_{i=2}^{n-1}\left|\frac{\partial }{\partial K_i}(K_nV_n(\boldsymbol{K}(K_n)))\right|>0,\ I\in\mathcal{I}.
\end{equation}
So we know that $K_n V_n(\boldsymbol{K}(K_n))$ increases w.r.t. $K_n,$ thus $K^*_n$ is unique. By the implicit function theorem (see \cite{rudin1964principles}), $K^*_n$ is smooth w.r.t. $\bar{K_i},i=1,2,\ldots,n-1$ and  
\begin{equation}
	\left|\frac{\partial K^*_n}{\partial \bar{K_i}}\right| = \frac{\left|\frac{\partial }{\partial K_i}(K_nV_n(\boldsymbol{K}(K_n)))\right|}{\frac{\partial }{\partial K_n}(K_nV_n(\boldsymbol{K}(K_n)))}<1,
\end{equation}
for $i = 2,\ldots,n-1.$ Furthermore, $K_j^*=K_j^*(K_n^*(\bar{K}_i))$ is smooth w.r.t. $\bar{K_i}$ for $1\leq i \leq n-1,n+1\leq j\leq m$ and 
\begin{equation}
	\left|\frac{\partial}{\partial \bar{K_i}}K_j^*(K_n^*(\bar{K}_i))\right| = \left|\frac{\partial K_j^*}{\partial K_n}\right|\cdot \left| \frac{\partial K_n}{\partial \bar{K}_i}\right|<1,
\end{equation}
for $2\leq i \leq n-1,n+1\leq j\leq m.$

\textbf{Final Step: }When $n=0,$ for arbitrary $K_1,$ by the inductive steps, there exist $K^*_{2}(K_1),\ldots,K^*_{m}(K_1)$ such that for $j=2,\ldots,m,$
\begin{equation}
	K^*_j(K_1)V_j(\boldsymbol{K}(K_1))  =1,
\end{equation}
where $\boldsymbol{K}(K_1)=(K_1,K_{2}^*(K_1),\ldots,K_{m}^*(K_1))$ is smooth w.r.t. $K_1.$ Note that
\begin{equation}
	K^*_1(K_1)V_1(\boldsymbol{K}(K_1))  <1\ (K_1\rightarrow b_0),
\end{equation}
and 
\begin{equation}
	K^*_1(K_n)V_1(\boldsymbol{K}(K_1))  \rightarrow \infty\ (K_1\rightarrow \infty).
\end{equation}
By continuity, there exists $K_1^*\in\left[b_0,\infty\right)$ such that
\begin{equation}
	K^*_1 V_1(\boldsymbol{K}(K^*_1))= 1.
\end{equation}
Then we finish the proof of lemma.
\end{proof}
Using the case $n = 0$ in Lemma \ref{lemma1}, the proof is complete.
\subsection{Proof of Proposition \ref{pro32}}
  Let 
    \begin{equation}
        \rho_t({\boldsymbol{z}}|\theta) = \sum_{j=1}^m \mathbf{P}^{(j)}\rho_{jt}({\boldsymbol{z}}|\theta).
    \end{equation}
First we are going to show that given $\theta,$ $\rho_t({\boldsymbol{z}}|\theta)$ has a maximum point on $D=\left\{{\boldsymbol{z}}\in\mathbb{R}^{p+q}:G(\boldsymbol{z})\leq I\right\},$ which is a convex set. Noting that 
\begin{equation}
    \rho_t({\boldsymbol{z}}|\theta) \rightarrow 0\  (\boldsymbol{z}\rightarrow \infty\  \text{on}\  D),
\end{equation}
and
\begin{equation}
    \rho_t({\boldsymbol{z}}|\theta)>0,\ \boldsymbol{z}\in D,
\end{equation}
we know that $\rho_t({\boldsymbol{z}}|\theta)$ has a maximum point by continuity. Take the largest maximum point and denote it as ${\boldsymbol{z}}^*(\theta).$ Then we are going to show that $\theta \rho_t({\boldsymbol{z}}^*(\theta)|\theta)$ has a maximum point on $(0,\infty).$ Noting that
\begin{equation}
    \theta \rho_t({\boldsymbol{z}}^*(\theta)|\theta) \rightarrow 0 \ (\theta \rightarrow 0).
\end{equation}
Furthermore, we have
\begin{align}
    \theta\cdot\phi\left(C_j - \frac{\lambda_j}{2}\left\Vert\boldsymbol{z}^*(\theta)-\hat{\boldsymbol{z}}_{j,t} \right\Vert^2\right)\cdot e^{ -\gamma_j \theta}  &\leq \theta\cdot\phi\left(C_j \right) e^{  -\gamma_j \theta } \rightarrow 0\ (\theta \rightarrow \infty).
\end{align}
Thus, we have 
\begin{equation}
   \theta \cdot\rho_{jt}({\boldsymbol{z}}^*(\theta)|\theta) \rightarrow 0 \ (\theta \rightarrow \infty),
\end{equation}
which in turn implies that 
\begin{equation}
   \theta \cdot\rho_{t}({\boldsymbol{z}}^*(\theta)|\theta) = \sum_{j=1}^{m}\mathbf{P}^{(j)}\left(\theta \cdot\rho_{jt}({\boldsymbol{z}}^*(\theta)|\theta)\right) \rightarrow 0 \ (\theta \rightarrow \infty).
\end{equation}
By continuity, the maximum point exists.

\subsection{Proof of Proposition \ref{pro5}}
Consider 
    \begin{align}
\left({\boldsymbol{z}}_1,\theta\left({\boldsymbol{z}}_1\right)\right) &= \underset{\left(\boldsymbol{z},\theta\right):G(\boldsymbol{z}) \leq I}{\arg\widetilde{\max}}\quad\theta\mathbf{P}^{(1)}\rho_{1t}({\boldsymbol{z}}|\theta)\\
       &= \underset{\left(\boldsymbol{z},\theta\right):G(\boldsymbol{z}) \leq I}{\arg\widetilde{\max}}\quad K_{1t} \phi\left(C_1 - \frac{\lambda_1}{2}\left\Vert\boldsymbol{z}-\hat{\boldsymbol{z}}_{1,t} \right\Vert^2\right)\cdot e^{ -\gamma_1 \theta}
    \end{align}
    From Proposition \ref{p5}, we know that
    \begin{equation}
    	\left({\boldsymbol{z}}_1,\theta\left({\boldsymbol{z}}_1\right)\right) = \underset{\left(\boldsymbol{z},\theta\right):G(\boldsymbol{z}) \leq I}{\arg{\max}}\quad K_{1t} \phi\left(C_1 - \frac{\lambda_1}{2}\left\Vert\boldsymbol{z}-\hat{\boldsymbol{z}}_{1,t} \right\Vert^2\right)\cdot e^{ -\gamma_1 \theta},
    \end{equation}
    where the right-hand side is unique.
    Given $K_{1t}\in(0,\infty),$ when $(K_{2t},K_{3t},\ldots,K_{mt})\rightarrow \boldsymbol{0},$
    \begin{equation}
     \theta\sum_{j=1}^m \mathbf{P}^{(j)}\rho_{jt}({\boldsymbol{z}}|\theta) \rightarrow \theta\mathbf{P}^{(1)}\rho_{1t}({\boldsymbol{z}}|\theta),
    \end{equation}
    continuously and uniformly. Thus for any $\delta_0>0$
    \begin{equation}
        \left({\boldsymbol{z}}_t^{(s)}(I),\theta\left({\boldsymbol{z}}_t^{(s)}(I)\right)\right) \in U(\left({\boldsymbol{z}}_1,\theta\left({\boldsymbol{z}}_1\right)\right),\delta_0)
    \end{equation}
    when $K_{jt}(2\leq j\leq m)$ are sufficiently small, where $U$ represents the neighborhood. According to the Lagrange multiplier method (see \cite{rudin1964principles}), $\left({\boldsymbol{z}}_1,\theta\left({\boldsymbol{z}}_1\right)\right)$ and $\left({\boldsymbol{z}}_t^{(s)}(I),\theta\left({\boldsymbol{z}}_t^{(s)}(I)\right)\right)$ are given by
    \begin{equation}
    	\frac{\partial}{\partial( {\boldsymbol{z}},\theta,a) }L_1 =\frac{\partial}{\partial\left({\boldsymbol{z}},\theta,a\right)}\left(\theta\mathbf{P}^{(1)}\rho_{1t}({\boldsymbol{z}}|\theta)-a\left(G(\boldsymbol{z})-I\right)\right) =0,
    \end{equation}
    \begin{equation}
    	\frac{\partial}{\partial({\boldsymbol{z}},\theta,a)}L_m =\frac{\partial}{\partial\left({\boldsymbol{z}},\theta,a\right)} \left(\theta\sum_{j=1}^m\mathbf{P}^{(j)}\rho_{jt}({\boldsymbol{z}}|\theta)-a\left(G(\boldsymbol{z})-I\right)\right)=0,
    \end{equation}
     respectively.
     We calculate the Hessian matrix of $L_1$ with respect to $ ({\boldsymbol{z}},\theta)$ as 
     \begin{equation}
     	H_{L_1} ({\boldsymbol{z}},\theta) = \begin{pmatrix}
     		\frac{\partial^2 {L_1}}{\partial {\boldsymbol{z}}^2} & \frac{\partial^2 {L_1}}{\partial {\boldsymbol{z}} \partial \theta} \\
     		\frac{\partial^2 {L_1}}{\partial \theta \partial {\boldsymbol{z}}} & \frac{\partial^2 {L_1}}{\partial \theta^2} \\
     	\end{pmatrix}
     \end{equation}
     Therefore $H_{L_1}({\boldsymbol{z}}_1,\theta\left({\boldsymbol{z}}_1\right))$ is negative definite, so by continuity $H_{L_1}({\boldsymbol{z}},\theta)$ is negative definite on $U(({\boldsymbol{z}}_1,\theta({\boldsymbol{z}}_1)),\delta_0)$ for some $\delta_0>0.$ Moreover, we have
     \begin{equation}
     	H_{L_m}({\boldsymbol{z}},\theta) \rightarrow H_{L_1}({\boldsymbol{z}},\theta)
     \end{equation}
     on $U(({\boldsymbol{z}}_1,\theta({\boldsymbol{z}}_1)),\delta_0)$ when $(K_{2t},K_{3t},\ldots,K_{mt})\rightarrow \boldsymbol{0}.$ There exists $B_t(K_{1t})>0$ when $K_{jt}<B_t(K_{1t})$ for all $j=2,3,\ldots,m,$ $H_{L_m}({\boldsymbol{z}},\theta)$ is negative definite on $U(({\boldsymbol{z}}_1,\theta({\boldsymbol{z}}_1)),\delta_0).$ Easily, $B_t$ can be chosen as bounded w.r.t. $K_{1t}.$ As a result, the right-hand side of (\ref{p5}) is unique, which is equal to the left-hand side. By the implicit function theorem (see \cite{rudin1964principles}), it is smooth w.r.t. $(K_{1t},\ldots,K_{mt})$ and
     \begin{equation}
     \frac{\partial}{\partial K_{jt}}\left({\boldsymbol{z}}_t^{(s)}(I),\theta\left({\boldsymbol{z}}_t^{(s)}(I)\right)\right) \left.= H^{-1}_{L_m}\left(\boldsymbol{z},\theta\right)\cdot \mathbf{P}^{(j)}\frac{\partial}{\partial(\boldsymbol{z},\theta)}(\theta F_{jt}(\boldsymbol{z},\theta)) \right|_{\left({\boldsymbol{z}}_t^{(s)}(I),\theta\left({\boldsymbol{z}}_t^{(s)}(I)\right)\right)}.
     \end{equation}

     \subsection{Proof of Proposition \ref{pro7}}
     Without loss of generality, assume that $t=0$. Given $K_{10},$ let
	\begin{equation}
		\Psi_{B_0}|_{K_{10}} = \{K_{10}\}\times \left\{(K_{2t},\ldots,K_{mt}):,0<K_{jt}<{B}_t(K_{1t})\right\}.
	\end{equation} Note that 
	\begin{align}
		&\frac{\partial}{\partial K_{i0}}\int_t^{t'}K_{j0}F_{j\tau}\left({\boldsymbol{z}}_0^{(s)}(I),\theta\left({\boldsymbol{z}}_0^{(s)}(I)\right)\right) \mathrm{d}\tau  \\ =&\mathds{1}_{\{i=j\}}\cdot\int_0^{t'}F_{j\tau}\left({\boldsymbol{z}}_0^{(s)}(I),\theta\left({\boldsymbol{z}}_0^{(s)}(I)\right)\right) \mathrm{d}\tau +K_{j0}\int_0^{t'}\frac{\partial}{\partial K_{i0}}F_{j\tau}\left({\boldsymbol{z}}_0^{(s)}(I),\theta\left({\boldsymbol{z}}_0^{(s)}(I)\right)\right) \mathrm{d}\tau,
	\end{align}
	where
	\begin{align}
		&\int_0^{t'}\frac{\partial}{\partial K_{i0}}F_{j\tau}\left({\boldsymbol{z}}_0^{(s)}(I),\theta\left({\boldsymbol{z}}_0^{(s)}(I)\right)\right) \mathrm{d}\tau \\=& \left(\int_0^{t'}\left.\frac{\partial F_{j\tau}}{\partial  (\boldsymbol{z},\theta)}\right|_{\left({\boldsymbol{z}}_0^{(s)}(I),\theta\left({\boldsymbol{z}}_0^{(s)}(I)\right)\right)} \mathrm{d}\tau\right)'\cdot\frac{\partial }{\partial K_{i0}}\left({\boldsymbol{z}}_0^{(s)}(I),\theta\left({\boldsymbol{z}}_0^{(s)}(I)\right)\right)	\end{align}
	By Proposition \ref{pro6}, the above equation can be further simplified to
	\begin{equation}
		\label{equ67}
		\left(\int_0^{t'}\left.\frac{\partial F_{j0}}{\partial  (\boldsymbol{z},\theta)}\right|_{\left({\boldsymbol{z}}_0^{(s)}(I)-\tau\boldsymbol{d},\theta\left({\boldsymbol{z}}_0^{(s)}(I)\right)\right)} \mathrm{d}\tau\right)'\cdot\frac{\partial }{\partial K_{i0}}\left({\boldsymbol{z}}_0^{(s)}(I),\theta\left({\boldsymbol{z}}_0^{(s)}(I)\right)\right),
	\end{equation}
	where 
	\begin{align}
		&\int_0^{t'}\left.\frac{\partial F_{j0}}{\partial  (\boldsymbol{z},\theta)}\right|_{\left({\boldsymbol{z}}_0^{(s)}(I)-\tau\boldsymbol{d},\theta\left({\boldsymbol{z}}_0^{(s)}(I)\right)\right)} \mathrm{d}\tau \\=&
		\left. \left(\begin{array}{c}
			-\lambda_j e^{-\gamma_j\theta}\int_{0}^{t'}\phi'\left(C_j-\frac{\lambda_j}{2}\Vert{\boldsymbol{z}}-\tau\boldsymbol{d}-\hat{\boldsymbol{z}}_{j,0}\Vert^2\right)\left({\boldsymbol{z}}-\tau\boldsymbol{d}-\hat{\boldsymbol{z}}_{j,0}\right)\mathrm{d}\tau \\
			-\gamma_j e^{-\gamma_j \theta}\int_0^{t'}\phi\left(C_j-\frac{\lambda_j}{2}\Vert{\boldsymbol{z}}-\tau\boldsymbol{d}-\hat{\boldsymbol{z}}_{j,0}\Vert^2\right)\mathrm{d}\tau
		\end{array}\right)\right|_{\left({\boldsymbol{z}}_0^{(s)}(I),\theta\left({\boldsymbol{z}}_0^{(s)}(I)\right)\right)}\\
		\stackrel{\triangle}{=}&\left.\boldsymbol{\alpha}^{(j)}_1(t',\boldsymbol{z},\theta)\right|_{\left({\boldsymbol{z}}_0^{(s)}(I),\theta\left({\boldsymbol{z}}_0^{(s)}(I)\right)\right)},
	\end{align}
	and 
	\begin{align}
		&\frac{\partial}{\partial K_{i0}}\left({\boldsymbol{z}}_0^{(s)}(I),\theta\left({\boldsymbol{z}}_0^{(s)}(I)\right)\right) \\
		=&\left. H^{-1}_{L_m}\left(\boldsymbol{z},\theta\right)\cdot \mathbf{P}^{(i)}\frac{\partial}{\partial(\boldsymbol{z},\theta)}(\theta F_{i0}(\boldsymbol{z},\theta)) \right|_{\left({\boldsymbol{z}}_0^{(s)}(I),\theta\left({\boldsymbol{z}}_0^{(s)}(I)\right)\right)}\\
		=&H^{-1}_{L_m}\left(\boldsymbol{z},\theta\right)\cdot \mathbf{P}^{(i)}\left. \left(\begin{array}{c}
			-\lambda_i\theta e^{-\gamma_i\theta}\phi'\left(C_i-\frac{\lambda_i}{2}\Vert{\boldsymbol{z}}-\hat{\boldsymbol{z}}_{i,0}\Vert^2\right)\left({\boldsymbol{z}}-\hat{\boldsymbol{z}}_{i,0}\right) \\
			(1-\gamma_i\theta)  e^{-\gamma_i \theta}\phi\left(C_i-\frac{\lambda_i}{2}\Vert{\boldsymbol{z}}-\hat{\boldsymbol{z}}_{i,0}\Vert^2\right)
		\end{array}\right)\right|_{\left({\boldsymbol{z}}_0^{(s)}(I),\theta\left({\boldsymbol{z}}_0^{(s)}(I)\right)\right)}\\
		\stackrel{\triangle}{=}&H^{-1}_{L_m}\left(\boldsymbol{z},\theta\right)\cdot \mathbf{P}^{(i)}\left.\boldsymbol{\alpha}^{(i)}_2(\boldsymbol{z},\theta)\right|_{\left({\boldsymbol{z}}_0^{(s)}(I),\theta\left({\boldsymbol{z}}_0^{(s)}(I)\right)\right)}.
	\end{align}
	Therefore, for \eqref{equ67}, we have  
	\begin{align}
		 &\left|\mathbf{P}^{(i)}\left(\boldsymbol{\alpha}^{(j)}_1(t',\boldsymbol{z},\theta)\right)'H^{-1}_{L_m}\left(\boldsymbol{z},\theta\right) \left.\boldsymbol{\alpha}^{(i)}_2(\boldsymbol{z},\theta)\right|_{\left({\boldsymbol{z}}_0^{(s)}(I),\theta\left({\boldsymbol{z}}_0^{(s)}(I)\right)\right)}\right| \\
		 \label{eq76}
		 \leq &\left.\mathbf{P}^{(i)} \left\Vert\boldsymbol{\alpha}^{(j)}_1(t',\boldsymbol{z},\theta)\right\Vert \cdot \left|e_{\min}\left(H^{-1}_{L_m}\left(\boldsymbol{z},\theta\right)\right)\right|\cdot\left\Vert\boldsymbol{\alpha}^{(i)}_2(\boldsymbol{z},\theta) \right\Vert\right|_{\left({\boldsymbol{z}}_0^{(s)}(I),\theta\left({\boldsymbol{z}}_0^{(s)}(I)\right)\right)},
	\end{align}
	where $e_{\min}\left(H^{-1}_{L_m}\left(\boldsymbol{z},\theta\right)\right)$ represents the smallest eigenvalue of $H^{-1}_{L_m}\left(\boldsymbol{z},\theta\right)$; see \cite{strang2012linear}. It can be seen that
	\begin{equation}
		\left\Vert\boldsymbol{\alpha}^{(j)}_1\left(t',{\boldsymbol{z}}_0^{(s)}(I),\theta\left({\boldsymbol{z}}_0^{(s)}(I)\right)\right)\right\Vert,\ \left\Vert\boldsymbol{\alpha}^{(i)}_2\left({\boldsymbol{z}}_0^{(s)}(I),\theta\left({\boldsymbol{z}}_0^{(s)}(I)\right)\right) \right\Vert
	\end{equation}
	are both bounded w.r.t. $(K_{10},\ldots,K_{m0})$ on $\Psi_{B_0}|_{K_{10}}$ and $t'\in(t_{\epsilon},\infty)$ because \begin{equation}
		\left({\boldsymbol{z}}_0^{(s)}(I),\theta\left({\boldsymbol{z}}_0^{(s)}(I)\right)\right) \in U(\left({\boldsymbol{z}}_1,\theta\left({\boldsymbol{z}}_1\right)\right),\delta_0)
		\ 	\text{on} \ \Psi_{B_0}.
	\end{equation}
	Moreover, when $(K_{20},\ldots,K_{m0})\rightarrow \boldsymbol{0},$
	\begin{align}
		H_{L_m}\left({\boldsymbol{z}}_0^{(s)}(I),\theta\left({\boldsymbol{z}}_0^{(s)}(I)\right)\right) &\rightarrow H_{L_1}\left({\boldsymbol{z}}_1,\theta\left({\boldsymbol{z}}_1\right)\right)
	\end{align}
	which is negative definite and the largest eigenvalue of which has a non-zero upper bound for any $(K_{10},\ldots,K_{m0})$ in $\Psi_{B_0}|_{K_{10}}.$
	Therefore, (\ref{eq76}) is bounded w.r.t. $(K_{10},\ldots,K_{m0})$ in $\Psi_{B_0}|_{K_{10}}$ and $t'\in(t_{\epsilon},\infty)$.
	Given $K_{10}>0,$ 
	\begin{align}
		\mathbf{P}^{(j)}\frac{\partial}{\partial K_{i0}}\int_t^{t'}K_{j0}F_{j\tau}\left({\boldsymbol{z}}_0^{(s)}(I),\theta\left({\boldsymbol{z}}_0^{(s)}(I)\right)\right) \mathrm{d}\tau   = \mathds{1}_{\{i=j\}}\cdot\int_0^{t'}F_{j\tau}\left({\boldsymbol{z}}_0^{(s)}(I),\theta\left({\boldsymbol{z}}_0^{(s)}(I)\right)\right) \mathrm{d}\tau  +o(1)
	\end{align}
	for $2\leq j\leq m$ and $2 \leq i \leq m$ when $(K_{20},\ldots,K_{m0})\rightarrow 0.$ As 
	\begin{equation}
		\int_0^{t'}F_{j\tau}\left({\boldsymbol{z}}_0^{(s)}(I),\theta\left({\boldsymbol{z}}_0^{(s)}(I)\right)\right)\mathrm{d}\tau
	\end{equation}
	has non-zero lower bound on $\Psi_{B_0}|_{K_{10}},$ which is uniform w.r.t. $t' \in (t_{\epsilon},\infty).$ The existence of $\tilde{B}_0$ is clearly.
 \subsection{Proof of Proposition \ref{pro8}}

 Still, we assume that $t=0$ without loss of generality. Given $K_{10},$ we show first that this is true for $B_t(\cdot)$ in Proposition \ref{pro5}. From the proof in Proposition \ref{pro5}, we know that
	\begin{align}
		\theta\sum_{j=1}^m\mathbf{P}^{(j)}\rho_{j0}(\boldsymbol{z}|\theta) =\theta\sum_{j=1}^m \mathbf{P}^{(j)}K_{j0}F_{j0}(\boldsymbol{z},\theta),\\
		L_m = \theta\sum_{j=1}^m \mathbf{P}^{(j)}K_{j0}F_{j0}(\boldsymbol{z},\theta)-a\left(G(\boldsymbol{z}-I)\right).
	\end{align}
	Thus it only requires $\left(\mathbf{P}^{(2)}K_{20},\ldots,\mathbf{P}^{(m)}K_{m0}\right)\rightarrow \boldsymbol{0}$ to make sure that 
	\begin{align}
		\label{eq86}
		\left({\boldsymbol{z}}_0^{(s)}(I),\theta\left({\boldsymbol{z}}_0^{(s)}(I)\right)\right) \rightarrow \left({\boldsymbol{z}}_1,\theta\left({\boldsymbol{z}}_1\right)\right),
	\end{align}
	which is constant vector w.r.t. $K_{10}\in[b,\infty),$ and
	\begin{align}
		\label{eq87}
		H_{L_m}({\boldsymbol{z}},\theta) \rightarrow H_{L_1}({\boldsymbol{z}},\theta).
		\end{align}
		which is negative definite on $U(({\boldsymbol{z}}_1,\theta({\boldsymbol{z}}_1)),\delta_0)$ for any $K_{10}\in[b,\infty).$
	Note that the convergence in (\ref{eq86}) and (\ref{eq87}) is uniform w.r.t. $K_{10}\in [b,\infty),$ so there exists increasing $B'_0(b)$ such that when 
	$
		\mathbf{P}^{(j)}K_{j0}\leq B'_0(b)
	$ i.e. $K_{j0}\leq \frac{B'_0(b)}{\mathbf{P}^{(j)}},$
	the conclusion of Proposition \ref{pro5} holds for any $K_{10}\geq b$. Let \begin{equation}
		B_0(K_{10}) = \inf_{j\geq 2}\frac{B'_0(K_{10})}{\mathbf{P}^{(j)}},
	\end{equation}
	then 
	\begin{equation}
		\inf_{K_{1t}\geq b }{B}_t(K_{1t}) =\inf_{K_{1t}\geq b,j\geq 2}\frac{B'_0(K_{10})}{\mathbf{P}^{(j)}} \geq \inf_{j\geq 2}\frac{B'_0(b)}{\mathbf{P}^{(j)}} \geq M_0,
	\end{equation}
	when $\mathbf{P}^{(1)}$ is sufficiently large. 
	Then let
	\begin{equation}
		\Psi_{B_0}|_{K_{10}\geq b } = [b,\infty)\times \left\{(K_{2t},\ldots,K_{mt}):,0<K_{jt}<{B}_t(K_{1t})\right\}.
	\end{equation}
	 From the proof of Proposition \ref{pro7}, we know that
	\begin{align}
			&\frac{\partial}{\partial K_{i0}}\int_t^{t'}K_{j0}F_{j\tau}\left({\boldsymbol{z}}_0^{(s)}(I),\theta\left({\boldsymbol{z}}_0^{(s)}(I)\right)\right) \mathrm{d}\tau  \\ =&\mathds{1}_{\{i=j\}}\cdot\int_0^{t'}F_{j\tau}\left({\boldsymbol{z}}_0^{(s)}(I),\theta\left({\boldsymbol{z}}_0^{(s)}(I)\right)\right) \mathrm{d}\tau +K_{j0}
				\mathbf{P}^{(i)}\left(\boldsymbol{\alpha}^{(j)}_1(t',\boldsymbol{z},\theta)\right)'H^{-1}_{L_m}\left(\boldsymbol{z},\theta\right) \left.\boldsymbol{\alpha}^{(i)}_2(\boldsymbol{z},\theta)\right|_{\left({\boldsymbol{z}}_0^{(s)}(I),\theta\left({\boldsymbol{z}}_0^{(s)}(I)\right)\right)},
	\end{align}
	where 
	\begin{align}
		&\left|K_{j0}
		\mathbf{P}^{(i)}\left(\boldsymbol{\alpha}^{(j)}_1(t',\boldsymbol{z},\theta)\right)'H^{-1}_{L_m}\left(\boldsymbol{z},\theta\right) \left.\boldsymbol{\alpha}^{(i)}_2(\boldsymbol{z},\theta)\right|_{\left({\boldsymbol{z}}_0^{(s)}(I),\theta\left({\boldsymbol{z}}_0^{(s)}(I)\right)\right)}\right| \\
		\leq &	\left.K_{j0}\mathbf{P}^{(i)} \left\Vert\boldsymbol{\alpha}^{(j)}_1(t',\boldsymbol{z},\theta)\right\Vert \cdot \left|e_{\min}\left(H^{-1}_{L_m}\left(\boldsymbol{z},\theta\right)\right)\right|\cdot\left\Vert\boldsymbol{\alpha}^{(i)}_2(\boldsymbol{z},\theta) \right\Vert\right|_{\left({\boldsymbol{z}}_0^{(s)}(I),\theta\left({\boldsymbol{z}}_0^{(s)}(I)\right)\right)}.
		\end{align}
	For the same reason as Proposition \ref{pro7}, 
	\begin{equation}
		\left\Vert\boldsymbol{\alpha}^{(j)}_1\left(t',{\boldsymbol{z}}_0^{(s)}(I),\theta\left({\boldsymbol{z}}_0^{(s)}(I)\right)\right)\right\Vert,\ \left\Vert\boldsymbol{\alpha}^{(i)}_2\left({\boldsymbol{z}}_0^{(s)}(I),\theta\left({\boldsymbol{z}}_0^{(s)}(I)\right)\right) \right\Vert
	\end{equation}
	are both bounded w.r.t. $ \left(\mathbf{P}^{(1)},\ldots,\mathbf{P}^{(m)}\right)\in\left\{\sum_{j=1}^m\mathbf{P}^{(j)}=1\right\}, (K_{10},\ldots,K_{m0})\in\Psi_{B_0}|_{K_{10}\geq b}$ and $t'\in(t_{\epsilon},\infty).$ Furthermore, when $(\mathbf{P}^{(2)}K_{20},\ldots
	,\mathbf{P}^{(m)}K_{m0})\rightarrow \boldsymbol{0},$
	\begin{align}
		H_{L_m}\left({\boldsymbol{z}}_0^{(s)}(I),\theta\left({\boldsymbol{z}}_0^{(s)}(I)\right)\right) &\rightarrow H_{L_1}\left({\boldsymbol{z}}_1,\theta\left({\boldsymbol{z}}_1\right)\right)
	\end{align}
	which is negative definite. When $\mathbf{P}^{(1)}$ is away from $0$, say $\mathbf{P}^{(1)}\geq\frac{1}{m},$ the largest eigenvalue of  $H_{L_1}\left({\boldsymbol{z}}_1,\theta\left({\boldsymbol{z}}_1\right)\right)$ has a non-zero upper bound for any $\left(\mathbf{P}^{(1)},\ldots,\mathbf{P}^{(m)}\right)\in\left\{\sum_{j=1}^m\mathbf{P}^{(j)}=1,\mathbf{P}^{(1)}\geq\frac{1}{m}\right\}$ and $(K_{10},\ldots,K_{m0})\in\Psi_{B_0}|_{K_{10}\geq b}.$ Thus, it only requires $\mathbf{P}^{(i)}K_{j0}\rightarrow0\ (i,j=2,\ldots,m)$ for (\ref{eq62}) to hold true. There exists increasing $\tilde{B}'_0(b)$ such that when $\mathbf{P}^{(i)}K_{j0} <\tilde{B}'_0(b)\ (i,j=2,\ldots,m),$ (\ref{eq62}) holds for any $K_{10}\geq b.$ Let \begin{equation}
		\tilde{B}_0(K_{10}) = \inf_{j\geq 2}\frac{\tilde{B}'_0(K_{10})}{\mathbf{P}^{(j)}},
	\end{equation}
	then 
	\begin{equation}
		\inf_{K_{1t}\geq b }\tilde{B}_t(K_{1t}) =\inf_{K_{1t}\geq b,j\geq 2}\frac{\tilde{B}'_0(K_{10})}{\mathbf{P}^{(j)}} \geq \inf_{j\geq 2}\frac{\tilde{B}'_0(b)}{\mathbf{P}^{(j)}} \geq M_0,
	\end{equation}
	when $\mathbf{P}^{(1)}$ is sufficiently large.
\subsection{Proof of Proposition \ref{pro9}}
Let $\mu_j'= \langle{{\boldsymbol{z}}^*_0(I)-\hat{\boldsymbol{z}}_{j0},\boldsymbol{d}}\rangle$ and $\nu'=\frac{1}{2}\Vert{\boldsymbol{d}}\Vert^2.$ We have 
	\begin{equation}
		\tilde{\beta}_{t}(I) = \sum_{j=1}^{m}\mathbf{P}^{(j)} \tilde{\beta}_{j0}(I)\exp\left(\lambda_j( \mu_j' t- \nu' t^2)\right).
	\end{equation}
	When $\sup_{1\leq i,j\leq m }\left\Vert \frac{\boldsymbol{a}_i}{\lambda_i}-\frac{\boldsymbol{a}_j}{\lambda_j}\right\Vert\rightarrow 0,$ naturally we have 
	\begin{equation}
		\sup_{1\leq i,j\leq m } \Vert{\hat{\boldsymbol{x}}_{i0}-\hat{\boldsymbol{x}}_{j0}}\Vert\rightarrow 0.
	\end{equation}
	According to the definition of ${\boldsymbol{z}}^*_0(I)$ and the monotonic descent of $F_{j0}(\cdot)$ with respect to $\Vert\hat{\boldsymbol{x}}-\hat{\boldsymbol{x}}_{j0}\Vert,$ denoting $\boldsymbol{z}^*_0(I)=(\boldsymbol{x}^*_0(I)',\boldsymbol{\xi}^*_0(I)')',$ we have 
	\begin{equation}
		\left\Vert{\boldsymbol{x}}^*_0(I)-\hat{\boldsymbol{x}}_{j0}\right\Vert \leq \sup_{1\leq i,j \leq m } \left\Vert\hat{\boldsymbol{x}}_{i0}-\hat{\boldsymbol{x}}_{j0}\right\Vert, \quad j =1,2,\ldots,m.
	\end{equation}
	Since $\boldsymbol{d}=(\boldsymbol{c}',\boldsymbol{0}')',$ there is
	\begin{equation}
		|\mu_j|=|\langle{{\boldsymbol{z}}^*_0(I)-\hat{\boldsymbol{z}}_{j0},\boldsymbol{d}}\rangle| \leq \left\Vert{\boldsymbol{x}}^*_0(I)-\hat{\boldsymbol{x}}_{j0}\right\Vert\cdot\Vert\boldsymbol{c}\Vert \rightarrow 0, \quad j =1,2,\ldots,m.
	\end{equation}
	The first derivative of $\tilde{\beta}_t(I)$ with respect to $t$ is
	\begin{equation}
		\frac{\partial \tilde{\beta}_t(I)}{\partial t} = \sum_{j=1}^{m}\mathbf{P}^{(j)} \tilde{\beta}_{j0}(I)\exp\left(\lambda_j( \mu_j' t- \nu' t^2)\right)\cdot\lambda_j(\mu_j'-2\nu' t),
	\end{equation}
	and the second derivative is
	\begin{equation}
		\frac{\partial^2 \tilde{\beta}_t(I)}{\partial t^2} = \sum_{j=1}^{m}\mathbf{P}^{(j)} \tilde{\beta}_{j0}(I)\exp\left(\lambda_j( \mu_j' t- \nu' t^2)\right)\cdot\lambda_j\left(-2\nu'+\lambda_j(\mu_j'-2\nu' t)^2\right).
	\end{equation}
	Take $\delta_0 >0.$ There exists $\epsilon>0$ s.t. when $\sup_{1\leq i,j\leq m }\left\Vert \frac{\boldsymbol{a}_i}{\lambda_i}-\frac{\boldsymbol{a}_j}{\lambda_j}\right\Vert<\epsilon,$
	\begin{equation}
		|\mu_j| < \min\left\{2\nu'\delta_0,\sqrt{\frac{2\nu'}{M}},\left|2\nu'\delta_0-\sqrt{\frac{2\nu'}{M}}\right|\right\}.
	\end{equation}
	In this case for $j=1,2,\ldots,m,$ we have 
	\begin{equation}
		\mu_j'-2\nu' t \leq \mu_j'-2\nu' \delta_0 <0,\quad t\in[\delta_0,\infty),
	\end{equation}
	and 
	\begin{equation}
		-2\nu'+\lambda_j(\mu_j'-2\nu' t)^2 \leq \max\{-2\nu'+M\mu_j'^2,-2\nu'+M(\mu_j'-2\nu' \delta_0)^2\}<0,\quad t\in[0,\delta_0],
	\end{equation}
	which is exactly
	\begin{equation}
		\frac{\partial \tilde{\beta}_t(I)}{\partial t}<0,\quad t\in[\delta_0,\infty),
	\end{equation}
	and
	\begin{equation}
		\frac{\partial^2 \tilde{\beta}_t(I)}{\partial t^2}<0,\quad t\in[0,\delta_0].
	\end{equation}
	So $\tilde{\beta}_t(I)$ is monotonic on $[\delta_0,\infty)$ and has at most one maximum point on $[0,\delta_0].$
 
\newpage
\bibliography{mybibfile}

\begin{thebibliography}{31}
\providecommand{\natexlab}[1]{#1}
\providecommand{\url}[1]{\texttt{#1}}
\expandafter\ifx\csname urlstyle\endcsname\relax
  \providecommand{\doi}[1]{doi: #1}\else
  \providecommand{\doi}{doi: \begingroup \urlstyle{rm}\Url}\fi

\bibitem[Barberis and Thaler(2003)]{BARBERIS20031053}
Nicholas Barberis and Richard Thaler.
\newblock Chapter 18 $\text{A}$ survey of behavioral finance.
\newblock In \emph{Financial Markets and Asset Pricing}, volume~1 of
  \emph{Handbook of the Economics of Finance}, pages 1053--1128. Elsevier,
  2003.

\bibitem[Berry and Haile(2014)]{berry2014identification}
Steven~T Berry and Philip~A Haile.
\newblock Identification in differentiated products markets using market level
  data.
\newblock \emph{Econometrica}, 82\penalty0 (5):\penalty0 1749--1797, 2014.

\bibitem[Berry et~al.(1995)Berry, Levinsohn, and Pakes]{berry1995automobile}
Steven~T Berry, Jame Levinsohn, and Ariel Pakes.
\newblock Automobile prices in market equilibrium.
\newblock \emph{Econometrica}, 63:\penalty0 841--890, 1995.

\bibitem[Bertsekas and Tsitsiklis(2008)]{bertsekas2008introduction}
Dimitri Bertsekas and John~N Tsitsiklis.
\newblock \emph{Introduction to probability}, volume~1.
\newblock Athena Scientific, 2008.

\bibitem[Bils and Klenow(2004)]{bils2004some}
Mark Bils and Peter~J Klenow.
\newblock Some evidence on the importance of sticky prices.
\newblock \emph{Journal of Political Economy}, 112\penalty0 (5):\penalty0
  947--985, 2004.

\bibitem[Boyd and Vandenberghe(2004)]{boyd2004convex}
Stephen~P Boyd and Lieven Vandenberghe.
\newblock \emph{Convex optimization}.
\newblock Cambridge university press, 2004.

\bibitem[Burdett and Judd(1993)]{burdett1993equilibrium}
Kenneth Burdett and Kenneth~L Judd.
\newblock Equilibrium price dispersion.
\newblock \emph{Econometrica}, 51\penalty0 (4):\penalty0 955--969, 1993.

\bibitem[Chamberlin(1933)]{chamberlin1938theory}
Edward Chamberlin.
\newblock \emph{Theory of Monopolistic Competition}.
\newblock Harvard University Press, 1933.

\bibitem[Cobb and Douglas(1928)]{cobb1928theory}
Charles~W Cobb and Paul~H Douglas.
\newblock A theory of production.
\newblock \emph{American Economic Review}, 18:\penalty0 139--165, 1928.

\bibitem[Cohn(2013)]{cohn2013measure}
Donald~L Cohn.
\newblock \emph{Measure theory}, volume~5.
\newblock Springer, 2013.

\bibitem[Dixit and Stiglitz(1977)]{dixit1977monopolistic}
Avinash~K Dixit and Joseph~E Stiglitz.
\newblock Monopolistic competition and optimum product diversity.
\newblock \emph{American Economic Review}, 67\penalty0 (3):\penalty0 297--308,
  1977.

\bibitem[Folland(1999)]{folland1999real}
Gerald~B Folland.
\newblock \emph{Real Analysis: Modern Techniques and Their Applications}.
\newblock John Wiley \& Sons, 1999.

\bibitem[Hansen(1982)]{hansen1982large}
Lars~P Hansen.
\newblock Large sample properties of generalized method of moments estimators.
\newblock \emph{Econometrica}, 50\penalty0 (4):\penalty0 1029--1054, 1982.

\bibitem[Kashyap(1995)]{kashyap1995sticky}
Anil~K Kashyap.
\newblock Sticky prices: New evidence from retail catalogs.
\newblock \emph{Quarterly Journal of Economics}, 110\penalty0 (1):\penalty0
  245--274, 1995.

\bibitem[Lam et~al.(2001)Lam, Vandenbosch, Hulland, and
  Pearce]{lam2001evaluating}
Shun~Yin Lam, Mark Vandenbosch, John Hulland, and Michael Pearce.
\newblock Evaluating promotions in shopping environments: Decomposing sales
  response into attraction, conversion, and spending effects.
\newblock \emph{Marketing Science}, 20\penalty0 (2):\penalty0 194--215, 2001.

\bibitem[Liang and Liu(2020)]{liang2020classification}
Zongxia Liang and Yang Liu.
\newblock A classification approach to general $\text{S}$-shaped utility
  optimization with principals' constraints.
\newblock \emph{SIAM Journal on Control and Optimization}, 58\penalty0
  (6):\penalty0 3734--3762, 2020.

\bibitem[Liang et~al.(2024)Liang, Liu, Ma, and Vinoth]{liang2024unified}
Zongxia Liang, Yang Liu, Ming Ma, and Rahul~P Vinoth.
\newblock A unified formula of the optimal portfolio for piecewise hyperbolic
  absolute risk aversion utilities.
\newblock \emph{Quantitative Finance}, 24\penalty0 (2):\penalty0 281--303,
  2024.

\bibitem[Lucas and Rossi-Hansberg(2002)]{lucas2002internal}
Robert~E Lucas and Esteban Rossi-Hansberg.
\newblock On the internal structure of cities.
\newblock \emph{Econometrica}, 70\penalty0 (4):\penalty0 1445--1476, 2002.

\bibitem[Menzio(2023)]{menzio2023optimal}
Guido Menzio.
\newblock Optimal product design: Implications for competition and growth under
  declining search frictions.
\newblock \emph{Econometrica}, 91\penalty0 (2):\penalty0 605--639, 2023.

\bibitem[Metropolis et~al.(1953)Metropolis, Rosenbluth, Rosenbluth, Teller, and
  Teller]{metropolis1953equation}
Nicholas Metropolis, Arianna~W Rosenbluth, Marshall~N Rosenbluth, Augusta~H
  Teller, and Edward Teller.
\newblock Equation of state calculations by fast computing machines.
\newblock \emph{The journal of chemical physics}, 21\penalty0 (6):\penalty0
  1087--1092, 1953.

\bibitem[Nash(1951)]{nash1951non}
John Nash.
\newblock Non-cooperative games.
\newblock \emph{Annals of mathematics}, pages 286--295, 1951.

\bibitem[Newey and West(1986)]{newey1986simple}
Whitney~K Newey and Kenneth~D West.
\newblock A simple, positive semi-definite, heteroskedasticity and
  autocorrelationconsistent covariance matrix.
\newblock 1986.

\bibitem[Perdikaki et~al.(2012)Perdikaki, Kesavan, and
  Swaminathan]{perdikaki2012effect}
Olga Perdikaki, Saravanan Kesavan, and Jayashankar~M Swaminathan.
\newblock Effect of traffic on sales and conversion rates of retail stores.
\newblock \emph{Manufacturing \& Service Operations Management}, 14\penalty0
  (1):\penalty0 145--162, 2012.

\bibitem[Romer(1986)]{romer1986increasing}
Paul~M Romer.
\newblock Increasing returns and long-run growth.
\newblock \emph{Journal of Political Economy}, 94\penalty0 (5):\penalty0
  1002--1037, 1986.

\bibitem[Ruckstuhl(2010)]{ruckstuhl2010introduction}
Andreas Ruckstuhl.
\newblock Introduction to nonlinear regression.
\newblock \emph{IDP Institut fur Datenanalyse und Prozessdesign, Zurcher
  Hochschule fur Angewandte Wissenschaften}, page 365, 2010.

\bibitem[Rudin(1964)]{rudin1964principles}
Walter Rudin.
\newblock \emph{Principles of Mathematical Analysis}.
\newblock McGraw-hill New York, 3 edition, 1964.

\bibitem[Simonson and Tversky(1992)]{simonson1992choice}
Itamar Simonson and Amos Tversky.
\newblock Choice in context: Tradeoff contrast and extremeness aversion.
\newblock \emph{Journal of Marketing Research}, 29\penalty0 (3):\penalty0
  281--295, 1992.

\bibitem[Singh et~al.(2006)Singh, Hansen, and Blattberg]{singh2006market}
Vishal~P Singh, Karsten~T Hansen, and Robert~C Blattberg.
\newblock Market entry and consumer behavior: An investigation of a wal-mart
  supercenter.
\newblock \emph{Marketing Science}, 25\penalty0 (5):\penalty0 457--476, 2006.

\bibitem[Strang(2012)]{strang2012linear}
Gilbert Strang.
\newblock \emph{Linear algebra and its applications}.
\newblock 2012.

\bibitem[Tversky and Kahneman(1992)]{tversky1992advances}
Amos Tversky and Daniel Kahneman.
\newblock Advances in prospect theory: cumulative representation of
  uncertainty.
\newblock \emph{Journal of Risk and Uncertainty}, 5:\penalty0 297--323, 1992.

\bibitem[Zimmermann and Auinger(2023)]{zimmermann2023developing}
Robert Zimmermann and Andreas Auinger.
\newblock Developing a conversion rate optimization framework for digital
  retailers--case study.
\newblock \emph{Journal of Marketing Analytics}, 11:\penalty0 233--243, 2023.

\end{thebibliography}
\end{document}